\documentclass[11pt]{article}

\usepackage{mathtools,amsmath,amsfonts,amssymb,amsthm,mathtools}
\usepackage{fullpage,hyperref,xspace}
\usepackage[footnotesize]{caption}
\usepackage[caption=false,font=footnotesize]{subfig}
\usepackage{graphicx}
\usepackage{xcolor}
\graphicspath{{./final_figs/}}

\newtheorem{theorem}{Theorem}[section]
\newtheorem{corollary}[theorem]{Corollary}
\newtheorem{lemma}[theorem]{Lemma}
\newtheorem{proposition}[theorem]{Proposition}
\newtheorem{definition}[theorem]{Definition}
\newtheorem{remark}[theorem]{Remark}

\renewcommand{\th}{%
    \ifmmode
        ^\mathrm{th}%
    \else%
        \textsuperscript{th}\xspace%
    \fi%
}
\newcommand{\ost}{%
    \ifmmode
        ^\mathrm{st}%
    \else%
        \textsuperscript{st}\xspace%
    \fi%
}
\newcommand{\rd}{%
    \ifmmode
        ^\mathrm{rd}%
    \else%
        \textsuperscript{rd}\xspace%
    \fi%
}
\newcommand{\nd}{%
    \ifmmode
        ^\mathrm{nd}%
    \else%
        \textsuperscript{nd}\xspace%
    \fi%
}

\newcommand{\clip}[1]{#1^{\scriptscriptstyle \rm{cl}}}

\newcommand{\loc}{\scriptscriptstyle \rm{loc}}
\newcommand{\RIP}{\scriptscriptstyle \rm{RIP}}
\newcommand{\DRIP}{{\cl D^{\scriptscriptstyle (m \times n)}_{\RIP}}}
\newcommand{\spset}[1]{\Sigma^n_{#1}}
\newcommand{\cpset}[1]{\overline{\Sigma}{}^{n}_{#1}}
\newcommand{\lrset}[1]{\cl C_{#1}^{n_1\times n_2}}
\newcommand{\splus}{\raisebox{.15mm}{\scalebox{.6}{$+$}}}
\newcommand{\Amap}{{\sf A}}
\newcommand{\Amapp}{\underline{\sf A}}

\newcommand{\Dmap}{{\sf D}}

\newcommand{\Id}{\bs I}
\newcommand{\scp}[2]{\langle #1,\, #2 \rangle}
\newcommand{\compl}{{\tt c}}
\newcommand{\Rbb}{\mathbb{R}}

\newcommand{\inv}[1]{\frac{1}{#1}}
\newcommand{\supp}{{\rm supp}\,}
\newcommand{\tinv}[1]{{\textstyle\frac{1}{#1}}}
\newcommand{\sign}{{\rm sign}\,}

\renewcommand{\leq}{\leqslant}
\renewcommand{\geq}{\geqslant}

\DeclareMathOperator{\rank}{rank}
\DeclareMathOperator{\ve}{vec}
\DeclareMathOperator{\tr}{tr}

\newcommand{\bs}{\boldsymbol}
\newcommand{\bb}{\mathbb}
\newcommand{\cl}{\mathcal}

\newcommand{\ts}{\textstyle}

\newcommand{\ie}{\emph{i.e.},\xspace}
\newcommand{\eg}{\emph{e.g.},\xspace}
\newcommand{\iid}{%
    \ifmmode
        \mathrm{i.i.d.}%
    \else%
        i.i.d.\@\xspace%
    \fi%
}
\newcommand{\rv}{\mbox{r.v.\@}\xspace}
\newcommand{\rvs}{\mbox{r.v.'s\@}\xspace}
\newcommand{\st}{\mbox{s.t.\@}\xspace}
\newcommand{\whp}{\mbox{w.h.p.\@}\xspace}

\title{Quantized Compressive Sensing with RIP Matrices:\\ The Benefit of Dithering}
\author{Chunlei Xu$^*$ and Laurent Jacques\footnote{CX  and  LJ  are  with  Image  and  Signal  Processing  Group  (ISPGroup),  ICTEAM/ELEN,  UCLouvain, Belgium. E-mail: \url{{chunlei.xu,laurent.jacques}@uclouvain.be}. The authors are funded by the Belgian F.R.S.-FNRS. Part of this study is funded by the project \textsc{AlterSense} (MIS-FNRS).}}

\begin{document}
\maketitle

\begin{abstract}
Quantized compressive sensing (QCS) deals with the problem of coding compressive measurements of low-complexity signals with quantized, finite precision representations, \ie a mandatory process involved in any practical sensing model. While the resolution of this quantization clearly impacts the quality of signal reconstruction, there actually exist incompatible combinations of quantization functions and sensing matrices that proscribe arbitrarily low reconstruction error when the number of measurements increases. 

This work shows that a large class of random matrix constructions known to respect the restricted isometry property (RIP) is ``compatible'' with a simple scalar and uniform quantization if a uniform random vector, or a \emph{random dither}, is added to the compressive signal measurements before quantization. In the context of estimating low-complexity signals (\eg sparse or compressible signals, low-rank matrices) from their quantized observations, this compatibility is demonstrated by the existence of (at least) one signal reconstruction method, the \textit{projected back projection} (PBP), whose reconstruction error decays when the number of measurements increases. Interestingly, given one RIP matrix and a \emph{single} realization of the dither, a small reconstruction error can be proved to hold \emph{uniformly} for all signals in the considered low-complexity set. We confirm these observations numerically in several scenarios involving sparse signals, low-rank matrices, and compressible signals, with various RIP matrix constructions such as sub-Gaussian random matrices and random partial discrete cosine transform (DCT) matrices. 
\end{abstract}

\section{Introduction}
\label{sec:introduction}

Compressive sensing (CS) theory~\cite{CT2005,Do06,FR2013} shows us how to compressively and non-adaptively sample low-complexity signals, such as sparse vectors or low-rank matrices, in high-dimensional domains. In this framework, by exploiting the low-complexity nature of these signals, we can accurately estimate them with less measurements than the ambient domain dimension thanks to specific random sensing procedure and non-linear reconstruction algorithms (\eg $\ell_1$-norm minimization, greedy algorithms). In other words, by generalizing the concepts of sampling and reconstruction, CS has somehow extended Shannon-Nyquist theory initially restricted to the class of band-limited signals.

Specifically, anticipating on some of the mathematical notations introduced at the end of this section, CS theory describes how one can recover a signal $\bs x\in \bb R^n$ from $m \leq n$ measurements achieved from a sensing (or measurement) matrix $\bs \Phi\in \bb R^{m\times n}$
via this underdetermined linear model
\begin{equation}
  \label{eq:lin-acquiz-model}
  \ts \bs y = \bs \Phi \bs x + \bs n.  
\end{equation}
In this sensing model, $\bs y\in \bb R^m$ is the measurement vector, $\bs n \in \bb R^m$ stands for a (potential) additive measurement noise, and $\bs x$ is assumed restricted to a low-complexity signal set $\cl K \subset \bb R^n$, \eg the set $\bs\Psi \spset{k}:=\{\bs\Psi \bs \alpha\in \bb R^n:~\|\bs \alpha\|_0:=|\text{supp}(\bs \alpha)|\leq k,\ \bs \alpha \in \bb R^n\}$ of $k$-sparse vectors in an orthonormal basis $\bs \Psi \in \bb R^{n \times n}$.

The recovery of $\bs x$ is then guaranteed if $\tinv{\sqrt m}\bs \Phi$ respects the restricted isometry property (RIP) over $\cl K$, which essentially states that $\tinv{\sqrt m}\bs \Phi$ behaves as an approximate isometry for all elements of $\cl K$ (see Sec.~\ref{sec:cond-ensur-rip}). Interestingly, many random constructions of sensing matrices respect the RIP with high probability (\whp\footnote{Hereafter, we write \whp if the probability of failure of the considered event decays exponentially with respect to the number of measurements.})~\cite{BDDW08,MPT2008,Rau10,FR2013}. For instance, if $\bs \Phi$ is a \emph{Gaussian random matrix} with entries \emph{identically and independently distributed} (\iid) as a standard normal distribution $\cl N(0,1)$, \whp, $\inv{\sqrt m} \bs \Phi$ respects the RIP over $\cl K = \bs\Psi \spset{k}$ provided $m = O( k \log (n/k) )$. For a more general set $\cl K$, the RIP is verified as soon as $m$ is sufficiently large compared to the intrinsic complexity of $\cl K^* := \cl K \cap \bb B^n$ in $\bb R^n$ --- as measured by the squared Gaussian mean width $w(\cl K^*)^2$ or the Kolmogorov entropy of $\cl K^*$~\cite{KM2005,MPT2008,BDDW08} (see also Sec.~\ref{sec:low-complex-space} and Sec.~\ref{sec:specialcasePBP}). 

Under the satisfiability of the RIP, many signal reconstruction methods (\eg basis pursuit denoise~\cite{CT2005}, greedy algorithms such as the orthogonal matching pursuit~\cite{TG2007}, or iterative hard thresholding~\cite{FR2013}) achieve a stable and robust estimate of $\bs x$ from the sensing model~\eqref{eq:lin-acquiz-model}. For instance, if $\cl K = \spset{k}$, they typically display the following reconstruction error bound, or $\ell_2-\ell_1$ instance optimality~\cite{CDD09,CT2005,Do06},
\begin{equation}
  \label{eq:inst-optim}
\ts \|\bs x - \hat{\bs x}\|\ \leq\ C \frac{\|\bs x - \bs x_k\|_1}{\sqrt k} + D \epsilon,   
\end{equation}
where $\hat{\bs x}$ is the signal estimate, $\bs x_k$ the best $k$-term approximation of $\bs x$, $\epsilon \geq \|\bs n\|$ estimates the noise energy, and $C,D>0$ only depend on $\bs \Phi$ (\eg via its restricted isometry constant).   

\medskip

\paragraph{Quantization and compressive sensing:} The brief overview of CS theory above shows that, at least in the noiseless setting, the measurement vector $\bs y$ is assumed represented with infinite precision. However, in practice, any realistic sensing model integrates digitalization and finite precision data representations, \eg to store, transmit or process the acquired observations. This imposes us to adopt a quantized CS formalism where the objective is to reliably estimate a low-complexity signal $\bs x\in \cl K\subset\bb R^n$ from the quantized measurements 
\begin{equation}
\label{QCSproblem}
  \ts \bs y = \cl Q^{\rm g}(\bs \Phi \bs x),
\end{equation}
where $\cl Q^{\rm g}: \bs u \in \bb R^m \mapsto \cl Q^{\rm g} (\bs u) \in \cl A \subset \bb R^m$ is some \emph{quantization} function, or \emph{quantizer}, mapping $m$-dimensional vectors to a discrete set, or \emph{codebook}, $\cl A\subset \bb R^m$.

While initial approaches modeled the quantization distortion as an additive, bounded noise inducing a constant error bound in~\eqref{eq:inst-optim}~\cite{candes2006stable,gunturk2013sobolev,BJKS2015}, the interplay between CS, quantization, reconstruction error, number of measurements and bit-depth have now been deeply studied. This includes $\Sigma\Delta$-quantization~\cite{gunturk2013sobolev}, non-regular scalar quantizers~\cite{B_TIT_12}, 
non-regular \emph{binned} quantization~\cite{pai_nonadapt_MIT06,kamilov_2012}, and even vector quantization by frame permutation~\cite{vivekQuantFrame}. These quantizers, when associated with an appropriate signal reconstruction procedure, achieve different decay rates of the reconstruction error when the number of measurements $m$ increases. For instance, for a $\Sigma\Delta$-quantizer combined with Gaussian or sub-Gaussian sensing matrices~\cite{gunturk2013sobolev}, or with random partial circulant matrices generated by a sub-Gaussian random vector~\cite{FKS17}, this error decays polynomially in $m$ for an appropriate reconstruction procedure, and even exponentially in the case of a one-bit quantizer with adaptive thresholds~\cite{BFNPW2017}.
\medskip

\paragraph{Compatibility between quantized CS, RIP matrices, and signal space:} In this work, we adopt a different standpoint compared to the above-mentioned literature that focuses on optimizing a specific choice of sensing scheme, quantizer and algorithm to achieve optimal reconstruction errors. Our goal is to show that a simple scalar quantization procedure, \ie a uniform scalar quantizer, is \emph{compatible with the large class of sensing matrices known to satisfy the RIP}, provided that we combine the quantization with a random, uniform pre-quantization \emph{dither}~\cite{B_TIT_12,J2015,JC2016}. As clarified later, this compatibility expresses the possibility to compute close estimates of low-complexity signals observed by the corresponding quantized compressive sensing model. Accessing to a broader set of sensing matrices for QCS is desirable in many CS applications where specific, structured sensing matrices are constrained by technology or physics. For instance, (random) partial Fourier matrices are ubiquitous in magnetic resonance imaging~\cite{adcock2017breaking}, radio-astronomy~\cite{carrillo2012sparsity}, radar applications, and communications applications~\cite{ender2010compressive,anitori2013design}.

Mathematically, our aim is to estimate a signal $\bs x$ in a signal set $\cl K$ from the (scalar) QCS model
\begin{equation}
\label{eq:Uniform-dithered-quantization}
\bs y = \Amap(\bs x) = \Amap(\bs x; \bs\Phi, \bs\xi) :=\cl Q(\bs \Phi\bs x+\bs \xi),
\end{equation}
where $\Amap$ is a quantized random mapping, \ie $\Amap: \bb R^n\mapsto \delta \bb Z^m$, $\cl Q(\cdot):=\delta \lfloor\frac{\cdot}{\delta} \rfloor$ is a uniform scalar quantization of resolution\footnote{In this work, the term ``resolution'' does not refer to the number of bits used to encode the quantized values~\cite{GN1998}.} $\delta >0$ applied componentwise onto vectors (or entry-wise on matrix signals), and $\bs \xi \in \bb R^m$ is a pre-quantization \emph{dither}. We will show that, if we set $\bs \xi$ as a uniform random vector over $[0,\delta]^m$ --- what we compactly write as $\bs \xi \sim \cl U^{m}([0,\delta])$ --- then, as will be clear later, the dither \emph{attenuates} the impact of $\cl Q$ over the linear measurements of $\bs x$.  

Moreover, our objective is also to study the estimation of signals belonging to a general low-complexity set $\cl K$ in $\bb R^n$, \eg the set of sparse or compressible vectors or the set of low-rank matrices. In fact, we consider any low-complexity sets that support the RIP of $\tinv{\sqrt m}\bs \Phi$. Implicitly, this last requirement ensures that we recover CS reconstruction guarantees if the quantization disappears (\eg when $\delta \to 0$). 
\medskip

In this context, we prove the above-mentioned compatibility between the QCS model~\eqref{eq:Uniform-dithered-quantization} and the class of RIP matrices as follows. Defining $\cl P_{\cl K}(\bs z)$ as the closest point to $\bs z$ in $\cl K$~(see Sec.~\ref{sec:PBP}), we demonstrate that the simple projected back projection (PBP)
$$
\ts \hat{\bs x} := \cl P_{\cl K}(\inv{m} \bs \Phi^\top \bs y),
$$
of the quantized measurements $\bs y$ onto the set $\cl K$ achieves a reconstruction error $\|\bs x - \hat{\bs x}\|$ that decays like $O(m^{-1/p})$ when $m$ increases, for some $p >1$ only depending on $\cl K$. 

With this respect, the main results of this paper can be summarized as follows (see Sec.~\ref{sec:specialcasePBP} for their accurate descriptions). We show first that if $\cl K$ is the set of sparse vectors, the set of low-rank matrices\footnote{Up to the identification of these matrices with their vector representation (see Sec.~\ref{sec:PBP-low-rank}).}, or any finite union of low-dimensional subspaces (\eg model-based sparsity~\cite{baraniuk2010model} or group-sparse models~\cite{Ayaz16}), and in the case where $\bs \xi \sim \cl U^m([0, \delta])$ and $\tinv{\sqrt m} \bs \Phi$ is generated from a random matrix distribution (see Def.~\ref{def:RIPGen}) known to generate \whp RIP matrices over $\cl K$ (and its \emph{multiples}, see Sec.~\ref{sec:low-complex-space}), then, \whp over both $\bs \Phi$ and $\bs \xi$, and uniformly over all $\bs x \in \cl K \cap \bb B^n$ sensed from~\eqref{eq:Uniform-dithered-quantization}, PBP achieves the reconstruction error 
$$
\ts \|\bs x - \hat{\bs x}\| = O\big( (1+ \delta) \big(\frac{C_{\cl K}}{m}\big)^{-1/2}\big),
$$
with $C_{\cl K}$ depending on $\cl K$ and up to omitted log factors in the involved dimensions.

Second, if $\cl K$ is a bounded, convex and symmetric set of $\bb R^n$, \eg the set of compressible signals $\cpset{k} := \{\bs u \in \bb R^n: \|\bs u\|_1 \leq \sqrt k, \|\bs u\| \leq 1\} \supset \spset{k} \cap \bb B^n$,
then, for the same choices of sensing matrix $\bs \Phi$, dither $\bs \xi$ and QCS signal sensing model~\eqref{eq:Uniform-dithered-quantization}, the PBP error reaches \whp the decay rate 
$$
\ts \|\bs x - \hat{\bs x}\| = O\big( (1+ \delta)^{1/2} \big(\frac{C_{\cl K}}{m}\big)^{-1/p}\big),
$$
with $p=16$ and $p=18$ in the uniform and in the non-uniform (\ie valid of a single $\bs x \in \cl K$) recovery settings, respectively.

Knowing if other reconstruction algorithms can reach faster error decay in our QCS model and for general RIP matrices is an open question. In this regard, PBP can be seen as a good initialization for more advanced reconstruction algorithms iteratively enforcing the consistency of the estimate with the observations $\bs y$~\cite{ZBC2010,oymak2015near,shi16,JLBB2013,JDV2015, J2016}.   

\paragraph{Dithering and limited projection distortion (LPD):} Let us stress that the dither introduced in the QCS model~\eqref{eq:Uniform-dithered-quantization} has a crucial role in our developments. Intuitively, its importance in the quantization scheme comes from this simple observation:
$$
\text{for $u \sim \cl U([0,1])$,}\quad \bb E \lfloor \lambda + u\rfloor = \lambda,\quad \forall \lambda \in \bb R,
$$
as proved in Lemma~\ref{lem1} in App.~\ref{app:vanishing-dithered-quantiz}. By the law of large numbers, this means that for $m$ different random variables (\rvs) $u_i \sim_{\iid} \cl U([0,1])$ with $1\leq i \leq m$ and $m$ increasingly large, an arbitrary projection of the vector $\bs r(\bs a) := \lfloor \bs a + \bs u\rfloor - \bs a := (\lfloor a_1 + u_1\rfloor - a_1, \cdots, \lfloor a_m + u_m\rfloor - a_m)^\top$ for some vector $\bs a \in \bb R^m$ onto a fixed direction $\bs b \in \bb R^m$ tends to the expectation $\bs b^\top \big(\bb E \bs r(\bs a)\big) = 0$. Moreover, this effect should persist for all $\bs a$ and $\bs b$ selected in a set whose dimension is small compared to $\bb R^m$, \ie if these vectors are selected in the image of a low-complexity set $\cl K$ by a RIP matrix~$\tinv{\sqrt m} \bs \Phi$. 

In order to accurately bound the deviation between these projections around zero, we show in Sec.~\ref{sec:LPD-noisy-linear-map} and Sec.~\ref{sec:LPD} that, given a resolution $\delta>0$, a RIP matrix $\tinv{\sqrt m} \bs \Phi$, a uniform random dither $\bs \xi$, and provided that $m$ is large compared to the intrinsic complexity of $\cl K$ (as measured by its Kolmogorov entropy, see Sec.~\ref{sec:low-complex-space}), the quantized random mapping $\Amap: \bs u \to \delta \lfloor \delta^{-1}(\bs \Phi \bs u + \bs \xi)\rfloor$ of \eqref{eq:Uniform-dithered-quantization} respects \whp the \emph{limited projection distortion} (or LPD) property over $\cl K$, \ie assuming $\cl K \subset \bb B^n$ (or restricting it to the unit ball $\bb B^n$ if it is unbounded)
$$
\ts \tinv{m}|\scp{\Amap(\bs u)}{\bs \Phi \bs v} - \scp{\bs \Phi \bs
  u}{\bs \Phi \bs v}|\ \leq\ \nu,\quad \forall \bs u, \bs v \in \cl K,
$$
where $\nu >0$ is a certain distortion depending on $\bs \Phi$, $\delta$, $n$ and $m$. This can be proved using tools from measure concentration theory and some extra care to deal with the discontinuities of $\cl Q$. In fact, we deduce in Sec.~\ref{sec:LPD} that $\nu = \epsilon\,(1 + \delta)$ if the dither is random and uniform, where $\epsilon$ is an arbitrary small distortion impacting the requirement on $m$. For instance, forgetting all other dependencies, $m = O(\epsilon^{-2}\log(1/\epsilon))$ for the set of sparse vectors, as similarly established for ensuring the RIP of a Gaussian random matrix on the same set~\cite{BDDW08} (see Sec.~\ref{sec:LPD} and Sec.~\ref{sec:conditions-llpd}). Moreover, for the localized LPD (or L-LPD) obtained by fixing $\bs u \in \cl K$ in the LPD, the distortion is reduced to $\nu = \delta\epsilon$, as deduced in Sec.~\ref{sec:LPD-noisy-linear-map}.

\paragraph{Utility of the LPD:} Anticipating on Sec.~\ref{sec:PBP-gen-error}, we now give a quick explanation illustrating why the LPD property allows us to bound the reconstruction error of PBP in the estimation of bounded $k$-sparse signals. We first note that if $\inv{\sqrt{m}}\bs \Phi$ satisfies the RIP over $\spset{2k}$, then this matrix also embeds scalar products of sparse vectors, \ie $\ts \inv{m}\,|\langle \bs
\Phi\bs u,\bs \Phi\bs v \rangle -\langle \bs u,\bs v\rangle|\leq \epsilon$ for all $\bs u, \bs v \in \bb B^n$ such that $\bs u \pm \bs v \in \spset{2k}$ (see, \eg~\cite{FR2013}). Then, from the LPD of $\Amap$ over $\cl K = \spset{2k} \cap \bb B^n$ and for the same vectors $\bs u$ and $\bs v$,  the triangular identity provides
\begin{equation}
  \label{eq:pre-lpd-rip}
  \ts \,|\inv{m} \langle \Amap(\bs u), \bs \Phi \bs v\rangle- \langle
  \bs u,\bs v\rangle|\ \leq\ \epsilon + \nu.
\end{equation}
Moreover, for a bounded sparse signal ${\bs x \in \spset{k} \cap \bb B^n}$, the PBP estimate $\hat{\bs x} = \cl P_{\spset{k}}(\tinv{m}\bs\Phi^\top \bs y)$ --- with $\cl P_{\spset{k}}(\bs z)$ the hard thresholding operator that keeps the $k$ largest elements (in magnitude) of $\bs z$ and sets the others to zero --- is the best $k$-sparse approximation of both $\bs a := \tinv{m}\bs\Phi^\top \bs y$ and the vector $\bar{\bs a}$ obtained by zeroing all the entries of $\bs a$ but those indexed in $\cl T := \supp(\hat{\bs x})\cup\supp(\bs x)$. This induces that $\|\hat{\bs x} - \bs x\| \leq \|\hat{\bs x} - \bar{\bs a}\| + \|\bar{\bs a} - \bs x\| \leq 2\|\bar{\bs a} - \bs x\|$, and, from the definition of the $\ell_2$-norm,
$$
\|\bar{\bs a} - \bs x\| = \sup_{\bs v \in \bb B^n} \scp{\bar{\bs a} - \bs x}{\bs v} = \sup_{\substack{\bs v \in \bb B^n\\\supp \bs v \subset \cl T}} \scp{\bs a - \bs x}{\bs v} = \sup_{\substack{\bs v \in \bb B^n\\\supp \bs v \subset \cl T}} \tinv{m}\scp{\Amap(\bs x)}{\bs\Phi \bs v} - \scp{\bs x}{\bs v}.
$$  

Therefore, since $\bs v \pm \bs x$ is at most $2k$-sparse, \eqref{eq:pre-lpd-rip} directly provides the final error bound
\begin{equation}
  \label{eq:spoiled-pbp-error-sparse-case}
  \|\hat{\bs x} - \bs x\| \leq 2(\epsilon + \nu).  
\end{equation}

While \eqref{eq:spoiled-pbp-error-sparse-case} results from the RIP and the LPD seen as deterministic properties of $\tinv{\sqrt m} \bs \Phi$ and~$\Amap$, respectively, we derive in Sec.~\ref{sec:specialcasePBP} the decay rate of the PBP error when $m$ increases from the conditions relating $m$, $\nu$ and $\epsilon$ to ensure \whp these properties when $\tinv{\sqrt m} \bs \Phi$, the dither $\bs \xi$ (and thus $\Amap$) are random. Moreover, we show in Sec.~\ref{sec:PBP-gen-error} that the usefulness of the LPD extends beyond sparse signals and covers the cases of low-rank matrices, union of low-dimensional subspace and bounded, convex and symmetric sets (\eg the set of compressible vectors in $\bb R^n$).

\medskip
\paragraph{Paper organization:} The rest of the paper is structured as follows. We present in Sec.~\ref{sec:motivationSOA} a few related works, namely, former usages of the PBP method in one-bit CS, variants of the LPD property in one-bit CS and non-linear CS, and certain known reconstruction error bounds of PBP and related algorithms for certain QCS and non-linear sensing contexts. Most of these works are based on sub-Gaussian random projections of signals altered by quantization or other non-linear disturbances, with two noticeable exceptions using subsampled Gaussian circulant sensing matrix and bounded orthogonal ensemble~\cite{DJR2017,HS2018QCSStruMatrix}. We then introduce a few preliminary concepts in Sec.~\ref{sec:preliminaries}, such as the characterization of low-complexity spaces, a general definition of random matrix distributions capable to generate, \whp, RIP matrices, the PBP method, and the formal definition of the (L)LPD for a general mapping $\Dmap:\bb R^n \to \bb R^m$. Sec.~\ref{sec:PBP-gen-error} establishes the reconstruction error bound of PBP when the LPD of $\Dmap$ and the RIP of $\tinv{\sqrt m}\bs \Phi$ are both verified and seen as deterministic properties. We realize this analysis for three kinds of low-complexity sets, namely, finite union of low-dimensional spaces (\eg the set of (group) sparse signals), the set of low-rank matrices, and the (unstructured) case of a general bounded convex set. In Sec.~\ref{sec:LPD-noisy-linear-map}, we prove that the L-LPD holds \whp over low-complexity sets when $\Dmap$ represents a linear sensing model corrupted by an additive sub-Gaussian noise. This analysis simplifies the characterization of the PBP reconstruction error from QCS observations in the non-uniform recovery case.  In Sec.~\ref{sec:LPD}, we prove that the quantized random mapping $\Amap$ defined in \eqref{eq:Uniform-dithered-quantization} from a uniform random dither is sure to respect, \whp, the (uniform) LPD provided $m$ is large compared to the complexity of $\cl K$.  In Sec.~\ref{sec:specialcasePBP}, from the results of the two previous sections, we leverage the general bounds found in Sec.~\ref{sec:PBP-gen-error} and establish the decay rate of the PBP reconstruction error when $m$ increases for the same low-complexity sets and random matrix distributions considered in Sec.~\ref{sec:PBP-gen-error} and Sec.~\ref{sec:cond-ensur-rip}, respectively. Finally, in Section~\ref{sec:numericalPBP}, we validate the PBP reconstruction error numerically for the particular sets discussed in Sec.~\ref{sec:specialcasePBP} and several structured and unstructured random sensing matrices~$\bs \Phi$. 
 
\paragraph{Conventions and notations:} Throughout this paper, we denote vectors and matrices with bold symbols, \eg~$\bs \Phi \in \bb R^{m\times n}$
or~$\bs u \in \bb R^m$, while lowercase light letters
are associated with scalar values. The class of functions from $\bb R^n \to \bb R$ that are continuous over $\cl E \subset \bb R^n$ is denoted by ${\sf C}^0(\cl E)$. The identity matrix in $\bb R^{n}$ reads
$\Id_n$ and the zero vector $\bs 0 := (0,\,\cdots,0)^\top\in \bb R^n$, its dimension being clear from the context. The~$i$\th component of a vector (or of a
vector function)~$\bs u$ reads either~$u_i$ or~$(\bs u)_i$, while the
vector~$\bs u_i$ may refer to the~$i\th$ element of a set of
vectors. The set of indices in~$\bb R^d$ is~$[d]:=\{1,\,\cdots,d\}$ and the support of $\bs u \in \bb R^d$ is $\supp \bs u \subset [d]$. The Kronecker symbol is denoted by $\delta_{ij}$ and is equal to 1 if $i=j$ and to 0 otherwise, while the indicator $\chi_{\cl S}(i)$ of a set $\cl S \subset [d]$ is equal to 1 if $i \in \cl S$ and to 0 otherwise. For any~$\cl S \subset [d]$ of cardinality~$S = |\cl S|$,
$\bs u_{\cl S} \in \bb R^{S}$ denotes the restriction of
$\bs u$ to~$\cl S$, while $\bs B_{\cl S}$ is the matrix obtained by restricting the columns
of~$\bs B \in \bb R^{d \times d}$ to those indexed by~$\cl S$. The complement of a set $\cl S$ reads $\cl S^\compl$. For any~$p\geq 1$, the~$\ell_p$-norm of~$\bs u$ is
$\|\bs u\|_p = (\sum_i |u_i|^p)^{1/p}$ with~$\|\!\cdot\!\|=\|\!\cdot\!\|_2$. The \emph{radius} of a bounded set $\cl K \subset \bb R^n$ is $\|\cl K\| := \sup\{\|\bs u\|\!:\!\bs u \in \cl K\}$. The~$\ell_p$-sphere in~$\Rbb^n$ is
$\bb S_{\ell_p}^{n-1}=\{\bs x\in\Rbb^n: \|\bs x\|_p=1\}$, and the unit $\ell_p$-ball reads~$\bb B_{\ell_p}^{n}=\{\bs x\in\Rbb^n: \|\bs x\|_p\leq 1\}$. For $\ell_2$, we write $\bb B^n = \bb B^n_{\ell_2}$ and $\bb S^{n-1} = \bb S_{\ell_2}^{n-1}$. By extension, $\bb B_{F}^{n_1 \times n_2}$ is the Frobenius unit ball of $n_1 \times n_2$ matrices $\bs U$ with $\|\bs U\|_F \leq 1$, where the Frobenius norm ${\|\cdot\|_F}$ is associated with the scalar product $\scp{\bs U}{\bs V} = \tr(\bs U^\top \bs V)$ through $\|\bs U\|_F^2 = \scp{\bs U}{\bs U}$, for two matrices $\bs U, \bs V$. We also extend many vector properties or concepts defined from the $\ell_2$-norm to matrices thanks to the Frobenius norm (\ie by vectorization). The common flooring and sign (or \emph{signum}) functions are $\lfloor \cdot \rfloor$ and $\sign(\cdot)$, respectively. Unless mentioned otherwise, the symbols $C,C', C'', ..., c,c',c'',... > 0$ are positive and \emph{universal} constants whose values can change from one line to the other. We also use the ordering notations $A \lesssim B$ (or $A \gtrsim B$), if there exists a $c > 0$ such that $A \leq c B$ (resp. $A \geq c B$) for two quantities $A$ and~$B$. Moreover, $A \asymp B$ if we have both $A \lesssim B$ and $B \lesssim A$.

Concerning statistical quantities, $\cl X^{m\times
  n}$ and~$\cl X^{m}$ denote an~$m\times n$ random matrix and an~$m$-length random vector, respectively, whose entries are
\emph{identically and independently distributed} (or $\iid$) as the probability distribution
$\cl X$, \eg $\cl N^{m \times n}(0,1)$ (or $\cl U^m([0, \delta])$) is
the distribution of a matrix (resp. vector)
whose entries are $\iid$ as 
the standard normal distribution~$\cl N(0,1)$ (resp. the uniform
distribution~$\cl U([0, \delta])$). We also use extensively the sub-Gaussian and sub-exponential
characterization of random variables (or r.v.) and of random vectors detailed in
\cite{V2012}. The sub-Gaussian and the
sub-exponential norms of a random
variable $X$ are thus denoted by $\|X\|_{\psi_2}$ and $\|X\|_{\psi_1}$, respectively,
with the Orlicz norm $\|X\|_{\psi_\alpha} := \sup_{p\geq 1} p^{-1/\alpha} (\bb E
|X|^p)^{1/p}$ for $\alpha \geq 1$. The random variable $X$ is therefore
sub-Gaussian (or sub-exponential) if $\|X\|_{\psi_2} < \infty$
(resp. $\|X\|_{\psi_1} < \infty$). 

\section{Related works}
\label{sec:motivationSOA}

\begin{table}[!t]
\begin{center}
\noindent\scalebox{1}{\tiny\def\arraystretch{1.5}
\begin{tabular}{|@{\,}p{13mm}|@{\ }p{21mm}@{\ }|@{\ }p{21mm}@{\ }|@{\ }p{21mm}@{\ }|@{\ }p{17mm}@{\ }|@{\ }p{17mm}@{\ }|@{\ }p{17mm}@{\,}|@{\ }p{17mm}@{\ }|}
\hline
&\multicolumn{3}{@{}c@{}|@{\ }}{QCS (quantizer $\cl Q$)\ \&\ Non-linear CS (non-linear $f_i$)\,}&\multicolumn{3}{@{\ }c@{}|@{\ }}{one-bit CS}&one-bit \& QCS\\
\hline
 Works&\cite{PV2016,OptQuanLass2016}&\cite{PVY2017}&\textbf{(this work)}&\cite{ALPV2014, PV2013}&\cite[Sec. 5]{JDV2015} &\cite{BFNPW2017}&\cite{DJR2017}\\ 
\hline
Signal set\newline ${\cl K \subset \bb R^n}$&With low (local) GMW on the tangent cone on~$\bs x$&With low (local) GMW; closed and star-shaped, or conic&ULS, LR, or convex sets \st $\cl K \cap \bb B^n$ has low KE. &   $\cl K\subset \bb B^n$, low GMW &   $\spset{k}\cap \bb B^n$ & $\spset{k}\cap \bb B^n$ &  $\spset{k}\cap \bb B^n$ \\ \hline
Condition on $\bs x$&$\bs x\in \bb S^{n-1} \cap \tinv{\mu}\cl K$,\newline $\mu = \bb E f_i (g) g$&$\mu\bs x/\|\bs x\| \in \cl K$&$\bs x \in \cl K\cap \bb B^n$&$\bs x \in \cl K$&$\bs x \in \cl K$&$\bs x \in \cl K$&$\bs x \in \cl K$\\ \hline
Sensing model\newline for $y_i$&$f_i(\bs \varphi_i^\top\bs x)$\newline $f_i$ random~\cite{PV2016},\newline or $f_i = \cl Q$ in~\cite{OptQuanLass2016} &$f_i(\bs \varphi_i^\top\bs x)$\newline $f_i$ random&$ \cl Q(\bs \varphi_i^\top\bs x+\xi_i)$&$ f_i(\bs \varphi_i^\top\bs x)$\newline $f_i$ binary \&\newline random &$\text{sign}(\bs \varphi_i^\top\bs x)$&$\text{sign}(\bs \varphi_i^\top\bs x~-~\tau_i)$&$\text{sign}(\bs \varphi_i^\top\bs x~-~\tau_i)$\newline
or \newline $\cl Q(\bs \varphi_i^\top \bs x - \tau_i)$
\\ 
\hline
   Sensing image & Any~\cite{PV2016};\newline quantized~\cite{OptQuanLass2016}&  Any&  $\delta \bb Z^m$ & binary &  binary &  binary &  binary\newline or $\delta (\bb Z + 1/2)^m$\\ 
\hline
 
   Dithering &   Not explicit & Not explicit & $\displaystyle\xi_i~\underset{\iid}{\sim}~\cl U([0,\delta])$ & Not explicit &   No  &  Adaptive\newline or random $\tau_i$ &  None\newline or random~$\tau_i$ \\ \hline
 
Condition(s) on $\bs \Phi =$
\vspace{-.5mm}\newline \scalebox{0.75}{$(\bs\varphi_1,...,\bs\varphi_m)^\top$}
&L-LPD variant& Not explicit &RIP (inducing (L)LPD of $\Amap$ with random dithering)&SPE &   SPE &   SPE and\newline ``$\ell_1$-quotient'' (see~\cite{BFNPW2017})&   ``$\ell_1/\ell_2$-RIP'' (see~\cite{DJR2017})\\
\hline
Known compatible matrices&GRM&GRM&Any matrix satisfying the RIP over (a multiple of) $\cl K$&SGRM&GRM&GRM&SGCM\\ \hline
Algorithms &   K-Lasso &   PBP &  PBP& $\displaystyle \max_{\bs z\in \cl K} \ts \bs y^{\top} (\bs \Phi\bs z)$ &   PBP &PBP&PBP \& SOCP program \\ 
\hline
Uniform/\newline \scalebox{.9}{Non-uniform} guarantees &  Non-uniform  &   Non-uniform&  Both& Non-uniform in~\cite{ALPV2014}, both in~\cite{PV2013}&  Both &  Uniform &  Uniform \\ 
\hline
\end{tabular}}\vspace{-5mm}
\end{center}
\caption{This table summarizes the setting and the conditions on the (quantized/one-bit or non-linear) sensing model, its associated sensing matrix, the low-complexity set including the observed signal, as well as the algorithms used for reconstructing signals from either (quantized/one-bit or non-linear) measurements. In this table, the acronyms are Gaussian mean width (GMW,  Sec.~\ref{sec:low-complex-space}), Kolmogorov entropy  (KE, Sec.~\ref{sec:low-complex-space}), union of low-dimensional subspaces (ULS), low-rank models (LR), second order cone programming (SOCP), 
  Gaussian random matrices (GRM), sub-Gaussian random matrices (SGRM), and subsampled Gaussian circulant matrix (SGCM)~\cite{DJR2017}. The concept of ``multiple of $\cl K$'' is defined in Sec.~\ref{sec:low-complex-space}.}\label{table:compare-related-works}
\end{table}
 
We now provide a comparison of our work with the most relevant literature. This one is connected to one-bit CS, QCS or other non-linear sensing models, when these are applied componentwise (as for scalar quantization) on compressive sensing measurements, possibly with a random or adaptive pre-quantization dither for certain studies, with algorithms similar or related to PBP. Most of the works presented below are summarized in Table~\ref{table:compare-related-works}, reporting there, among other aspects, the sensing model, the algorithm, the type of admissible sensing matrices and the low-complexity sets chosen in each of the referenced works.

\paragraph{PBP in one-bit CS:} Recently, signal reconstruction via projected back projection has been studied in the context of one-bit compressive sensing (one-bit CS), an extreme QCS scenario where only the signs of the compressive measurements are retained
\cite{BB2008,BFNPW2017,JDV2015,PV2013}. In this case~\eqref{QCSproblem} is turned into 
\begin{equation}
  \label{eq:bin-sensing-model}
\ts \bs y=\sign(\bs \Phi \bs x).  
\end{equation}
If the sensing matrix $\bs \Phi \in \bb R^{m
  \times n}$ satisfies the \textit{sign product embedding property}
(SPE) over $\cl K = \spset{k} \cap \bb S^{n-1}$~\cite{PV2013,JDV2015}, that is, up to some distortion $\epsilon >0$ and some universal normalization $\mu >0$, 
\begin{equation}
\tag{SPE}
\ts |\frac{\mu}{m}\scp{\sign(\bs \Phi \bs u)}{\bs \Phi \bs v} - \scp{\bs u}{\bs v}| \leq \epsilon,  
\end{equation}
for all $\bs u, \bs v \in \cl K$, then the reconstruction error of the
PBP of $\bs y$ is bounded by $2\epsilon$~\cite[Prop. 2]{JDV2015}. 
In other words, for a signal $\bs x \in \spset{k}$ with unknown norm (as implied by the invariance of \eqref{eq:bin-sensing-model} to positive
signal renormalization~\cite{BB2008}), the PBP method allows us to estimate the direction of a
sparse signal. This remains true for all methods assuming $\bs x$ to be of unit norm, such as those explained below. 

So far, the SPE property has only been proved for Gaussian random sensing
matrices, with \iid standard normal entries, for which $\mu=\sqrt{2/\pi}$. Such matrices respect
the SPE with high probability if $m = O(\epsilon^{-6} k \log \frac{n}{k})$,
conferring to PBP a (uniform) reconstruction error decay of
$O(m^{-1/6})$ when $m$ increases for all $\bs x \in \spset{k} \cap \bb S^{n-1}$. Besides, by
\emph{localizing} the SPE to a given $\bs u \in \spset{k} \cap \bb S^{n-1}$, a
non-uniform variant of the previous result, \ie where $\bs \Phi$ is
randomly drawn conditionally to the knowledge of $\bs x$, gives error decaying as fast as $O(m^{-1/2})$~\cite[Prop. 2]{JDV2015} when $m$ increases.   

For more general low-complexity set $\cl K \subset \bb B^n$ with ${\cl K \cap \bb S^{n-1}} \neq \emptyset$ and $\bs x \in {\cl K \cap \bb S^{n-1}}$, provided $\bs \Phi$ is a random sub-Gaussian matrix \cite{V2012}, the vector $\hat{\bs x} \in \cl K$ maximizing its scalar product with $\bs\Phi^\top\bs y$ is \whp a good estimate of $\bs x$ with small reconstruction error~\cite{ALPV2014}. This holds even if the binary sensing
model~\eqref{eq:bin-sensing-model} is noisy (\eg with possible random sign flips on a small percentage of the measurements). In fact, this
error decays like $C m^{-1/4} + D \|\bs x\|^{1/4}_\infty $ when $m$ increases, with $C >0$
depending only on the level of measurement noise, on the distribution of $\bs \Phi$, and on $\cl K$ (actually, on its Gaussian mean width, see Sec.~\ref{sec:specialcasePBP}), while
$D>0$ is associated with the non-Gaussian nature of the sub-Gaussian random matrix $\bs \Phi$ (\ie $D=0$
if it is Gaussian)~\cite[Thm 1.1]{ALPV2014}. Therefore, for one-bit CS with sub-Gaussian random matrices (\eg Bernoulli), the reconstruction error is not anymore guaranteed to decrease below a certain floor level $D$, and this level is driven by the sparsity of $\bs x$ (\ie it is high if $\bs x$ is very sparse). Actually, for Bernoulli random matrices (as well as for partial random Fourier matrices \cite{feuillen2018quantity}), there exist counterexamples of $2$-sparse signals that cannot be reconstructed, \ie with constant reconstruction error if $m$ increases, showing that the bound above is tight~\cite{PV2013}.
 
For Gaussian random sensing matrices, adding an adaptive or random dither to the compressive
measurements of a signal before their binarization allows accurate reconstruction of this signal (\ie its norm and direction)~\cite{KSW16,BFNPW2017}, using either PBP or a second-order cone program (SOCP). Additionally, for random observations altered by an adaptive dither before their one-bit quantization, \ie similarly to noise shaping techniques or $\Sigma\Delta$-quantization~\cite{gunturk2013sobolev,BJKS2015}, an appropriate reconstruction algorithm can achieve an exponential decay of its error in terms of the number of measurements. This is only demonstrated, however, in the case of Gaussian sensing matrices and for sparse signals~only.

More recently,~\cite{DJR2017} has shown that, if $\bs \Phi$ is a subsampled Gaussian circulant matrix in the binary observation model~\eqref{eq:bin-sensing-model}, PBP
can reconstruct the direction of any sparse vector up to an error
decaying as $O(m^{-1/4})$ (see~\cite[Thm 4.1]{DJR2017}). Moreover, by
adding a random dither to the linear random measurements before their
binarization, the same authors proved that the SOCP program of \cite{KSW16}
can fully estimate the vector $\bs x$ if an upper bound $R \geq \|\bs x\|$ is known a priori. These results extend to the dithered, uniformly quantized
CS expressed in~\eqref{eq:Uniform-dithered-quantization} (with quantization resolution $\delta>0$) with the same sensing matrix. In fact, with high probability, and for all
\emph{effectively sparse} signals $\bs x \in \bb B^n$, \ie such that $\|\bs x\|_1$ is small, the same SOCP program achieves a reconstruction error decaying like $O(m^{-1/6})$ when $m$ increases, provided that the dither is made of a Gaussian random vector with variance $R^2$ added to a uniform random vector adjusted to $\delta$~\cite[Thm 6.2]{DJR2017}.  

\paragraph{QCS and other non-linear sensing models:} The (scalar) QCS model~\eqref{eq:Uniform-dithered-quantization} can be seen as a special case of the more general, non-linear sensing model $\bs y = \bs f(\bs \Phi \bs x)$, with the random non-linear function $\bs f: \bs u \in \bb R^m \mapsto (f_1(u_1), \cdots, f_m(u_m))^\top\in\bb R^m$ such that $f_i \sim_{\iid} f$, for some random function $f: \bb R \to \bb R$~\cite{PV2016,PVY2017}. In the QCS context defined in~\eqref{eq:Uniform-dithered-quantization}, this non-linear sensing model corresponds to setting $f_i(\lambda) = \cl Q(\lambda + \xi_i)$ with $\xi_i \sim_{\iid} \cl U([0,\delta])$.

In~\cite{PVY2017}, the authors proved that, for a Gaussian random matrix $\bs \Phi$ and for a bounded, star-shaped\footnote{The set $\cl K$ is star-shaped if, for any $\lambda \in [0,1]$, $\lambda \cl K \subset \cl K$; in particular, $\bs 0 \in \cl K$.} set $\cl K$, provided that $f$ leads to finite moments $\mu := \bb E f(g)g$, $\sigma^2:=\bb E f(g)^2-\mu^2$, and $\eta^2 := \bb E f(g)^2$ with $g\sim \cl N(0,1)$, and provided that $f(g)$ is sub-Gaussian with finite sub-Gaussian norm $\psi := \|f(g)\|_{\psi_2}$~\cite{V2012}, one can estimate with high probability $\mu\frac{\bs x}{\|\bs x\|} \in \cl K$ from the solution $\hat{\bs x}$ of the PBP of $\bs y$ in~\eqref{eq:gen-non-lin-sensing} (see~\cite[Thm 9.1]{PVY2017}). In the specific case where $\bs f$ matches the QCS model~\eqref{eq:Uniform-dithered-quantization}, this analysis proves that for Gaussian random matrix $\bs \Phi$, the PBP of QCS observations estimates the direction $\bs x/\|\bs x\|$ with a reconstruction error decaying like $O((1+\delta)^2 \sqrt{w(\cl K)} m^{-\frac{1}{4}})$ when $m$ increases (the details of this analysis are given in App.~\ref{app:values-from-Klasso}).

A similar result is obtained in~\cite{PV2016} for the estimate $\hat{\bs x}$ provided by a K-Lasso program, which finds the element $\bs u$ of $\cl K$ minimizing the $\ell_2$-cost function $h(\bs u) := \|\bs \Phi \bs u - \bs y\|^2$, when $\bs y = \bs f(\bs \Phi \bs x)$, $\bs x \in \bb S^{n-1} \cap \inv{\mu}\cl K$ and under the similar hypotheses on the non-linear corruption $f_i \sim f$ than above (\ie with finite moments $\mu$, $\sigma$, and $\tilde\eta^2=\bb E(f(g)-\mu g)^2g^2)$\,). Of interest for this work,~\cite{PV2016} introduced a form of the (local) LPD (given in Sec.~\ref{sec:introduction} and Sec.~\ref{sec:preliminaries}) in the case where $\Amap  = f\circ \bs \Phi$ and $\bs \Phi$ is a Gaussian random matrix (with possibly unknown covariance between rows). The authors indeed analyzed when, for some $\epsilon > 0$, 
\begin{equation}
\label{eq:LPD-plan-vershynin}
  \ts \tinv{m}\,\big(\scp{f(\bs \Phi \bs x)}{\bs \Phi \bs v} - \scp{\bs \Phi\mu\bs x}{\bs \Phi \bs v}\big) \lesssim \epsilon,\quad \forall \bs v \in \cl D^* = \cl D \cap \bb B^n,
\end{equation}
with $\cl D = \cl D(\cl K, \mu \bs x) := \{\tau\bs h: \tau \geq 0, \bs h \in \cl K - \mu\bs x\}$ being the \emph{tangent cone} of $\cl K$ at $\mu \bs x$. An easy rewriting of~\cite[Proof of Thm 1.4]{PV2016} then essentially shows that the RIP of $\bs \Phi$ over $\cl D^*$ combined with~\eqref{eq:LPD-plan-vershynin} provides $\|\hat{\bs x} - \mu \bs x\| \lesssim \epsilon$. In particular, thanks to the Gaussianity of $\bs \Phi$, they prove that, with large probability,~\eqref{eq:LPD-plan-vershynin}~holds\footnote{As implied by Markov’s inequality combined with~\cite[Lemma 4.3]{PV2016}.} with $\epsilon = (w(\cl D^*)\sigma+\tilde\eta)/\sqrt{m}$, with $w(\cl D^*)$ the Gaussian mean width of $D^*$ measuring its intrinsic complexity (see Sec.~\ref{sec:low-complex-space}). For instance, if $\cl K = \spset{k}$ and if $f$ is such that $\mu,\sigma,\tilde\eta=O(1)$, this proves that $\epsilon = O(\sqrt{k \log(n/k)}/\sqrt{m})$. Correspondingly, if the $f_i$'s are thus selected to match the QCS model with a Gaussian sensing $\bs \Phi$, this shows that K-Lasso achieves a non-uniform reconstruction error decay of $O(1/\sqrt{m})$ of $\bs x \in \bb S^{n-1}\cap \tinv{\mu}\cl K$ if the Gaussian mean width of the tangent cone $\cl D^*$ can be bounded (\eg for sparse or compressible signals, or for low-rank matrices).  In other words, when instantiated to our specific QCS model, but only in the context of a Gaussian random matrix $\bs \Phi$ and with some restrictions on the norm of $\bs x$, the decay of the non-uniform reconstruction error of both PBP and K-Lasso in~\cite{PVY2017} and~\cite{PV2016}, respectively, are similar to the one achieved in our work (see Sec.~\ref{sec:specialcasePBP}). 

More recently, in the context of $\Sigma\Delta$ and noise-shaping quantization,~\cite{HS2018QCSStruMatrix} has extended initial works restricted to the use of Gaussian random matrices \cite{gunturk2013sobolev} by proving that if $\bs \Phi$ is either a bounded orthogonal ensemble (\eg a random partial Fourier matrices) or a partial circulant ensemble (also known as subsampled circulant matrix), a convex program proposed in the paper achieves a uniform reconstruction of any $\bs x\in \spset{k}\cap B^n$ with an error decaying polynomially fast in $m$ for the $\Sigma\Delta$ quantization, or exponentially fast in $m$ for the noise-shaping quantization. However, the analysis of~\cite{HS2018QCSStruMatrix} is limited for the estimation of sparse signals, although they characterize binary embeddings of finite and general low-complexity sets from the same quantization schemes. 

\paragraph{Bussgang's theorem and distorted correlators:}

Let us finally mention that asymptotic forms of the LPD property were already analyzed since 1952 in stochastic analysis and in the study of distorted Gaussian random processes \cite{bussgang1952crosscorrelation}, \eg for the design of efficient signal correlators under quantization constraints~\cite[Chap. 4]{zebadua:tel-01761603}\cite{zebadua2017compressed}.

For instance, Bussgang's theorem states that if $X_t$ is a zero-mean stationary Gaussian process, and that $Y_t := f(X_t)$ with $f$ a non-linear distortion, \eg a one-bit or a uniform quantizer, then the crosscorrelation of $X$ and $Y$ is proportional to the autocorrelation of $X$. 
Mathematically,
$$
R_{XY}(\tau)\ =\ C_g R_X(\tau),\quad \forall \tau \in \bb R,
$$
with $C^{-1}_g = \bb E f(X_0) X_0$, $R_{XY}(\tau) := \bb E Y_t X_{t + \tau} = \bb E Y_0 X_{\tau}$ and $R_{X}(\tau) := \bb E X_t X_{\tau} = \bb E X_0 X_{\tau}$, for all $t,\tau \in \bb R$ (from the stationarity of $X$).

For instance, for $f(\cdot) = \sign(\cdot)$, given two vectors $\bs x, \bs x' \in \bb S^{n-1}$ and the Gaussian processes $\scp{\bs \varphi}{\bs x}$ and $\scp{\bs \varphi}{\bs x'}$ defined on $\bb S^{n-1}$ with $\bs \varphi \sim \cl N^{n}(0,1)$,
Bussgang's theorem is equivalent to the ``arcsin law'', with $R_{XY}(\bs x) = \bb E f(\scp{\bs \varphi}{\bs x}) \scp{\bs \varphi}{\bs x'} = {(\frac{2}{\pi})}^{\scriptscriptstyle \frac{1}{2}}\,\scp{\bs x}{\bs x'}$ proportional to $R_{X}(\bs x) = \bb E \scp{\bs \varphi}{\bs x} \scp{\bs \varphi}{\bs x'} = \scp{\bs x}{\bs x'}$.

In fact, Bussgang's theorem is essentially an asymptotic form of the LPD distortion for a general distortion $\Dmap(\bs x) = f(\bs \Phi \bs x)$ with $\bs \Phi = (\bs \varphi_1^\top, \cdots, \bs \varphi_1^\top )^\top \sim \cl N^{m \times n}(0,1)$ and a distortion $g$ applied componentwise on $\bs \Phi \bs x$; for $m$ large and $\bs \varphi \sim \cl N^n(0,1)$,
$$
\ts \tinv{m} \scp{\Dmap(\bs x)}{\bs \Phi \bs x'}\ \approx\ \bb E\big[f\big(\scp{\bs \varphi}{\bs x}\big) f\big(\scp{\bs \varphi}{\bs x'}\big)\big]\ =\ C_g\,\bb E\big[\scp{\bs \varphi}{\bs x} f\big(\scp{\bs \varphi}{\bs x'}\big)\big]\ \approx\ \tinv{m} \scp{\bs \Phi \bs x}{\bs \Phi \bs x'}.  
$$
This observation is also exploited in the K-Lasso program for non-linear compressive sensing \cite{PV2016} where the first moment $\mu$ amounts to the proportionality factor $C_g$ above.

\section{Preliminaries}
\label{sec:preliminaries}

This section introduces the key concepts of this work, namely: how we measure the complexity of a signal set (Sec.~\ref{sec:low-complex-space}); the definition of random matrix distributions known to generate matrices satisfying (\whp) the restricted isometry property (RIP) over a low-complexity set (Sec.~\ref{sec:cond-ensur-rip}); the projected back projection (PBP) method estimating signals from their quantized observations (Sec.~\ref{sec:PBP}); and finally, in Sec.~\ref{sec:dist-inner-prod}, a complete definition of the limited projection distortion (LPD) which, combined with the RIP, allows us to bound the PBP reconstruction error in Sec.~\ref{sec:PBP}.     

\subsection{Characterization of low-complexity sets}
\label{sec:low-complex-space}

In this work, our ability to estimate a signal $\bs x$ from the QCS model~\eqref{eq:Uniform-dithered-quantization} is developed on the hypothesis that this signal belongs to a ``low-complexity'' set $\cl K \subset \bb R^n$. We characterize the low-complexity nature of $\cl K$ according to two different, but connected ways. 

We first suppose that, for any radius $\eta > 0$, the set $\cl K$ --- or its restriction to the unit $\ell_2$-ball $\bb B^n$ if $\cl K$ is unbounded (\eg conic) --- can be \emph{covered} by a relatively small number of translated $\ell_2$-balls of radius $\eta$. In other words, we assume that $\cl K$ has a small {\em Kolmogorov entropy} $\cl H(\cl K, \eta)$ compared to $n$~\cite{KT1961}, with, for any bounded set $\cl S$, 
$$
\ts \cl H(\cl S,\eta) := \log\,\min\{\,|\cl G|:~\cl G \subset \cl S \subset \cl G+\eta\bb B^n\},
$$  
where the addition is the Minkowski sum between sets.

Second, we consider any set $\cl K$ (or $\cl K \cap \bb B^n$ if $\cl K$ is unbounded) having a small Gaussian mean width compared to $\sqrt n$. For any bounded set $\cl S \subset \bb R^n$~\cite{PV2013,CRPW2012}, this width, denoted $w(\cl S)$, is defined~by
$$
w(\cl S)\ :=\ \bb E \sup_{\bs u \in \cl S} |\scp{\bs g}{\bs u}|,\quad \bs g \sim \cl N(0,\Id_n). 
$$ 
We provide below examples of sets for which $w$ has a known upper bound. 
\medskip

In this paper, we restrict ourself to two specific low-complexity set categories displaying very different relationships between their Kolmogorov entropy and their Gaussian mean width.
\begin{itemize}
\item ({\em structured sets}) The low-complexity set $\cl K$ can be first \emph{structured}, meaning that $\cl K$ is a \emph{cone}, \ie $\alpha \cl K \subset \cl K$ for all $\alpha \geq 0$ (in particular, $\bs 0 \in \cl K$ and $\cl K$ is symmetric), and the Kolmogorov entropy $\cl H(\cl K \cap \bb B^n, \eta)$ of $\cl K \cap \bb B^n$ is bounded by a function of the complexity of ${\cl K \cap \bb B^n}$ multiplied by $\log 1/\eta$~\cite{CRPW2012,oymak2015near}. 

For instance, if $\cl K$ is a low-dimensional subspace of $\bb R^n$ (or a finite union of subspaces, as for the set of sparse vectors) or the set of low-rank matrices, $\cl H(\cl K \cap \bb B^n, \eta)$ is well controlled by standard covering arguments of~$\cl K \cap \bb B^n$~\cite{pisier1999volume}. In fact, for the set of sparse signals $\bs\Psi \spset{k}$ in an orthonormal basis (or a dictionary) $\bs \Psi$ of $\bb R^n$, and the set $\lrset{r}$ of rank-$r$ matrix of size $n_1\times n_2$, we have respectively
\begin{align*}
\cl H(\bs\Psi \spset{k} \cap \bb B^n, \eta)&\ts\ \lesssim\ k \log(n/k)\log(1 + 1/\eta)\quad \text{(see \eg~\cite{BDDW08})},\\
\cl H(\lrset{r} \cap \bb B^n, \eta)&\ts\ \lesssim\ r(n_1 + n_2) \log(1 + 1/\eta) \quad \text{(see \eg~\cite{CP11,CRPW2012})}.
\end{align*}
Interestingly, the terms independent of $\eta$ in these bounds also constrain the Gaussian mean width of the same sets, \ie  
$$
k\ \lesssim\ \ts w(\bs\Psi \spset{k} \cap \bb B^n)^2\ \lesssim\ k \log(n/k)\quad \text{and}\quad w(\lrset{r} \cap \bb B^{n_1\times n_2}_F)^2\ \asymp\ r(n_1 + n_2).
$$
In fact, this is verified for many other structured sets with low Kolmogorov entropy (see, \eg~\cite{CRPW2012,CP11,PV2013} and \cite[Table 1]{JC2016} for other examples). 

Consequently, inspired by this correspondence (also stressed in \cite{oymak2015near,JC2016}), we consider that a set $\cl K \subset \bb R^n$ is \emph{structured} if it is a cone and if the following bound holds: 
\begin{equation}
  \label{eq:structured-subset-KE-bound}
\ts \cl H(\cl K \cap \bb B^n, \eta)\ \lesssim\ w(\cl K \cap \bb B^n)^2 \log(1 + 1/\eta).  
\end{equation}

\item ({\em arbitrary, bounded and star-shaped sets}) The second category of low-complexity set is composed of arbitrary, bounded and star-shaped sets $\cl K \subset \bb R^n$. This includes bounded, convex and symmetric sets of $\bb R^n$, as those considered in the analysis of the PBP reconstruction error in Sec.~\ref{sec:PBP-convex-set} and Sec.~\ref{sec:specialcasePBP}. For these sets, $\cl H$ is not (in general) bounded as tightly as in~\eqref{eq:structured-subset-KE-bound}. Instead, we can invoke the looser bound induced by Sudakov inequality~\cite[Thm 3.18]{LedouxTalagrand91book}:
  \begin{equation}
    \label{eq:sudakov}
    \ts \cl H(\cl K, \eta)\ \lesssim\ \tinv{\eta^2}\,w(\cl K)^2.    
  \end{equation}
\end{itemize}
\medskip

We conclude this section by introducing three useful set constructions, built from the considered low-complexity sets and sharing their low-complexity nature. 
\begin{itemize}
\item ({\em set clipping}) The clipping of $\cl K$ is defined as
  \begin{equation}
    \label{eq:set-clipping}
    \clip{\cl K} :=
  \begin{cases}
    \cl K \cap \bb B^n&\text{if $\cl K$ is a cone},\\
    \cl K&\text{if $\cl K$ is bounded.}
  \end{cases}
\end{equation}
We sometimes simplify our developments by imposing $\|\clip{\cl K}\| \leq 1$. This thus means that $\cl K \subset \bb B^n$ if $\cl K$ is bounded.  

\item ({\em multiple of sets}) Given $s \in \bb N$, the $s\th$-\emph{multiple} of $\cl K$ is defined by 
$$
\ts \cl K^{\splus s} := \sum_{i=1}^s \cl K,
$$
where the sum is the Minkowski's sum of sets. Note that $\cl K^{\splus n} \subset n \cl K$ only if $\cl K$ is convex. We will use this set construction to characterize to which low-complexity set belongs the combination (\eg addition, or subtraction if $\cl K$ is symmetric) of a few elements of $\cl K$.
\item ({\em local set}) Given a radius $\eta > 0$, the $\eta$-local set of $\cl K$ corresponds to  
  \begin{equation}
    \label{eq:local-set}
    \cl K_{\loc}^{(\eta)} := (\cl K - \cl K) \cap \eta \bb B^n.    
  \end{equation}
  This set is important to characterize the Lipschitz continuity of a RIP matrix (see Sec.~\ref{sec:LPD-noisy-linear-map}).
\end{itemize}

\medskip

\subsection{Restricted isometry over low-complexity sets} 
\label{sec:cond-ensur-rip}

Another implicit assumption we make on the set $\cl K$, or on one of its multiples (see Sec.~\ref{sec:specialcasePBP}), is that it is compatible with the existence of random matrices satisfying \whp the restricted isometry property over $\cl K$, as defined below.
 
\begin{definition}[Restricted isometry property -- RIP]
\label{def:rip-def}
Given a distortion $\epsilon > 0$ and a bounded set $\cl K \subset \bb R^n$, a matrix $\tinv{\sqrt m} \bs \Phi \in \bb R^{m \times n}$ respects the restricted isometry property RIP$(\cl K, \epsilon)$ if
\begin{equation}
\label{eq:RIP}
\ts \big |\frac{1}{m}\|\bs\Phi \bs u\|^2 - \|\bs u\|^2 \big |\ \leq\ \epsilon \|\cl K\|^2,\quad\forall \bs u \in \cl K.  
\end{equation}
By extension, if $\cl K$ is a cone, then the RIP$(\cl K, \epsilon)$ amounts to
\begin{equation}
\label{eq:RIP-cone}
\ts \big |\frac{1}{m}\|\bs\Phi \bs u\|^2 - \|\bs u\|^2 \big |\ \leq\ \epsilon \|\bs u\|^2,\quad\forall \bs u \in \cl K.  
\end{equation}
\end{definition}

We clarify formally the link between the set $\cl K$ and the RIP of a sensing matrix by introducing the concept of \emph{RIP matrix distribution}. 
\begin{definition}[RIP matrix distribution]
\label{def:RIPGen}
Given a set $\cl K \subset \bb R^n$, a RIP matrix distribution $\DRIP(\cl K)$ for $\cl K$ is a probability distribution over $\bb R^{m \times n}$ such that:
\begin{itemize}
\item[\em (i)] for a distortion $\epsilon > 0$ and a failure probability $0<\zeta<1$, provided that  
\begin{equation}
  \label{eq:gen-RIP-lsc-cond}
  \ts m\ \gtrsim \epsilon^{-2}\, \|\clip{\cl K}\|^{-2}\,w(\clip{\cl K})^{2}\,\cl P_{\log}(m, n, 1/\epsilon, 1/\zeta),
\end{equation}
where $\cl P_{\log}$ is a \emph{polylogarithmic}\footnote{A polynomial of the logarithm of its arguments; we further assume $\cl P_{\log}$ has positive coefficients and $\cl P_{\log}\geq 1$.} function that depends on the distribution $\DRIP$, a random matrix $\tinv{\sqrt m}\bs\Phi \in \bb R^{m\times n}$ drawn from $\DRIP$ respects with probability exceeding $1-\zeta$ the restricted isometry property RIP$(\cl K, \epsilon)$ defined in Def.~\ref{def:rip-def};

\item[\em (ii)] at fixed $m, n \in \bb N$, the probability distribution $\DRIP(\cl K)$ (\eg sub-Gaussian, partial random Fourier) remains unchanged to generate RIP matrices over any multiples or any local sets of $\cl K$, \ie $\DRIP(\cl K) = \DRIP(\cl K^{+s}) = \DRIP(\cl K^{(\eta)}_{\loc})$ for any $s \in \bb N$ and any $\eta > 0$.
\end{itemize}
\end{definition}

In this definition, the point \emph{(ii)} imposes that, at identical matrix dimensions, the same probability distribution $\DRIP$ applies to randomly generate a RIP matrix over $\cl K$, its multiples or any of its local sets; however, \eqref{eq:gen-RIP-lsc-cond} has to be modified accordingly with the Gaussian mean width of these set variants, \ie to verify with which distortion $\epsilon > 0$ and probability $\tinv{\sqrt m}\bs \Phi \sim \DRIP$ satisfies the RIP over each of these sets. As reviewed below, the CS literature provides many RIP matrix distributions compatible with the points \emph{(i)} and \emph{(ii)} above, both for conic and bounded star-shaped sets, and allowing for dense or structured random sensing matrices (\eg with fast matrix-to-vector multiplication~\cite{FR2013}).

\paragraph{Sub-Gaussian RIP matrix distributions and sparse signals:} For sensing $k$-sparse signals in an orthonormal basis $\bs \Psi \in \bb R^{n \times n}$, \ie signals of $\cl K = \bs\Psi \spset{k}$, many studies have proved that, given a distortion $\epsilon >0$ and as soon as 
\begin{equation}
  \label{eq:RIP-cond-subGauss-sparse}
\ts m \gtrsim \epsilon^{-2} k \log\frac{n}{\epsilon k},  
\end{equation}
a sub-Gaussian random matrix $\tinv{\sqrt m}\bs \Phi \in \bb R^{m\times n}$, \ie with random entries \iid from a centered sub-Gaussian distribution with unit variance (\eg Gaussian, Bernoulli)~\cite{CT06,MPT2008,BDDW08,FR2013}, satisfies the RIP$(\bs \Psi\spset{{2k}},\epsilon)$ with probability exceeding $1 - C \exp(-c \epsilon^2 m)$. Consequently, this proves also that we can lower bound this probability by $1-\zeta$ (with $\zeta \in (0,1)$) if
\begin{equation}
  \label{eq:req-m-subG-sparse}
  \ts m \gtrsim \epsilon^{-2} k \log(n/\epsilon k) \log(1/\zeta),  
\end{equation}
since then $m\epsilon^2 \gtrsim \log(1/\zeta)$ and $1 - C \exp(-c \epsilon^2 m) \geq 1 - \zeta$, for appropriate (hidden) multiplicative constants. Note that \eqref{eq:req-m-subG-sparse} is satisfied if \eqref{eq:gen-RIP-lsc-cond} holds with $\cl K = \bs \Psi\spset{{2k}}$ and $\cl P_{\log} = \log(n/\epsilon)\log(1/\zeta)$, since, as shown in the next paragraph, $w(\clip{\cl K})^2 \geq k - 1$. As explained below, sub-Gaussian random matrix can embed more general sets than the set of sparse signals. 

\paragraph{Structured RIP distributions and sparse signals:} Still in the context of sparse signals embedding, the RIP can also be proved for certain structured random matrix constructions, \ie with  reduced memory storage and small matrix-to-vector multiplication complexity. This is the case of partial random orthonormal matrix or bounded orthonormal systems (BOS) constructions~\cite{CT06,Rau10}. For instance, if
$\bs U \in \bb R^{n \times n}$ is an orthonormal basis such as the
discrete cosine transform (DCT) or the Hadamard transform\footnote{Or the discrete Fourier transform provided the results of this work are extended to the
  complex domain.}, and if the mutual
coherence $\mu(\bs U, \bs \Psi) := \sqrt n\,\max_{ij}|(\bs
U^\top \bs \Psi)_{ij}| \in [1, \sqrt n]$ is small, then, provided that\footnote{This result is an easy rewriting of~\cite[Theorem 8.1]{Rau10} for incoherent discrete bases of $\bb R^{n \times n}$.}
\begin{equation}
\label{eq:RIP-cond-prom-sparse}
m \gtrsim \mu^2 \epsilon^{-2} k \,(\log k)^2 \log n \log m \log 1/\zeta,  
\end{equation}
the matrix $\bs \Phi \in \bb R^{m \times n}$ formed by picking $m$
rows uniformly at random in the $n$ rows of the renormalized matrix $\sqrt n\,\bs U$ respects the RIP$(\bs \Psi \spset{k},\epsilon)$ with probability exceeding $1-\zeta$. Moreover, the same implicit matrix distribution $\DRIP$ applies for generating RIP matrices over multiples of $\bs \Psi \spset{k}$, \ie of the form $\bs \Psi \spset{sk}$ with $s \in \bb N$, or over any of its local sets since $\bs \Psi \spset{k}$ is conic.

Note that since $\ts k-1 \leq (\bb E \|\bs g_{\cl T}\|)^2 = w(\bs \Psi \spset{{\cl T}} \cap \bb B^n)^2 \leq w(\bs \Psi \spset{k}\cap \bb B^n)^2 \lesssim k \log n/k$, with the $k$-dimensional subspace $\spset{{\cl T}} := \{\bs u \in \bb R^n: \supp \bs u \subset {\cl T}\}$ associated with an arbitrary $k$-length support ${\cl T} \subset [n]$ and $\bs g \sim \cl N^n(0,1)$,~\eqref{eq:RIP-cond-prom-sparse} is verified if 
\begin{equation}
\label{eq:RIP-cond-prom-sparse-wK}
m \gtrsim \mu^2 \epsilon^{-2} w(\bs \Psi \spset{k}\cap \bb B^n)^2 \,(\log n)^4 \log 1/\zeta,
\end{equation}
which is compatible with~\eqref{eq:gen-RIP-lsc-cond}. 

\paragraph{RIP distributions for low-complexity signals:} For more general low-complexity sets $\cl K \subset \bb R^n$ (\eg effective sparsity and low-rank models; see \cite{Ayaz16,CRPW2012,GV12,RFP10}, \cite[Table 1]{JC2016}), the following known result shows that, from a requirement on $m$ following the template \eqref{eq:gen-RIP-lsc-cond}, one can ensure \whp the RIP of a sub-Gaussian random matrix $\bs \Phi$ over $\cl K$ (see also \cite{MPT2008,KM2005,D2015}). 
\begin{proposition}[Adapted from {\cite[Thm 1.4]{LMPV2017}}]
\label{prop:RIP-bounded-cvxset}
Given $\cl K \subset \bb R^n$, a distortion $\epsilon >0$ and a probability of failure $0< \zeta <1$, if
\begin{equation}
  \label{eq:RIP-bounded-cvxset-condm}
  \ts m\ \gtrsim\ \epsilon^{-2} \frac{w(\clip{\cl K})^2}{\|\clip{\cl K}\|^2}\log(1/\zeta),  
\end{equation}
and if $\bs \Phi \in \bb R^{m \times n}$ is a random matrix whose entries are \iid from a centered sub-Gaussian distribution with unit variance\footnote{Whose sub-Gaussian norm is hidden in the multiplicative constant of~\eqref{eq:RIP-bounded-cvxset-condm}.}, then $\tinv{\sqrt m} \bs \Phi$ respects the RIP$(\clip{\cl K}, \epsilon)$ (\ie the RIP$(\cl K, \epsilon)$ if $\cl K$ is conic) with probability exceeding $1-\zeta$.
\end{proposition}
The proof of this easy adaptation of {\cite[Thm 1.4]{LMPV2017}} is left to the reader. In other words, a sub-Gaussian matrix distribution is a RIP matrix distribution for all low-complexity sets with bounded Gaussian mean width, as specified in Def.~\ref{def:RIPGen}. Note that the point \emph{(ii)} in this definition obviously holds since \eqref{eq:RIP-bounded-cvxset-condm} already depends on the Gaussian mean width of $\clip{\cl K}$, which authorizes embedding of the multiples and local sets of $\cl K$. This applies also for the RIP matrix distributions below. 
   
\paragraph{SORS distributions and low-complexity signals:} Recently, \emph{subsampled orthogonal with random sign} (SORS) sensing matrices have been proved to respect \whp the RIP~\cite{oymak2015sors}.  This construction combines a partial random orthonormal matrix construction (see above) with a pre-modulation matrix set with random signs, \ie $\bs \Phi = \sqrt n\, \bs R_{\Omega} \bs U \bs D$ where $\bs R_{\Omega} \in \{0,1\}^{|\Omega|\times n}$ is the selection matrix such that $\bs R_{\Omega} \bs u = \bs u_{\Omega}$ for some $\Omega \subset [n]$, $\bs U \in \bb R^{n \times n}$ is an orthonormal matrix with $\max_{ij}|U_{ij}| = O(1/\sqrt n)$, and $\bs D$ is a $n \times n$ random diagonal matrix with the diagonal entries \iid $\pm 1$ with equal probability. 

\begin{proposition}[Adapted from {\cite[Thm 3.3]{oymak2015sors}}]
\label{prop:SORS}
Given $\cl K \subset \bb R^n$, a distortion $\epsilon >0$ and a probability of failure $0<\zeta<1$, if
\begin{equation}
  \label{eq:SORS}
  \ts m\ \gtrsim\ \epsilon^{-2} \max(1,\frac{w(\clip{\cl K})^2}{\|\clip{\cl K}\|^2}) (\log n)^4 \log(2/\zeta)^2, 
\end{equation}
and if $\bs \Phi \in \bb R^{m \times n}$ is a SORS random matrix, then, $\tinv{\sqrt m} \bs \Phi$ respects the RIP$(\clip{\cl K}, \min(\epsilon, \epsilon^2))$, with probability exceeding $1-\zeta$.
\end{proposition}
The proof of this easy adaptation of \cite[Thm 3.3]{oymak2015sors} is left to the reader. We thus see that the requirement on $m$ in \eqref{eq:SORS} follows the template \eqref{eq:gen-RIP-lsc-cond} as soon as $w(\clip{\cl K})^2 \gtrsim \|\clip{\cl K}\|^2$. This easily holds if $\cl K$ is either conic or star-shaped; there exists then a $\bs v \in \clip{\cl K}$ such that $\|\bs v\| = \|\clip{\cl K}\|$, and $w(\clip{\cl K}) \geq \bb E |\scp{\bs g}{\bs v}| = \sqrt{2/\pi} \|\clip{\cl K}\|$ for $\bs g \sim \cl N^n(0,1)$.
\bigskip

In conclusion, according to the requirements~\eqref{eq:RIP-cond-subGauss-sparse},~\eqref{eq:RIP-cond-prom-sparse-wK},~\eqref{eq:RIP-bounded-cvxset-condm} and~\eqref{eq:SORS}, sub-Gaussian random matrix, partial random orthonormal matrix, BOS, and SORS (for $0<\epsilon\leq 1$) are examples of RIP matrix distributions, as defined in Def.~\ref{def:RIPGen}.

\subsection{Projected back projection}
\label{sec:PBP}

As announced in the Introduction, the standpoint of this work is to show the compatibility of a RIP matrix $\bs \Phi$ with the dithered QCS model~\eqref{eq:Uniform-dithered-quantization}, provided that the dither is random and uniform, through the possibility to estimate $\bs x$ via the projected back projection (PBP) onto $\cl K$ of the quantized observations
$
\ts \bs y = \cl Q(\bs \Phi \bs x + \bs \xi).
$ 

More generally, given the distorted CS model 
\begin{equation}
  \label{eq:distorted-CS-model}
  \bs y = \Dmap(\bs x),\quad \bs x \in \cl K \cap \bb B^n,
\end{equation}
associated with the general (random) mapping $\Dmap: \bb R^n \to \bb R^m$ (\eg with $\Dmap \equiv \Amap$), the PBP of $\bs y$ onto $\cl K$ is mathematically defined~by 
\begin{equation}
  \label{eq:PBP-estimate}
\ts  \hat{\bs x}\ :=\ \cl P_{\cl K}(\inv{m} \bs \Phi^\top \bs y), 
\end{equation}
where $\cl P_{\cl K}$ is the (minimal distance) projector on $\cl K$, \ie
\begin{equation}
  \label{eq:proj-def}
\ts \displaystyle \cl P_{\cl K}(\bs z)\ \in\ \arg\min_{\bs u \in \cl K} \|\bs z - \bs u\|.  
\end{equation}
Throughout this work, we assume that $\cl P_{\cl K}$ can be computed, \ie in polynomial complexity with respect to $m$ and $n$. For instance, if $\cl K = \spset{k}$, $\cl P_{\cl K}$ is the standard best $k$-term hard thresholding operator, and if $\cl K$ is convex and bounded, $\cl P_{\cl K}$ is the orthogonal projection onto this set. There exist of course sets $\cl K$ for which $\cl P_{\cl K}$ is intractable or NP-hard, \eg for the set of sparse vectors in a $n \times d$ dictionary with $d > n$ (for which $\cl P_{\cl K}$ is equivalent to an $\ell_0$-constrained Lasso problem \cite{natarajan1995sparse}), and for the set of \emph{cosparse} vectors~\cite[Theorem 1]{tillmann2014projection}.

\begin{remark}
Hereafter, the analysis of the PBP method is divided into two categories: \textit{uniform estimation}, \ie with high probability on the generation of $\Dmap$ (\eg $\Dmap \equiv \Amap$, the source of randomness is then the dither; see \eqref{eq:Uniform-dithered-quantization}), \emph{all} signals in the set $\cl K$ can be estimated using the same mapping $\Dmap$; and \emph{non-uniform estimation} (or \textit{fixed signal estimation}) where $\Dmap$ is randomly generated for each observed signal. 
\end{remark}

\subsection{Limited projection distortion}
\label{sec:dist-inner-prod}

We already sketched at the end of the Introduction that a crucial element of our analysis is the combination of the RIP of $\bs \Phi$ with another property jointly verified by $(\bs \Phi, \bs \xi)$, or equivalently by the quantized random mapping $\Amap$ defined in~\eqref{eq:Uniform-dithered-quantization}. As will be clear later, this property, the (local) \emph{limited projection distortion} (or (L)LPD) and the RIP allow us to bound the reconstruction error of the PBP. We define it as follows for a general mapping $\Dmap$.

\begin{definition}[Limited projection distortion -- LPD]
\label{def:LPD} Given a matrix $\bs \Phi\in \bb R^{m\times n}$ and a distortion $\nu > 0$, we say that
a general mapping $\Dmap:\bb R^n \to \bb R^m$ respects the limited projection distortion property over a set
$\cl K\subset \bb R^n$ observed by $\bs \Phi$, or LPD$(\cl K, \bs \Phi, \nu)$, if 
\begin{equation}\label{eq:LPD}
\ts \inv{m}\,|\langle \Dmap(\bs u), \bs \Phi \bs v\rangle-\langle \bs
\Phi\bs u,\bs \Phi\bs v\rangle|\ \leq \nu \|\clip{\cl K}\|,\quad \forall \bs u, \bs v \in \clip{\cl K}.
\end{equation}
In particular, when $\bs u$ is fixed in~\eqref{eq:LPD}, we say that
$\Dmap$ respects the \emph{local} limited projection distortion on $\bs u$, or L-LPD$(\cl K, \bs \Phi, \bs u, \nu)$.
\end{definition}

As explained in Sec.~\ref{sec:motivationSOA}, the LPD property was (implicitly) introduced in~\cite{PV2016} in the special case where $\Dmap = f \circ \bs\Phi$ and $f$ is a non-linear function applied componentwise on the image of $\bs \Phi$. The LPD is also connected to the SPE introduced in~\cite{PV2013} for the specific case of a one-bit sign quantizer if we combine the LPD property with the RIP of $\tinv{\sqrt m}\bs \Phi$ in order to approximate $\tinv{m}\scp{\bs \Phi \bs u}{\bs \Phi \bs v}$ by $\scp{\bs u}{\bs v}$ in~\eqref{eq:LPD} (see Lemma~\ref{lem:scp-embed}). This literature was however restricted to the analysis of Gaussian random matrices. 

\begin{remark}\label{rem:delta-0}
If $\Dmap$ is the quantized random mapping $\Amap$ introduced in~\eqref{eq:Uniform-dithered-quantization} with the random dither $\bs \xi \sim \cl U^m([0,\delta])$, an arbitrary
low-distortion $\nu>0$ is expected in \eqref{eq:LPD} for large values of $m$ since, in
expectation, $\bb E_{\bs \xi} \langle \Amap(\bs u), \bs \Phi \bs
v\rangle-\langle \bs \Phi\bs u,\bs \Phi\bs v\rangle = 0$ from
Lemma~\ref{lem1} (see Sec.~\ref{sec:LPD-noisy-linear-map} and Sec.~\ref{sec:LPD}). Note also that for such a random dither, if $\delta$
tends to 0, then the quantizer $\Amap$ tends to the identity operator and $|\langle \Amap(\bs u), \bs \Phi \bs v\rangle-\langle \bs \Phi\bs u,\bs \Phi\bs v\rangle|$ must vanish. In fact, by Cauchy-Schwarz and the triangular inequality, this is sustained by the deterministic bound 
\begin{align}
\ts |\langle \Amap(\bs u), \bs \Phi \bs v\rangle-\langle \bs \Phi\bs u,\bs \Phi\bs v\rangle| = |\langle \Amap(\bs u) - \bs \Phi \bs u, \bs \Phi \bs v\rangle|&\leq \|\bs \Phi \bs v\|\, \|\Amap(\bs u) - \bs \Phi \bs u\| \leq \delta \sqrt m \|\bs \Phi \bs v\|,\label{eq:rem-delta-0}  
\end{align}
obtained by observing that $|\lfloor \lambda + u \rfloor - \lambda| \leq 1$ for all $\lambda \in \bb R$ and $u \in [0, 1]$.
\end{remark}

As developed in Sec.~\ref{sec:LPD-noisy-linear-map}, it is easy to prove the L-LPD of $\Dmap$ if this mapping is a linear mapping~$\bs \Phi$ corrupted by an additive noise $\bs \rho$ composed of \iid sub-Gaussian random components, \ie $\Dmap(\bs u) := \bs \Phi \bs u + \bs \rho$. Standard tools of measure concentration theory then show that $\tinv{m}\scp{\bs \rho}{\bs \Phi \bs v}$ concentrates, \whp, around 0. As clarified later, such a scenario includes the case $\Dmap \equiv \Amap$ since, given a fixed $\bs u$, the $m$ \iid \rvs $\rho_i := \cl Q((\bs \Phi \bs u)_i + \xi_i)-(\bs \Phi \bs u)_i$ are bounded and thus sub-Gaussian. However, proving the uniform LPD property of $\Amap$ requires to probe its geometrical nature. We need in particular to control the impact of the discontinuities introduced by $\cl Q$ in $\bs \rho$ (see Sec.~\ref{sec:LPD}).         
\medskip

We observe in \eqref{eq:LPD} that the (L)LPD characterizes the proximity of scalar products between distorted and undistorted random observations in the compressed domain $\bb R^m$. In order to assess how $\tinv{m}\scp{\Dmap(\bs u)}{\bs\Phi\bs v}$ approximates $\scp{\bs u}{\bs v}$ in the case where $\tinv{\sqrt m} \bs \Phi$ respects the RIP, we can consider this standard lemma from the CS literature (see \eg 
\cite{FR2013}).
\begin{lemma}
\label{lem:scp-embed}
Given a bounded, convex and symmetric subset $\cl K' \subset \bb R^n$ with $\|\cl K'\|=1$, if $\tinv{\sqrt m}\bs\Phi$ is RIP$(\cl K', \epsilon)$ with $\epsilon>0$, then 
\begin{equation}
  \label{eq:embed-scalar-product}
\ts  |\frac{1}{m}\scp{\bs \Phi \bs u}{\bs \Phi \bs v} - \scp{\bs u}{\bs v}|\leq 2\epsilon,\quad \forall\bs u, \bs v \in \cl K'.
\end{equation}
If $\cl K'$ is conic and if $\tinv{\sqrt m}\bs\Phi$ is RIP$(\cl K', \epsilon)$, we have
\begin{equation}
  \label{eq:embed-scalar-product-cone}
\ts  |\frac{1}{m}\scp{\bs \Phi \bs u}{\bs \Phi \bs v} - \scp{\bs u}{\bs v}|\leq \epsilon\,\|\bs u\| \|\bs v\|, \quad \forall\bs u, \bs v: \bs u \pm \bs v \in \cl K'.
\end{equation}
\end{lemma}
\begin{proof}
Note that since $\cl K'$ is convex and symmetric, $\pm \tinv{2}\cl K' \pm \tinv{2}\cl K' \subset \cl K'$. Given $\bs u, \bs v \in \cl K'$, if $\tinv{\sqrt m}\bs\Phi$ is RIP$(\cl K', \epsilon)$ with $\epsilon>0$, then, from the polarization identity, the fact that $\tinv{2}(\bs u \pm \bs v) \in \cl K'$ and from~\eqref{eq:RIP},
\begin{align*}
&\ts \frac{1}{m} \scp{\bs \Phi \bs u}{\bs \Phi \bs v} = \frac{1}{m}\|\bs \Phi (\frac{\bs u + \bs v}{2})\|^2 - \frac{1}{m}\|\bs \Phi (\frac{\bs u - \bs v}{2})\|^2 \leq \tinv{4}\|\bs u + \bs v\|^2 - \tinv{4}\|\bs u - \bs v\|^2 + 2\epsilon = \scp{\bs u}{\bs v}  + 2\epsilon.
\end{align*}
The lower bound is obtained similarly. If $\cl K'$ is conic and if $\tinv{\sqrt m}\bs\Phi$ is RIP$(\cl K', \epsilon)$, then, for all unit norm $\bs u, \bs v \in \bb R^n$ such that $\bs u \pm \bs v \in \cl K'$, \eqref{eq:RIP-cone}~provides, 
\begin{align*}
&\ts \frac{1}{m} \scp{\bs \Phi \bs u}{\bs \Phi \bs v} = \frac{1}{m}\|\bs \Phi (\frac{\bs u + \bs v}{2})\|^2 - \frac{1}{m}\|\bs \Phi (\frac{\bs u - \bs v}{2})\|^2\\
&\ts \leq \tinv{4}\|\bs u + \bs v\|^2 - \tinv{4}\|\bs u - \bs v\|^2 + \tinv{4}\epsilon\, (\|\bs u + \bs v\|^2 + \|\bs u - \bs v\|^2) = \scp{\bs u}{\bs v}  + \epsilon,
\end{align*}
with a similar development for the lower bound. A simple rescaling argument provides~\eqref{eq:embed-scalar-product-cone}.
\end{proof}

Therefore, applying the triangular identity, it is easy to verify the following corollary. 
\begin{corollary}
\label{cor:RIP-L-LPD}
Given a bounded, convex and symmetric subset $\cl K' \subset \bb R^n$ with $\|\cl K'\|=1$, if
$\tinv{\sqrt m}\bs \Phi$~respects the RIP$(\cl K',\epsilon)$ and
$\Dmap$ verifies the LPD$(\cl K', \bs\Phi, \nu)$ for $\epsilon, \nu>0$, then 
\begin{equation}
  \label{eq:RIP-L-LPD}
  \ts \,|\inv{m} \langle \Dmap(\bs u), \bs \Phi \bs v\rangle- \langle
  \bs u,\bs v\rangle|\ \leq\ 2\epsilon + \nu,\quad \forall\bs u, \bs v \in \cl K'.
\end{equation}
If $\cl K'$ is conic,  $\tinv{\sqrt m}\bs\Phi$ respects the RIP$(\cl K', \epsilon)$, and $\Dmap$ verifies the LPD$(\cl K', \bs\Phi, \nu)$ for $\epsilon, \nu>0$, then
\begin{equation}
  \label{eq:RIP-L-LPD-cone}
  \ts \,|\inv{m} \langle \Dmap(\bs u), \bs \Phi \bs v\rangle- \langle
  \bs u,\bs v\rangle|\ \leq\ \epsilon + \nu,\quad \forall\bs u, \bs v \in \cl K' \cap \bb B^n:  \bs u \pm \bs v \in \cl K'.
\end{equation}
The same observations hold if $\bs u$ is fixed when the L-LPD is invoked instead of the LPD.
\end{corollary}
Note that this corollary amounts to Lemma~\ref{lem:scp-embed} if $\Dmap$ is
identified with $\bs \Phi$, in which case $\nu=0$. 

\section{PBP reconstruction error in distorted CS}
\label{sec:PBP-gen-error}

We here provide a general analysis of the reconstruction error of the PBP method for 
the general \emph{distorted} CS model \eqref{eq:distorted-CS-model}. While the results of this section are later applied to the quantized, random mapping $\Amap$ of~\eqref{eq:Uniform-dithered-quantization} from Sec.~\ref{sec:LPD-noisy-linear-map}, this section is thus valid for any distorted sensing model $\Dmap: \bb R^n \to \bb R^m$ that meets the (L)LPD property associated with a certain RIP matrix $\tinv{\sqrt m}\bs \Phi$. Hereafter, we analyze the cases where the low-complexity signal set $\cl K$ is a union of low-dimensional subspaces, the set of low-rank matrices, or a convex and symmetric subset of $\bb B^n$.

\subsection{Union of low-dimensional subspaces}
\label{sec:PBP-union-subspaces}

We here assume that $\cl K:= \cup_{i\in[K]} \cl K_i$ is a union of $K$~low-dimensional subspaces $\cl K_i \subset \bb R^n$, \ie a ULS model. This model encompasses, \eg sparse signals in an orthonormal basis or in a dictionary~\cite{RSV08,CRPW2012}, co-sparse signal models~\cite{nam2013cosparse}, group-sparse signals~\cite{Ayaz16} and model-based sparsity~\cite{baraniuk2010model}.

The next theorem states that the PBP reconstruction error is bounded by the addition of the distortion induced by the RIP of $\bs \Phi$ (as in CS) and the one provided by (L)LPD of $\Dmap$.
\begin{theorem}[PBP for ULS]
\label{thm:PBP-Union-Subspace}
Given the ULS model $\cl K := \cup_{i\in[K]} \cl K_i \subset \bb R^n$ and  two distortions $\epsilon, \nu > 0$, if $\tinv{\sqrt m}\bs \Phi \in \bb R^{m \times n}$ respects the RIP$(\cl K - \cl K, \epsilon)$ and if the mapping $\Dmap: \bb R^n \to \bb R^m$ satisfies the LPD$(\cl K - \cl K, \bs \Phi, \nu)$, then, for all $\bs x \in \cl K \cap \bb B^n$, the estimate $\hat{\bs x}$ obtained by the PBP of $\bs y=\Dmap(\bs x)$ onto $\cl K$ satisfies 
$$
\ts \|\bs x-\hat{\bs x}\| \leq 2(\epsilon+\nu).
$$
Moreover, if $\bs x$ is fixed, then the same result holds if $\Dmap$ respects the L-LPD$(\cl K - \cl K, \bs \Phi, \bs x, \nu)$.
\end{theorem}

\begin{proof}
We generalizes the proof sketch given at the end of the Introduction for $\cl K = \spset{k}$. Since $\bs x \in (\cup_{i\in[K]} \cl K_i)\cap\bb B^n$ and $\hat{\bs x} \in \cl K$, there must exist two subspaces $\cl K_x := \cl K_{i}$ and $\hat{\cl K} := \cl K_{i'}$,  for some $i,i' \in [K]$ such that $\bs x \in \cl K_x$ and $\hat{\bs x} \in \hat{\cl K}$. Let us define $\bs a=\frac{1}{m}\bs \Phi^\top\bs y$, the subspace $\bar{\cl K}:=\cl K_x + \hat{\cl K}$ and the orthogonal complement $\bar{\cl K}^\bot$ of $\bar{\cl K}$. We can always decompose $\bs a$ as $\bs a= \bar{\bs a} + \bar{\bs a}^\bot$ with $\bar{\bs a} := \cl P_{\bar{\cl K}}(\bs a)$ and $\bar{\bs a}^\bot := \cl P_{\bar{\cl K}^\bot}(\bs a)$, with the projector $\cl P_{(\cdot)}$ defined in \eqref{eq:proj-def}. Since $\hat{\bs x}= \cl P_{\cl K}(\bs a)\in \cl K$, we have 
$$
\|\hat{\bs x}-\bs a\|^2 = \|\hat{\bs x}-\bar{\bs a} - \bar{\bs a}^\bot\|^2 \leq \|\bs x-\bs a\|^2 = \|\bs x-\bar{\bs a} - \bar{\bs a}^\bot\|^2. 
$$ 
Moreover, since both $\hat{\bs x}-\bar{\bs a}$ and $\bs x-\bar{\bs a}$ belong to $\bar{\cl K}$, the last inequality is equivalent to
$
\|\hat{\bs x}-\bar{\bs a}\|^2 + \|\bar{\bs a}^\bot\|^2 \leq \|\bs x-\bar{\bs a}\|^2 + \|\bar{\bs a}^\bot\|^2, 
$
which implies $\|\hat{\bs x}-\bar{\bs a}\|\leq\|\bs x-\bar{\bs a}\|$. Consequently, the triangular inequality gives
$$
\ts \|\bs x -\hat{\bs x}\| \leq \|\bs x -\bar{\bs a}\| + \|\hat{\bs x} - \bar{\bs a}\| \leq 2\|\bs x -\bar{\bs a}\|.
$$
From the assumptions of the theorem, $\bs \Phi$ and $\Dmap$ respect the RIP$(\cl K - \cl K, \epsilon)$ and the LPD$(\cl K - \cl K, \bs \Phi, \nu)$, respectively. Therefore, since $\bs x \in \overline{\cl K} \cap \bb B^n$ and $\bs x \pm \bs u \in \overline{\cl K} \subset \cl K - \cl K$ for all $\bs u \in \bar{\cl K} \cap \bb B^n$, 
Cor.~\ref{cor:RIP-L-LPD} with the cone $\cl K' = \cl K - \cl K$ gives 
\begin{align*}
\|\bs x -\hat{\bs x}\|&\ts \leq\ 2\,\|\bs x -\bar{\bs a}\|\ = \ 2\,\sup_{\bs u\in \bb B^{n}}\scp{\bs u}{\bs x-\bar{\bs a}}\ =\ 2\,\sup_{\bs u\in \bar{\cl K}\cap \bb B^n} \scp{\bs u}{\bs x-\bs a}\\
&\ts  =\ 2\,\sup_{\bs u\in \bar{\cl K} \cap \bb B^n} \big(\scp{\bs u}{\bs x} - \scp{\bs u}{\frac{1}{m}\bs \Phi^\top\bs y}\big)\\
&\ts  =\ 2\,\sup_{\bs u\in \bar{\cl K} \cap \bb B^n} \big(\scp{\bs u}{\bs x} - \frac{1}{m}\scp{\bs \Phi\bs u}{\Dmap(\bs x)}\big)\\
&\leq\ 2(\epsilon+\nu),
\end{align*}
which gives the result. Moreover, if $\bs x$ is fixed, only the L-LPD$(\cl K - \cl K, \bs \Phi, \bs x, \nu)$ is required.
\end{proof}

\subsection{Bounded low-rank matrices}
\label{sec:PBP-low-rank}

The (L)LPD and the RIP also allow to bound the reconstruction error of PBP in the estimation of low-rank matrices observed by the distorted CS model~\eqref{eq:distorted-CS-model}, up to an easy adaptation of this model to matrix sensing.  In fact, given a bounded rank-$r$ matrix $\bs X \in \cl K \cap \bb B_F^{n_1\times n_2}$, with 
$$
\ts \cl K := \lrset{r} = \{\bs Z \in \bb R^{n_1 \times n_2}: \rank \bs Z \leq r\},\quad r\leq \min(n_1,n_2),
$$
and introducing the Frobenius ball $\bb B_F^{n_1\times n_2} := \{\bs Z \in \bb R^{n_1 \times n_2}: \|\bs Z\|_F^2 := \tr(\bs Z^\top \bs Z) \leq 1\}$, the (linear) CS model reads $\bs y=\bs \Theta (\bs X)$, where $\bs \Theta:\bb R^{n_1\times n_2} \to \bb R^m$ is the linear measurement operator defined by $m$ scalar products of an $(n_1\times n_2)$-matrix with $m$ pre-defined $n_1\times n_2$ matrices $\{\bs \Phi_1,\,\cdots, \bs \Phi_m\}$,~\ie 
$$
\ts [\bs \Theta (\bs X)]_i=\scp {\bs \Phi_i}{\bs X} := \tr(\bs \Phi_i^\top \bs X).
$$
Correspondingly, the adjoint of $\bs \Theta$ is $\bs \Theta^*: \bs w\in \bb R^m \mapsto \bs \Theta^*(\bs w) = \sum^m_{i=1} w_i \bs \Phi_i \in \bb R^{n_1\times n_2}$. 

Equivalently, by vectorizing any matrix $\bs Z \in \bb R^{n_1 \times n_2}$ into its vector representation $\ve(\bs Z) \in \bb R^n$ with $n=n_1n_2$, \ie stacking up all its columns on top of one another, the CS model can be rewritten as $\bs \Phi \ve(\bs X)$, where $\bs \Phi :=[\ve(\bs \Phi_1),\cdots,\ve(\bs \Phi_m)]^\top\in \bb R^{m\times n}$, \ie $\bb R^{n_1 \times n_2}$ and $\bb B_F^{n_1 \times n_2}$ are thus identified with $\bb R^n$ and $\bb B^n$, respectively. Moreover, requesting $\bs \Theta$ to satisfy the RIP over some subset $\cl K \subset \bb R^{n_1\times n_2}$ is thus equivalent to ask $\bs \Phi$ to respect it over $\ve(\cl K) \subset \bb R^{n}$ as defined in~\eqref{eq:RIP}, with also $\bs \Phi^*\bs w = \ve(\bs \Theta^*\bs w)$ for $\bs w \in \bb R^m$. 

For simplicity, we thus consider that the distorted model~\eqref{eq:distorted-CS-model} is also defined over a vectorization of the matrix domain, \ie for a mapping $\Dmap: \bb R^n \to \bb R^m$, and the definition of the (L)LPD is thus considered in the same sense, \ie for the sensing matrix $\bs \Phi$ related to the vectorized form of $\bs \Theta$. 
\medskip

Before establishing the main result of this section, let us specify two useful properties for our developments. 
First, concerning $\bs \Phi$, since $\lrset{s} \pm \lrset{s} = \lrset{2s}$, the RIP$(\lrset{2s}, \epsilon)$ amounts to the RIP$(\lrset{s} - \lrset{s}, \epsilon)$ for any $s>0$.    

Second, the projector $\cl P_r := \cl P_{\lrset{r}}$ of any matrix $\bs Z \in \bb R^{n_1\times n_2}$ onto the set of rank-$r$ matrices $\cl K = \lrset{r}$ is given by~\cite{Fa02,RFP10}
\begin{equation}
\label{RankProj}
\ts \cl P_r(\bs Z) :=\arg\min_{\bs U\in \lrset{r}}\|\bs U-\bs Z\|_F = \bs U \cl M_r(\bs \Sigma) \bs V^\top.
\end{equation}
In~\eqref{RankProj}, $\cl M_{r}(\bs D)$ is the $r$-thresholding operator setting all but the $r$-first diagonal entries of $\bs D$ to zero, and $\bs U \bs \Sigma \bs V^\top$ is the singular value decomposition (SVD) of $\bs Z$, where $\bs U \in \bb R^{n_1\times n_1}$ and $\bs V \in \bb R^{n_2\times n_2}$ are the unitary matrices formed by the left and right singular vectors of $\bs Z$, respectively, and $\bs \Sigma \in \bb R^{n_1\times n_2}$ is the (rectangular) diagonal matrix formed by the (decreasing) singular values $\{\sigma_i: 1\leq i\leq \min(n_1,n_2)\}$ of $\bs Z$ (\ie $\Sigma_{ij} = \sigma_i \delta_{ij}$, $\sigma_i \geq \sigma_{i+1}$).  In other words, (\ref{RankProj}) provides the best rank-$r$ matrix approximation of $\bs Z$ in the Frobenius norm. 

As in the previous section, we are now ready to leverage both the (L)LPD of a mapping $\Dmap$ and the RIP of $\tinv{\sqrt m}\bs \Phi$ for proving that PBP provides a controllable error distortion to estimate a matrix~$\bs X \in \lrset{r} \cap \bb B_F^{n_1\times n_2}$. The proof is similar to the one of Theorem~\ref{thm:PBP-Union-Subspace} once we correctly identify a common subspace for both $\bs X$ and its PBP estimate. 

\begin{theorem}[PBP for bounded low-rank matrices]
\label{thm:PBP-Low-Rank} 
Given the low-rank model $\cl K := \lrset{r} \subset \bb R^{n_1\times n_2}$ with $0<r<\min(n_1,n_2)$, and two distortions $\epsilon, \nu > 0$, if $\tinv{\sqrt m}\bs \Phi: \bb R^{n} \to \bb R^{m}$ (\ie $\tinv{\sqrt m}\bs \Theta: \bb R^{n_1 \times n_2} \to \bb R^{m}$) respects the RIP$(\lrset{2r}, \epsilon)$ and if the mapping $\Dmap: \bb R^n \to \bb R^m$ satisfies the LPD$(\lrset{2r}, \bs \Phi, \nu)$, then, for all $\bs X \in \cl K \cap \bb B^{n_1\times n_2}_F$, the estimate $\hat{\bs X}$ obtained by the PBP of $\bs y$ on $\lrset{r}$, \ie 
$$
\ts \hat{\bs X} := \cl P_{r}(\frac{1}{m} \bs \Theta^*(\bs y)),
$$
satisfies 
$$
\ts \|\bs X-\hat{\bs X}\|_F \leq 2(\epsilon + \nu).
$$
Moreover, if $\bs X$ is fixed, then the same result holds if $\Dmap$ respects the L-LPD$(\lrset{2r}, \bs \Phi, \bs X, \nu)$.
\end{theorem}

\begin{proof}
Let $\bs U \bs \Sigma \bs V^\top$ and $\bs U' \bs \Sigma' \bs V'^\top$
be the SVD decompositions of $\bs X$ and $\bs W = \frac{1}{m} \bs
\Theta^*(\bs y)$, respectively. We note $\bs u_i$, $\bs u'_j$, $\bs v_i$ and $\bs v'_j$ the columns of $\bs U$, $\bs U'$, $\bs V$ and $\bs V'$, respectively, with $i \in [r]$ and $j \in [n]$. We thus have $\hat{\bs X} = \bs U' \cl
M_r(\bs \Sigma') \bs V'^\top$ and 
$$
\bs X, \hat{\bs X} \in \cl S = {\rm span}\big(\{\bs u_i\bs v_i^\top: i \in [r]\}, \{\bs u'_i{\bs v'}_i^\top: i \in [r]\} \big) \subset \lrset{2r}.
$$
Let $\cl S^\perp$ be the orthogonal complement to the subspace $\cl S$.  
Since $\hat{\bs X} = \cl P_r( \bs W)$ and $\bs W = \overline{\bs W} + \bs W^\perp$, with $\overline{\bs W} = \cl P_{\cl S}(\bs W)$ and ${\bs W}^\perp = \cl P_{\cl S^\perp}(\bs W)$, we have 
$$
\ts \|\hat{\bs X}-(\overline{\bs W} + \bs W^\perp)\|^2_F \leq \|\bs X - (\overline{\bs W} + \bs W^\perp)\|^2_F.
$$ 
However, both $\hat{\bs X} - \overline{\bs W}$ and $\bs X - \overline{\bs W}$ belong to $\cl S$ and are thus orthogonal to $\bs W^\perp$. Decomposing both sides of the last inequality by Pythagoras' theorem and simplifying the common terms, we thus find $\|\hat{\bs X}-\overline{\bs W}\|_F \leq \|\bs X - \overline{\bs W}\|_F$, which gives 
$$
\|\hat{\bs X}- \bs X\|_F \leq 2\|\bs X - \overline{\bs W}\|_F
$$ 
by the triangular inequality.

We now proceed as in the proof of Theorem~\ref{thm:PBP-Union-Subspace}. We first note that for any $\bs T \in \cl S \cap \bb B_F^{n_1 \times n_2}$, $\bs X \pm \bs T \in \cl S$ since $\cl S \subset \lrset{2r}$ is a subspace including ${\bs X \in \cl S \cap \bb B_F^{n_1 \times n_2}}$. Second, by assumption, $\bs \Theta$~respects the RIP$(\lrset{2r}, \epsilon)$ and the random mapping $\Dmap$ satisfies the LPD$(\lrset{2r}, \bs \Phi, \nu)$. Therefore, (the matrix form of) Cor.~\ref{cor:RIP-L-LPD} with $\cl K' = \lrset{2r} \supset \cl S$ gives
\begin{align*}
\|\hat{\bs X}- \bs X\|_F&\ts \leq\ 2 \|\bs X - \overline{\bs W}\|_F = \ 2\,\sup_{\bs T\in \bb B^{n_1\times n_2}_F} \scp{\bs T}{\cl P_{\cl S}(\bs X - \bs W)}\\
&\ts = \ 2\,\sup_{\bs T\in \cl S\cap\bb B^{n_1\times n_2}_F} \scp{\bs T}{\bs X -\bs W}\ =\ 2\,\sup_{\bs T\in \cl S\cap\bb B^{n_1\times n_2}_F} \big(\scp{\bs T}{\bs X} - \frac{1}{m}\scp{\bs T}{\bs \Theta^*(\bs y)}\big)\\
&\ts \leq 2\,(\epsilon + \nu).
\end{align*}
Moreover, if $\bs X$ is fixed, we also see that only the L-LPD$(\lrset{2r}, \bs \Phi, \bs X, \nu)$ is required on~$\Dmap$. 
\end{proof}

\subsection{Bounded, convex and symmetric sets} 
\label{sec:PBP-convex-set}

The PBP method can also achieve a small reconstruction error for any signal belonging to a bounded, convex and symmetric set provided that both the RIP of $\tinv{\sqrt m}\bs \Phi$ and the (L)LPD of $\Dmap$ hold on~$\cl K- \cl K$. However, this error is amplified compared to ones observed for the more structured sets analyzed in the previous section.

\begin{theorem}[PBP for bounded convex sets]
\label{thm:PBP-convex-set} 
Given a bounded, symmetric and convex set $\cl K \subset \bb B^n$, and two distortions $\epsilon, \nu > 0$, if $\tinv{\sqrt m} \bs \Phi \in \bb R^{m \times n}$ respects the RIP$(\cl K, \epsilon)$ and if the mapping $\Dmap: \bb R^n \to \bb R^m$ satisfies the LPD$(\cl K, \bs \Phi, \nu)$, then, for all $\bs x \in \cl K$, the estimate $\hat{\bs x}$ obtained by the PBP of $\bs y=\Dmap(\bs x)$ onto $\cl K$ satisfies 
\begin{equation}
  \label{eq:error-PBP-convex-set}
  \ts \|\bs x -\hat{\bs x} \| \leq (4 \epsilon + 2 \nu)^{\frac{1}{2}}.
\end{equation}
Moreover, if $\bs x$ is fixed, then the same result holds if $\Dmap$ respects the L-LPD$(\cl K, \bs \Phi, \bs x, \nu)$.
\end{theorem}

\begin{proof}
Since $\bs x\in \cl K$ and $\cl K$ is symmetric and convex, the nonexpansivity of the orthogonal projector $\cl P_{\cl K}$ onto $\cl K$~\cite{JD2005} gives
$$
\ts \|\bs x- \hat{\bs x}\|^2 = \|\cl P_{\cl K}(\bs x)-\cl P_{\cl K}(\bs a)\|^2\ \leq\ \scp{\bs x-\bs a}{\cl P_{\cl K}(\bs x)-\cl P_{\cl K}(\bs a)}, 
$$
with $\bs a :=\frac{1}{m} \bs \Phi^\top \bs y$. Moreover, since $\cl K$ is symmetric,
\begin{align*}
\ts \scp{\bs x-\bs a}{\cl P_{\cl K}(\bs x)-\cl P_{\cl K}(\bs a)}&\ts \leq |\scp{\bs x-\bs a}{\cl P_{\cl K}(\bs x)}| + |\scp{\bs x-\bs a}{\cl P_{\cl K}(\bs a)}|\\
&\ts \leq\ 2\,\sup_{\bs u\in \cl K}\scp{\bs u}{\bs x-\bs a}\\
&\ts =\ 2\,\sup_{\bs u\in \cl K}\,\big(\scp{\bs u}{\bs x} - \frac{1}{m}\scp{\bs\Phi \bs u}{\Dmap(\bs x)}\big). 
\end{align*}
Therefore, Cor.~\ref{cor:RIP-L-LPD} with the symmetric convex set $\cl K' = \cl K \subset \bb B^n$ provides $\sup_{\bs u\in \cl K}\,(\scp{\bs u}{\bs x} - \frac{1}{m}\scp{\bs\Phi \bs u}{\Dmap(\bs x)}) \leq 2 \epsilon + \nu$, which gives the result.
Moreover, if $\bs x$ is fixed, only the L-LPD$(\cl K, \bs \Phi, \bs x, \nu)$ is required on $\Dmap$.
\end{proof}

Note that the Theorem~\ref{thm:PBP-convex-set} presents a worst case analysis of the reconstruction error of the PBP method for bounded convex sets. In particular,~\eqref{eq:error-PBP-convex-set} displays only a reconstruction error which is the square root of those presented in Thm.~\ref{thm:PBP-Low-Rank} and Thm.~\ref{thm:PBP-Union-Subspace}.  We will see, however, that at least for the convex set of compressible signals $\cpset{s} :=\{\bs u \in \bb R^n: \|\bs u\|_1 \leq \sqrt s, \|\bs u\| \leq 1 \}$ and if $\Dmap$ is the quantized random mapping $\Amap$ given in~\eqref{eq:Uniform-dithered-quantization}, the numerical reconstruction errors of the PBP presented in Sec.~\ref{sec:numericalPBP} behaves similarly to the ones predicted in the case of structured sets, if the random sensing matrix is either Gaussian, Bernoulli or a random partial Fourier/DCT sensing. 

\section{Local LPD for noisy linear mapping}
\label{sec:LPD-noisy-linear-map}

The previous sections have focused on exploring the implications of the (L)LPD when this property holds for a mapping $\Dmap$ inducing the distorted CS model~\eqref{eq:distorted-CS-model}. The question is, however, to understand for which mapping $\Dmap$ and under which conditions on the space $\cl K$ and on the number of observations $m$ we can expect this property to be verified with high probability. 

In this section, we first prove that the L-LPD holds when the mapping $\Dmap$ amounts to the corruption of a linear sensing with an additive sub-Gaussian noise $\bs \rho \in \bb R^m$. Hereafter, the condition that the linear sensing must respect is less constraining than the RIP, \ie we ask it to be Lipschitz continuous over $\cl K$ in the following sense.
\begin{definition}[Lipschitz continuity]
\label{def:lipschitz}
A linear mapping $\bs \Gamma: \bb R^n \to \bb R^m$ is said to be
$(\eta,L)$-Lipschitz continuous over $\cl K \subset \bb R^n$ for some
$L,\eta > 0$ if
$$
\|\bs \Gamma(\cl K_{\loc}^{(\eta)})\| \leq L \eta,\quad \text{and}\quad \|\bs \Gamma (\clip{\cl K})\| \leq L \|\clip{\cl K}\|
$$
with $\cl K_{\loc}^{(\eta)}$ and $\clip{\cl K}$ defined in~\eqref{eq:local-set} and~\eqref{eq:set-clipping}, respectively.
\end{definition}

Interestingly, in the case where $\cl K$ is a cone and if $\tinv{\sqrt m}\bs \Phi$ respects the RIP$(\cl K - \cl K,\epsilon)$, then it is also $(\eta,\sqrt{1+\epsilon})$-Lipschitz continuous for any $\eta>0$. This is easily observed from~\eqref{eq:RIP-cone}, which gives $\tinv{\sqrt m}\|\bs\Phi(\bs u - \bs v)\| \leq \sqrt{1+\epsilon} \|\bs u - \bs v\|$, for all $\bs u,\bs v \in \cl K$ with $\bs u - \bs v \in \eta\bb B^n$.  
Moreover, matrices $\tinv{\sqrt m} \bs\Phi$ drawn from a RIP matrix distribution over a bounded, star-shaped sets are also Lipschitz continuous (see Sec.~\ref{sec:cond-lipsch-cont}).

\medskip
The key element that allows us to prove the (L)LPD of $\Dmap$ is that a Lipschitz continuous mapping preserves the low-complexity nature of a set in its image, as measured by its Kolmogorov entropy. 

\begin{lemma}[Adapted from {\cite[Lemma 4]{Ker2016SML}}]
\label{lem:lipschitz-mapping-and-kolmogorov}
Given a bounded subset $\cl K \subset \bb R^n$, a radius $\eta>0$ and a linear mapping $\bs \Phi \in \bb R^{m \times n}$, if $\tinv{\sqrt m}\,\bs \Phi \in \bb R^{m \times n}$ is $(\eta,L)$-Lipschitz continuous over $\cl K$ for some $L>0$, then
\begin{equation}
  \label{eq:bound-kolmo-lipschitz}
  \cl H(\bs \Phi\cl K, L \eta\sqrt m) \leq \cl H(\cl K, \eta).  
\end{equation}
\end{lemma}

\begin{proof}
Let $\cl K_{\eta}$ be an optimal $\eta$-covering of $\cl K$ for some
$\eta>0$. Then, all $\bs a \in \bs
\Phi\cl K$ can be rewritten as $\bs a = \bs \Phi \bs x = \bs \Phi \bs
x_0 + \bs \Phi \bs r$ for some $\bs x \in \cl K$, with $\bs x_0$ the
closest point to $\bs x$ in $\cl K_{\eta}$, and ${\bs r \in (\cl K
- \cl K) \cap \eta \bb B^n}$. Therefore, $\bs \Phi \cl
K_{\eta}$ is a $(L \eta\sqrt m)$-covering of $\cl J :=
\bs\Phi \cl K$ since, from the Lipschitz continuity of $\tinv{\sqrt
  m}\bs \Phi$, $\|\bs \Phi \bs r\| \leq L \eta \sqrt m$.  In particular, 
$\cl H(\cl J, L \eta\sqrt m) \leq \log|\bs \Phi\cl K_{\eta}| \leq \log|\cl K_{\eta}| = \cl H(\cl K, \eta)$, since the covering $\bs \Phi \cl
K_{\eta}$ of $\cl J$ is not necessarily optimal.  
\end{proof}

We can now state the main result of this section.
\begin{proposition}[L-LPD for noisy linear sensing]
Given a set $\cl K \subset \bb R^n$ with $\|\clip{\cl K}\| \leq 1$, a distortion $\epsilon >0$ and a matrix $\tinv{\sqrt m} \bs \Phi$ that is $(\epsilon, L)$-Lipschitz continuous over $\cl K$ for some $L>0$, if $\bs \rho \in \bb R^m$ is a vector with \iid centered sub-Gaussian random components, \ie $\|\rho_i\|_{\psi_2} \leq R$ for $1\leq i \leq m$ and $R \geq 1$, and
$$
m \gtrsim \epsilon^{-2} \cl H(\clip{\cl K}, c\epsilon),
$$
then, given $\bs u \in \clip{\cl K}$, the mapping $\Dmap(\cdot) := \bs \Phi \cdot + \bs \rho$ respects the L-LPD$(\cl K, \bs \Phi, \bs u, RL \epsilon)$ with probability exceeding $1 - C\exp(-c' \epsilon^2 m)$. 
\end{proposition}

\begin{proof}
We first note that since $\|\rho_i\|_{\psi_2} \leq R$, then, for fixed $\bs u, \bs v \in \cl K \cap \bb B^n$, $\scp{\Dmap(\bs u) - \bs \Phi \bs u}{\bs\Phi \bs v} = \sum_{i=1}^m \rho_i\, (\bs\Phi \bs v)_i$ is a weighted sum of \iid centered sub-Gaussian random variables. Therefore, from~\cite[Prop. 5.10]{V2012}, $\bb P[|\scp{\Dmap(\bs u) - \bs \Phi \bs u}{\bs\Phi \bs v}|>t] \leq C \exp(-c \frac{t^2}{R^2\|\bs\Phi \bs v\|^2})$. Applying the change of variable $t=\epsilon R \|\bs \Phi \bs v\| \sqrt m$, this shows that
\begin{equation}
  \label{eq:fixed-point-concent-subgaussian-noise}
\ts \tinv{m} |\scp{\Dmap(\bs u) - \bs \Phi \bs u}{\bs\Phi \bs v}| \leq \epsilon R \tinv{\sqrt m} \|\bs \Phi \bs v\|,  
\end{equation}
with probability exceeding $1 - C \exp(-c \epsilon^2 m)$. Given an optimal $\epsilon$-covering $\cl K_\epsilon$ of $\clip{\cl K}$, \ie with $\log |\cl K_\epsilon| = \cl H(\cl K, \epsilon)$, a standard union bound argument then provides that~\eqref{eq:fixed-point-concent-subgaussian-noise} holds for all $\bs v' \in \cl K_\epsilon$ with probability exceeding $1 - C \exp(-\tinv{2} c \epsilon^2 m)$ provided $m \geq \tinv{2c} \epsilon^{-2}\cl H(\clip{\cl K}, \epsilon)$.

Moreover, since $\|\rho_i^2\|_{\psi_1} \leq 2 \|\rho_i\|^2_{\psi_2} \leq 2 R^2$ for $1\leq i \leq m$~\cite[Lemma 5.14]{V2012}, $\|\bs \rho\|^2 = \sum_{i=1}^m (\rho_i)^2$ is a sum of sub-exponential \iid random variables so that from~\cite[Cor. 5.17]{V2012} (with setting there $\epsilon = 1$) proves that $\|\bs \rho\|^2 \leq 2 R^2 m$ with probability exceeding $1 - 2\exp(-c m)$. 

Therefore, conditionally to the last two random events, which occur jointly with probability exceeding $1 - C \exp(-\tinv{2} c \epsilon^2 m)$ (by union bound) under the same requirement on $m$, we have $\tinv{\sqrt m} \|\bs \Phi \bs u\| \leq L \|\clip{\cl K}\| \leq L$, since we assumed $\tinv{\sqrt m} \bs \Phi$ to be $(\epsilon,L)$-Lipschitz continuous over $\cl K$, and $\tinv{\sqrt m} \|\bs \rho\| \leq \sqrt 2 R$. Consequently, for any $\bs v \in \cl K \cap \bb B^n$ whose closest point in $\cl K_\epsilon$ is $\bs v'$, \ie $\|\bs v - \bs v'\| \leq \epsilon$, we have 
\begin{align*}
\ts \tinv{m} |\scp{\bs \rho}{\bs\Phi \bs v}| &\ts \leq \tinv{m} |\scp{\bs \rho}{\bs\Phi \bs v'}| + \tinv{m} |\scp{\bs \rho}{\bs\Phi(\bs v - \bs v')}|\ \leq \epsilon R \tinv{\sqrt m} \|\bs \Phi \bs v'\|+ \tinv{m} |\scp{\bs \rho}{\bs\Phi(\bs v - \bs v')}|\\
&\ts \leq \epsilon R L + \tinv{m} \|\bs \rho\| \|\bs \Phi(\bs v - \bs v')\| \leq \epsilon R L + \sqrt 2 L R \epsilon = (1+ \sqrt 2) RL \epsilon,
\end{align*}
where we used Cauchy-Schwarz in the last line. A simple rescaling of $\epsilon$ then provides the result.  
\end{proof}

As provided in the next corollary, the previous proposition enables us to characterize the L-LPD property of the quantized and dithered mapping $\Amap$ introduced in~\eqref{eq:Uniform-dithered-quantization}. This is straightforwardly obtained by observing that, given the resolution $\delta > 0$ defining both the quantizer $\cl Q$ and the support of the random uniform dither $\bs \xi \sim \cl U^m([0,\delta])$, for any $\bs u \in \bb R^n$, we have
$$
\bs y = \Amap(\bs u) = \bs \Phi \bs u + \bs \rho,
$$ 
with $\bs \rho = \Amap(\bs u) - \bs \Phi \bs u$, $\bb E \bs \rho = \bs 0$ (from Lemma~\ref{lem1}) and
$$
\|\rho_i\|_{\psi_2} = \|\cl Q((\bs \Phi \bs u)_i + \xi_i) - (\bs \Phi \bs u)_i\|_{\psi_2} \leq \|\cl Q((\bs \Phi \bs u)_i + \xi_i) - (\bs \Phi \bs u + \xi_i)_i\|_{\psi_2} + \|\xi_i\|_{\psi_2} \leq 2 \delta, 
$$
since for any \rv $X$ such that $|X| \leq s$ for some $s>0$, $\|X\|_{\psi_2} \leq s$. 

\begin{corollary}[L-LPD for quantized, dithered mapping]
\label{cor:llpd-noisy-linear}
Given a low-complexity set $\cl K \subset \bb R^n$ with $\|\clip{\cl K}\| \leq 1$, a distortion $\epsilon >0$, a quantization resolution $\delta>0$ and a matrix $\tinv{\sqrt m} \bs \Phi$ that is $(\epsilon, L)$-Lipschitz continuous over $\cl K$ for some $L>0$, if 
\begin{equation}
  \label{eq:llpd-noisy-linear-cond-m}
  m \geq C \epsilon^{-2} \cl H(\clip{\cl K}, c\epsilon),  
\end{equation}
where $C,c>0$ depends on $L$, then, given $\bs u \in \clip{\cl K}$, the random mapping $\Amap(\cdot) = \cl Q(\bs \Phi \cdot + \bs \xi)$ associated with $\bs \xi \sim \cl U^m([0,\delta])$ respects the L-LPD$(\cl K, \bs \Phi, \bs u, \delta \epsilon)$ with probability exceeding $1 - C\exp(-c' \epsilon^2 m)$.
\end{corollary}

\section{LPD for a quantized, dithered random mapping}
\label{sec:LPD}

Cor.~\ref{cor:llpd-noisy-linear} in the previous section states that the L-LPD almost trivially holds for the quantized, dithered random mapping $\Amap$ introduced in~\eqref{eq:Uniform-dithered-quantization} by observing that the corresponding sensing model is equivalent to a noisy linear sensing corrupted by an additive sub-Gaussian noise with \iid components. However, by analyzing more carefully the interplay between the quantizer discontinuities and the dithering in $\Amap$, we can prove that the uniform LPD holds too. Moreover, for structured sets for which the Kolmogorov entropy only increases logarithmically with the involved radius (\eg for the sets of sparse vectors and low-rank matrices), the sample complexity ensuring the LPD is very similar to the one guaranteeing the L-LPD, up to few logarithmic factors.  

Hereafter, we split our developments in two steps by noting that ${\Amap = \Amapp \circ \bs \Phi}$, with $\Amapp: \bs a \in \bb R^m \mapsto \cl Q(\bs a + \bs \xi) \in \delta \bb Z^m$ the dithered quantizer associated with a random vector $\bs \xi \sim \cl U^m([0,\delta])$ and the quantizer $\cl Q(\cdot) := \delta \lfloor \cdot/\delta\rfloor$. First, we analyze in Sec.~\ref{sec:preparotary-results} the geometric properties of $\Amapp$ and show that it respects a certain limited projection distortion property directly in $\bb R^m$. Second, in Sec.~\ref{sec:bridging} the characteristics of ${\Amap = \Amapp \circ \bs \Phi}$ are then explained by the Lipschitz embedding realized by any RIP matrix $\bs \Phi \in \bb R^{m \times n}$ between a low-complexity set $\cl K \subset \bb R^n$ and its (still low-complexity) image $\cl J := \bs \Phi \cl K$, and from the limited projection distortions induced by $\Amapp(\cl J) \subset \delta \bb Z^m$ in~$\bb R^m$.  

\subsection{Analysis of the dithered quantizer}
\label{sec:preparotary-results}

As a direct effect of the random dither $\bs \xi$, the following Lemma shows that the number of components of $\Amapp$ that are discontinuous in a $\ell_\infty$-neighborhood of an arbitrary point of $\bb R^m$ is controlled by the size of this neighborhood. This continuity analysis is inspired by a strategy developed in~\cite{BRM2017} for more general discontinuous mappings. 

\begin{lemma} 
\label{lem:bound-discont}
Given $\bs a \in \bb R^m$, $\epsilon > 0$, $0<\rho<\delta/2$, $\bs \xi \sim \cl U^m([0,\delta])$, we denote by $\Amapp_i:\bb R^m \to \bb R$ the $i^{\rm th}$ component of the random mapping $\Amapp(\cdot) := \cl Q(\cdot + \bs \xi)$ ($i \in [m]$), and define the discrete random variable (associated with the randomness of $\bs \xi$)
$$
Z = Z(\bs a + \rho \bb B^m_{\ell_\infty}) := |\{i: \Amapp_i(\cdot) \notin {\sf C}^0(\bs a + \rho \bb B^m_{\ell_\infty})\}| \in \{0,\,\cdots,m\},
$$
\ie $Z$ counts the components of $\Amapp$ that are discontinuous over $\bs a + \rho \bb B^m_{\ell_\infty}$ (\ie having at least two distinct values over this set). The random variable $Z$ has a binomial distribution with $m$ trials and a probability of success $p:=\frac{2\rho}{\delta}$, \ie $Z \sim {\rm Bin}(m,p)$. Therefore, $\bb E Z = mp=m \frac{2\rho}{\delta}$, $\inv{m}(\bb E Z^2 - (\bb E Z)^2) = p(1-p) =: \sigma^2  < p$ and 
\begin{equation}
    \label{eq:number-of-discont}
    \ts \bb P[Z \geq m \tfrac{2\rho}{\delta} + \epsilon] \leq \exp(-\tinv{2}\frac{3m \epsilon^2}{3\sigma^2 + \epsilon}).
\end{equation}
In particular, setting $\epsilon = p > \sigma^2$ provides 
\begin{equation}
    \label{eq:number-of-discont-special-case}
\ts \bb P[Z \geq m \tfrac{4\rho}{\delta} ] \leq \exp(- m \tfrac{3\rho}{4\delta}).
\end{equation}
\end{lemma}

\begin{proof}
For $i \in [m]$, let $Z_i$ be the random variable equal to 1 if $\Amapp_i(\cdot) \notin {\sf C}^0(\bs a + \rho \bb B^m_{\ell_\infty})$, and to 0 otherwise. We have $Z = \sum_i Z_i$. Moreover, since analyzing the discontinuity of $\Amapp_i$ over $\bs a + \rho \bb B^m_{\ell_\infty}$ amounts to studying that of $\lambda \in \bb R \mapsto \cl Q(\lambda + \xi_i) \in \delta \bb Z$ over the interval $[a_i-\rho, a_i+\rho]$, we directly find $\bb P[Z_i = 1] = \tfrac{2\rho}{\delta}$ from the uniformity of $\xi_i$ and the fact that the discontinuities of $\cl Q$ are located on $\delta \bb Z$.  The rest of the Lemma is then simply obtained by observing that the \rvs $\{Z_i:1\leq i \leq m\}$ are \iid and by a simple application of Bernstein's inequality to $Z - \bb E Z$~\cite{V2012}.  
\end{proof}

Similarly to some developments made in Sec.~\ref{sec:LPD-noisy-linear-map}, the next lemma indicates that, on a fixed vectors $\bs a \in \bb R^m$ and in a fixed direction induced by a vector $\bs b \in \bb R^m$, $\Amapp(\bs a)$ concentrates around its expectation $\bb E\,\Amapp(\bs a) = \bs a$ (from Lemma~\ref{lem1}), and this concentration improves exponentially with~$m$.  

\begin{lemma}
\label{lem:simple-concent}
Given $\bs a, \bs b \in \bb R^m$, $\epsilon > 0$, $\bs \xi \sim \cl U^m([0, \delta])$, and $\Amapp(\cdot) := \cl Q(\cdot + \bs \xi)$, we have
\begin{equation}
  \label{eq:concent-random-map-with-random-zeros}
  \bb P\big[\,|\scp{\Amapp(\bs a) - \bs a}{\bs b}| \geq \delta \epsilon \sqrt{
m} \,\|\bs b\|\,\big] \leq 2\exp(- 2\epsilon^2 m).
\end{equation}
\end{lemma}
\begin{proof}

We define the \iid random variables $\{X_i = \cl Q(a_i + \xi_i) - a_i
\in [-\delta ,\delta ],\ i\in [m]\}$. From Lemma~\ref{lem1}, $\bb E X_i =
0$. Moreover
$|X_i| \leq  \delta $ for all $i\in [m]$. Therefore, Hoeffding's inequality\footnote{We could use here the sub-Gaussianity of the \rvs $X_i$ to study the concentration of $\scp{\Amapp(\bs a)}{\bs b}$ around $\scp{\bs a}{\bs b}$. However, Hoeffding's inequality is more accurate with respect to the multiplicative constants of the tail bound.} provides
$\bb P[|\scp{\Amapp(\bs a) - \bs a}{\bs b}|\geq t] =
\bb P[|\sum_i X_ib_i| \geq t ] \leq 2\exp(-\frac{2t^2}{\delta^2 \|\bs b\|^2})$ for $t\geq 0$. Operating the change of variable $t= \epsilon\delta \sqrt{
m}\,\|\bs b\|$ gives the result.
\end{proof}

We are now ready to show that projecting a vector onto the image of any vector of a set $\cl J \subset \bb R^m$ by $\Amapp$ is close to the inner product of the two vectors, \ie $\Amapp$ respects a form of the limited projection distortion directly in $\bb R^m$ provided that $\cl J$ has bounded Kolmogorov entropy~\cite{KT1961}. In fact, given a fixed direction $\bs b/\|\bs b\|$ with $\bs b \in \bb R^m$, this amounts to extending Lemma~\ref{lem:simple-concent} and to analyzing how $\Amapp(\bs a + \bs \xi) - \bs a$ concentrates in this direction for all $\bs a \in \cl J$. 

\begin{proposition}[Limited projection distortion of $\Amapp$ in $\bb R^m$ and in a fixed direction]
\label{prop:Rmlpd-on-fixed-b}
Given a distortion $\epsilon > 0$, a fixed $\bs b \in \bb R^m$, and a subset $\cl J \subset
\bb R^m$ with bounded Kolmogorov entropy $\cl H(\cl J, \cdot)$, provided 
\begin{equation}
  \label{eq:Rmlpd-on-fixed-b-cond}
\ts m \geq C \epsilon^{-2}\,\cl H(\cl J, c\, \epsilon^3 \sqrt m),  
\end{equation}
with $C \leq 24$ and $c\geq 1/27$, the random mapping $\Amapp(\cdot) = \cl Q(\cdot + \bs \xi)$ defined from $\bs \xi \sim \cl U^m([0, \delta])$ respects, 
\begin{equation}
  \label{eq:Lpd-in-Rm}
  \ts |\scp{\Amapp(\bs a)}{\bs b} - \scp{\bs a}{\bs b}|\ \leq\
  \, \min(\epsilon\,(1+ \delta), \delta) \sqrt{m}\, \|\bs
  b\|,\quad \forall \bs a \in \cl J,
\end{equation}
with probability exceeding $1 - 3 \exp(- C^{-1} \epsilon^2 m)$.
\end{proposition}
  
\medskip
Thanks to Lemmata~\ref{lem:bound-discont} and~\ref{lem:simple-concent}, the proof of this proposition is relatively simple and inspired by some methodology defined in~\cite{BRM2017}. We observe that the number of components of $\Amapp$ that are continuous (and thus constant) over each ball of a dense covering of $\cl J$ are close to $m$ when this value is large. Then, we bound the discontinuous components, which are in minority, with deterministic bounds known on the quantizer. This method is considerably simpler, in this case, than the softening strategy of the discontinuous mapping $\Amapp$ proposed in~\cite{J2015,JC2016}, as inspired by the softening of the one-bit (``sign'') quantizer introduced in~\cite{PV2014}. 

\begin{proof}[Proof of Prop.\ref{prop:Rmlpd-on-fixed-b}]
From the homogeneity of~\eqref{eq:Lpd-in-Rm} we can fix $\bs b \in \bb R^m$ and assume $\|\bs b\|
= 1$. Let $\cl J_{\eta \sqrt m}$ be an optimal $(\eta \sqrt m)$-covering in the $\ell_2$-metric of
$\cl J$ for some $\eta > 0$ to be fixed later, \ie $\log |\cl J_{\eta \sqrt m}|
= \cl H(\cl J, \eta \sqrt m)$. 

Let us first consider an arbitrary vector $\bs a' \in \cl J$. By definition of the covering above, we can always write $\bs a' = \bs a + \bs r$ for some $\bs a \in \cl J_{\eta\sqrt m}$ and $\bs r \in (\eta \sqrt m)\bb B^m$. Given a value $P>0$ whose value will be fixed later, we define the set
$$
\cl T = \cl T(\bs r) := \{i \in [m]: |r_i| \leq \eta \sqrt P\}.
$$
We observe that $|\cl T^\compl| \leq \frac{m}{P}$ independently of $\bs r$, since $\eta^2 \, m \geq \|\bs
r\|^2_2 \geq \|\bs r_{\cl T^\compl}\|^2_2 \geq \eta^2 P |\cl T^\compl|$. 
Writing $\sigma_i = \chi_{\cl T}(i)$ for $i\in [m]$ and $\bar{\bs r} := \bs \sigma \odot \bs r \in (\eta\sqrt P) \bb B^m_{\ell_\infty}$, we now develop the LHS of~\eqref{eq:Lpd-in-Rm}:
\begin{align}
&\ts |\scp{\Amapp(\bs a')}{\bs b} - \scp{\bs a'}{\bs b}|\ =\ |\sum_i b_i (\cl Q(a_i + r_i + \xi_i) - (a_i + r_i))|\nonumber\\
&\ts \leq |\sum_i b_i (\cl Q(a_i + \bar r_i + \xi_i) - (a_i + \bar  r_i))|\nonumber\\
&\ts \qquad + |\sum_i b_i (\cl Q(a_i + \bar r_i + \xi_i) - (a_i + \bar r_i) - \cl Q(a_i + r_i + \xi_i) + (a_i + r_i))|\nonumber\\
&\ts \leq |\sum_i b_i (\cl Q(a_i + \bar r_i + \xi_i) - (a_i + \bar r_i))|\nonumber\\
&\ts \qquad + |\sum_{i\in \cl T^\compl} b_i (\cl Q(a_i + \xi_i) + r_i - \cl Q(a_i + \xi_i + r_i ))|\nonumber\\
&\ts \leq |\sum_i b_i (\cl Q(a_i + \bar r_i + \xi_i) - (a_i + \bar r_i))|\ +\ \delta \|\bs b_{\cl T^\compl}\|_1 \nonumber\\
\label{eq:partial-bound-1}
&\ts \leq |\sum_i b_i (\cl Q(a_i + \bar r_i + \xi_i) - (a_i + \bar r_i))|\ +\ \delta \frac{1}{\sqrt P} \sqrt m\ =:\ I\ +\ \delta \frac{1}{\sqrt P} \sqrt m,
\end{align}
where we have used the fact that $|(\lfloor a\rfloor + b) - (\lfloor a + b\rfloor) |\leq 1$ for all $a,b\in \bb R$, and $\chi_{\cl T}(i)=1$ if~$i\in T$ and 0 otherwise. 

Notice that the term $I = |\scp{\Amapp(\bs a + \bar{\bs r}) - (\bs a + \bar{\bs r})}{\bs b}|$ in~\eqref{eq:partial-bound-1} has thus to be characterized for a neighborhood $(\eta\sqrt P) \bb B^m_{\ell_\infty}$ of $\bs a$ since $\|\bar{\bs r}\|_{\infty} \leq \eta\sqrt P$. Interestingly, from Lemma~\ref{lem:bound-discont} and
\eqref{eq:number-of-discont-special-case}, a union bound argument allows us to bound the number $Z$ of discontinuous components\footnote{Remark that by bounding $Z$ we do not forbid $\Amapp$ to have different discontinuous components for different balls of $\cl J_{\eta \sqrt m} +  (
\eta \sqrt{P}) \bb B_{\ell_\infty}^m$.} of $\Amapp$ over all $\ell_\infty$-balls of $\cl J_{\eta \sqrt m} +  (
\eta \sqrt{P}) \bb B_{\ell_\infty}^m$ so that 
$$
\ts \bb P\big[\,\exists \bs u \in \cl J_{\eta
  \sqrt m}: Z(\bs u + (\eta \sqrt P)\bb B^m_{\ell_\infty}) > m
\tfrac{4 \eta \sqrt P}{\delta}\,\big]\ \leq\ \exp\big(\,\cl H(\cl J, \eta
\sqrt m) - \frac{3}{4}m \tfrac{\eta \sqrt P}{\delta}\big).
$$
Equivalently,
$$
\ts \max \{Z(\bs u + (\eta \sqrt P)\bb B^m_{\ell_\infty}): \bs u \in \cl J_{\eta \sqrt m}\} \leq m
\tfrac{4 \eta \sqrt P}{\delta},
$$
with probability exceeding $1-\exp(\cl H(\cl J, \eta
\sqrt m) - \frac{3}{4}m \tfrac{\eta \sqrt P}{\delta})$.

Let us denote $\cl C = \cl C(\bs a) \subset [m]$ the set of components of $\Amapp$ that are continuous over $\bs a + (\eta\sqrt P) \bb B^m_{\ell_\infty}$, with $\sigma'_i := \chi_{\cl C}(i)$ for $i\in[m]$. We also write $\tilde{\bs r} = \bs \sigma' \odot \bs \sigma \odot \bs r = \bs \sigma' \odot \bar{\bs r}$. By construction, we have thus $\cl Q(a_i + \tilde r_i + \xi_i) = \cl Q(a_i + \xi_i)$ for $i \in [m]$. The bound above shows that with the same probability
\begin{align}
I&= \ts |\sum_{i} b_i (\cl Q(a_i + \bar r_i + \xi_i) - (a_i + \bar r_i))|\nonumber\\
&\ts \leq |\sum_{i} b_i (\cl Q(a_i + \tilde r_i + \xi_i) - (a_i + \tilde r_i))|\nonumber\\
&\ts \qquad + |\sum_{i} b_i (\cl Q(a_i + \tilde r_i + \xi_i) - (a_i + \tilde r_i) - \cl Q(a_i + \bar r_i + \xi_i) + (a_i + \bar r_i))|\nonumber\\
&\ts \leq |\sum_{i} b_i (\cl Q(a_i + \xi_i) - a_i)| + |\sum_{i} b_i \tilde r_i|\nonumber\\
&\ts \qquad + |\sum_{i} b_i (\cl Q(a_i + \xi_i + \tilde r_i) + \bar r_i - (\cl Q(a_i + \xi_i + \bar r_i) + \tilde r_i) )|.\nonumber
\end{align}
Note that in the last term, since $\tilde{r}_i = \sigma'_i \bar{r}_i$ is equal to $\bar{r}_i=\sigma_i r_i$ if $i\in \cl C$ or if $i \in \cl T^\compl$, all terms of the sum are zero but those for which $i\in (\cl C \cup \cl T^\compl)^\compl = \cl C^\compl \cap \cl T \subset \cl C^\compl$. For these non-zero terms, $|(\cl Q(a_i + \xi_i + \tilde r_i) + \bar r_i - (\cl Q(a_i + \xi_i + \bar r_i) + \tilde r_i)|\leq\delta$. Therefore,
\begin{align}
I &\ts \leq |\sum_{i} b_i (\cl Q(a_i + \xi_i) - a_i)| + \|\bs b\| \|\tilde{\bs r}\| + \delta \|\bs b_{\cl C^\compl}\|_1\nonumber\\
&\ts \leq |\sum_{i} b_i (\cl Q(a_i + \xi_i) - a_i)|\ +\ \eta \sqrt m +\ 2\delta (\tfrac{\eta \sqrt P}{\delta})^{1/2} \sqrt m\nonumber\\
\label{eq:partial-bound-2}
&\ts =: I\!I\ +\ \eta \sqrt m\ +\ 2\delta (\tfrac{\eta \sqrt P}{\delta})^{1/2} \sqrt m.
\end{align}

Applying Lemma~\ref{lem:simple-concent} with a union bound on all $\bs a
\in \cl J_{\eta \sqrt m}$ ensures that  
\begin{equation}
  \label{eq:partial-bound-3}
\ts I\!I = |\sum_{i} b_i (\cl Q(a_i + \xi_i)
  - a_i)| \leq \delta \epsilon\, \sqrt{m},
\end{equation} 
holds with probability exceeding $1 - 2 \exp(\cl H(\cl J, \eta \sqrt m) - 2 \epsilon^2 m)$.   

Gathering all the bounds from~\eqref{eq:partial-bound-1},~\eqref{eq:partial-bound-2} and~\eqref{eq:partial-bound-3}  and applying a last union bound on the
failure of the different events shows that
\begin{align*}
\ts |\scp{\Amapp(\bs a')}{\bs b} - \scp{\bs a'}{\bs b}|
&\ts \leq \delta \epsilon\, \sqrt m + \eta \sqrt m  + 2\delta (\tfrac{\eta \sqrt P}{\delta})^{1/2} \sqrt m + \delta \frac{1}{\sqrt P} \sqrt m,
\end{align*}
with probability exceeding 
$$
1 - \exp(\cl H(\cl J, \eta
\sqrt m) - \tfrac{3}{4} m \tfrac{\eta \sqrt P}{\delta}) - 2\exp(\cl H(\cl J, \eta \sqrt m) - 2 \epsilon^2 m).
$$

Assuming first that $\epsilon \leq 1$, we can now set our free parameters and fix $P= \delta^2 \epsilon^{-2}$, $\eta =
\epsilon^3 \leq \epsilon$, so that $\tfrac{\eta \sqrt
  P}{\delta} = \epsilon^2 \leq \epsilon $ and $(\frac{1}{P})^{1/2} =
\frac{\epsilon}{\delta}$. Therefore, from a few
simplifications, we have
\begin{align*}
\ts |\scp{\Amapp(\bs a')}{\bs b} - \scp{\bs a'}{\bs b}|
&\ts \leq\, 3 \delta\epsilon \sqrt m + 2\epsilon \sqrt m \leq 3 \epsilon (1+\delta) \sqrt m,
\end{align*}
with probability exceeding $1 - 3 \exp(- \frac{3}{8}\,m \epsilon^2)$ provided 
$$
\ts m \geq \frac{8}{3} \epsilon^{-2}\,\cl H(\cl J, \epsilon^3 \sqrt m).
$$

Moreover, using Cauchy-Schwarz, we have deterministically
$$
\ts |\scp{\Amapp(\bs a')}{\bs b} - \scp{\bs a'}{\bs b}|\ \leq \|\Amapp(\bs a') - \bs a'\| \|\bs b\|\ \leq\ \delta \sqrt m,
$$
from the definition of $\Amapp$ and since $|\lfloor \lambda + u \rfloor - \lambda| \leq 1$ for all $\lambda \in \bb R$ and all $u \in [0,1]$. Finally, \eqref{eq:Lpd-in-Rm}  is obtained from the rescaling $\epsilon \to \epsilon/3$, and observing that $\delta < (1+\delta)\epsilon$ if $\epsilon > 1$.

\end{proof}
\medskip

We now prove a (global) limited projection distortion property of $\Amapp$ for scalar products of all elements of a subset $\cl J \subset \bb R^m$, hence extending the previous proposition to all elements of $\bs b \in \cl J$ provided that this set has bounded Kolmogorov entropy. 

\begin{proposition}[Limited projection distortion of $\Amapp$ in $\bb R^m$]
\label{prop:Rmlpd-on-all-b-for-a-in-subset}
Given $\epsilon, \nu >0$, a subset $\cl J \subset
\bb R^m$ with bounded Kolmogorov entropy $\cl H(\cl J, \cdot)$, and another subset $\cl J' \subset \bb R^m$, possibly reduced to a
single point, such that, given $\bs b' \in \cl J$,  
\begin{equation}
  \label{eq:cond-on-Lpd-in-Rm-on-a}
  \ts \bb P[\forall \bs a \in \cl J': |\scp{\Amapp(\bs a)}{\bs b'} - \scp{\bs a}{\bs b'}|\ \leq\ \nu \sqrt{m} \|\bs
  b'\|]\ \geq 1 - C \exp(-c \epsilon^2 m),
\end{equation}
provided that
\begin{equation}
\label{eq:cond-main-result-embed}
\ts m \geq \frac{2}{c} \epsilon^{-2}\, \cl H(\cl J, \epsilon\sqrt m), 
\end{equation}
we have with
probability exceeding $1 - C \exp(-\frac{c}{2} \epsilon^2 m)$, 
\begin{equation}
  \label{eq:Lpd-in-Rm-uniform}
  \ts |\scp{\Amapp(\bs a)}{\bs b} - \scp{\bs a}{\bs b}|\ \leq\ \nu \sqrt m \|\cl J\| + \epsilon\delta m.
\end{equation} 
\end{proposition}

\begin{proof}
Let $\cl J_{\epsilon\sqrt m}$ be an optimal $(\epsilon\sqrt m)$-covering of $\cl
J$. From~\eqref{eq:cond-on-Lpd-in-Rm-on-a}, provided \eqref{eq:cond-main-result-embed} holds, we have, with probability exceeding $1 - C \exp(-\tinv{2} c\, \epsilon^2 m)$, 
$|\scp{\Amapp(\bs a)}{\bs b'} - \scp{\bs a}{\bs b'}|\ \leq\ \nu \sqrt m \|\bs b'\|$, for all $\bs a \in \cl J$ and all $\bs b' \in \cl J_{\epsilon\sqrt m}$. Therefore, for all $\bs b \in \cl J$, we can write $\bs b = \bs b_0 + \bs r$
with $\bs r \in (\epsilon\sqrt m)\bb B^m$ and
$\bs b_0 \in \cl J_{\epsilon\sqrt m}$, so that, using Cauchy-Schwarz,
\begin{align*}
\ts |\scp{\Amapp(\bs a)}{\bs b} - \scp{\bs a}{\bs b}|
&\ts = |\scp{\Amapp(\bs a) - \bs a}{\bs b}| \leq |\scp{\Amapp(\bs a) - \bs a}{\bs b_0}| + |\scp{\Amapp(\bs a) - \bs a}{\bs r}|\\
&\ts \leq \nu \sqrt m\,\|\cl J\|\ +\ \|\Amapp(\bs a) - \bs a\| \|\bs r\|\ \leq\ \nu \sqrt m\,\|\cl J\|\ +\ \delta \epsilon m.
\end{align*} 
\end{proof}

\subsection{Limited projection distortion of $\Amap$}
\label{sec:bridging}

We can now determine under which conditions $\Amap$ has limited projection distortion by 
bridging the analysis of $\Amapp$ with the action of a RIP matrix $\tinv{\sqrt m}\bs \Phi \in \bb R^{m\times n}$. 

\begin{proposition}[LPD for dithered and quantized random projections]
\label{prop:main-result-embed}
Given a subset $\cl K \subset
\bb R^n$ with $\|\clip{\cl K}\|\leq 1$, a distortion $\epsilon>0$, a quantization resolution $\delta > 0$, and a matrix $\tinv{\sqrt m} \bs \Phi$ that is $(\epsilon^3, L)$-Lipschitz continuous over $\clip{\cl K}$ with $L > 0$, if
\begin{equation}
\label{eq:cond-main-result-embed-unif}
\ts m \geq C \epsilon^{-2}\, \cl H(\clip{\cl K},\, c\, L^{-1}\epsilon^3 ),
\end{equation}
then, with probability exceeding $1 - C' \exp(- c' \epsilon^2 m)$, the random mapping $\Amap(\cdot) =\cl Q(\bs \Phi \cdot + \bs \xi)$ from $\bs \xi \sim \cl U^m([0,\delta])$ respects the LPD$(\cl K, \bs \Phi, \nu)$ with $\nu = L \min(\epsilon\, (1+\delta), \delta)$, \ie
\begin{equation}
  \label{eq:Lpd-in-Rn}
  \ts \tinv{m} |\scp{\Amap(\bs u)}{\bs \Phi \bs v} - \scp{\bs \Phi \bs u}{\bs \Phi \bs v}|\ \leq\ L \min(\epsilon\, (1+\delta), \delta),\quad \forall \bs u,\bs v \in \clip{\cl K}.
\end{equation}
\end{proposition}

We can observe that, compared to the L-LPD of $\Amap$ analyzed in Cor.~\ref{cor:llpd-noisy-linear}, the LPD reached in the context of Prop.~\ref{prop:main-result-embed} displays a quite different distortion, \eg if $\epsilon \to 0$ and $\epsilon/\delta$ is kept constant, the LPD and the L-LPD distortions decay linearly and quadratically in $\epsilon$, respectively.
  
The uniformity of the LPD in~\eqref{eq:Lpd-in-Rn} also imposes a larger number of measurements in \eqref{eq:cond-main-result-embed-unif} compared to \eqref{eq:llpd-noisy-linear-cond-m}, as induced by the smaller radius $\epsilon^3 < \epsilon$ of the Kolmogorov entropy (for $\epsilon \leq 1$).  However, as explained in Sec.~\ref{sec:conditions-llpd}, for conic sets such as the set of sparse signals or the set of low-rank matrices, the Kolmogorov entropy only depends logarithmically on the inverse of this radius; unstructured sets suffer more from this radius reduction since $\cl H$ is proportional to the square inverse of the radius, from Sudakov's inequality in \eqref{eq:sudakov}.

\begin{proof}[Proof of Prop.~\ref{prop:main-result-embed}.]
The proof consists in applying
Prop.~\ref{prop:Rmlpd-on-all-b-for-a-in-subset} in the case where $\cl J' = \cl J =
\bs \Phi \clip{\cl K}$, and to use
Lemma~\ref{lem:lipschitz-mapping-and-kolmogorov} to bound the Kolmogorov entropy of $\bs
\Phi \clip{\cl K}$ with the one of $\clip{\cl K}$.

Let us first assume $0<\epsilon \leq 1$. In this setting, from Prop.~\ref{prop:Rmlpd-on-fixed-b}, we have for a given $\bs v' \in \clip{\cl K}$, 
\begin{equation}
  \label{eq:tmp-lpd-proof}
\ts \bb P[\forall \bs u \in \cl K: |\scp{\Amap(\bs u)}{\bs \Phi \bs v'} - \scp{\bs \Phi \bs u}{\bs \Phi \bs v'}|\ \leq\ \epsilon\,(1+\delta) \sqrt m\, \|\bs \Phi \bs v'\|]\ \geq 1 - C \exp(- c \epsilon^2 m),  
\end{equation}
provided that
$$
\ts m\ \geq\ C' \epsilon^{-2}\,\cl H(\bs\Phi \clip{\cl K}, c' \epsilon^3\sqrt{m}).
$$ 
The bound~\eqref{eq:tmp-lpd-proof} matches the probability bound required in~\eqref{eq:cond-on-Lpd-in-Rm-on-a} of Prop.~\ref{prop:Rmlpd-on-all-b-for-a-in-subset} with $\nu = \epsilon\,(1+\delta) $. In this setting, Prop.~\ref{prop:Rmlpd-on-all-b-for-a-in-subset} requires in~\eqref{eq:cond-main-result-embed} to have $m \geq C \epsilon^{-2} \cl H(\bs \Phi \clip{\cl K}, \epsilon \sqrt m)$; all conditions on $m$ are thus guaranteed if
$$
m \geq C \epsilon^{-2}\,\cl H(\bs \Phi \clip{\cl K}, c \epsilon^3\sqrt m).
$$

However, if $\tinv{\sqrt m}\bs \Phi$ is $(\epsilon^3, L)$-Lipschitz continuous, Lemma~\ref{lem:lipschitz-mapping-and-kolmogorov}  justifies~\eqref{eq:cond-main-result-embed-unif} with
$$
\ts \cl H(\bs \Phi \clip{\cl K}, c\, \epsilon^3\sqrt m) \leq \cl H(\clip{\cl K}, c L^{-1} \epsilon^3 ).
$$

Consequently, by union bound over the events covered by Prop.~\ref{prop:Rmlpd-on-fixed-b} and Prop.~\ref{prop:Rmlpd-on-all-b-for-a-in-subset}, we have that~\eqref{eq:Lpd-in-Rm-uniform} in 
Prop.~\ref{prop:Rmlpd-on-all-b-for-a-in-subset} holds with
$\bs a = \bs \Phi \bs u$ and $\bs b = \bs \Phi \bs v$, for all $\bs u,
\bs v \in \clip{\cl K}$, and with
probability exceeding $1- C \exp(- c m \epsilon^2)$. Since $\|\bs \Phi \clip{\cl K}\| \leq \sqrt m L \|\clip{\cl K}\| \leq \sqrt mL$ from the assumed Lipschitz continuity of $\bs \Phi$ (see Def.~\ref{def:lipschitz}), this means that
$$
\ts |\scp{\Amap(\bs \Phi \bs u)}{\bs \Phi \bs v} - \scp{\bs \Phi\bs u}{\bs \Phi\bs v}|\ \leq\ \epsilon\,(1+\delta) \sqrt m \|\bs \Phi \clip{\cl K}\| + \epsilon \delta m \leq L \epsilon\,(1+\delta) m + \epsilon \delta m \leq 2L\epsilon\,(1+\delta)m. 
$$
Finally, \eqref{eq:Lpd-in-Rn} is obtained by applying the rescaling $\epsilon \to \epsilon/2$ above, the deterministic bound \eqref{eq:rem-delta-0}, \ie for all $\bs u, \bs v \in \cl K \cap \bb B^n$,
$$
\ts |\langle \Amap(\bs u), \bs \Phi \bs v\rangle-\langle \bs \Phi\bs u,\bs \Phi\bs v\rangle| \leq \delta \sqrt m \|\bs \Phi \bs v\| \leq L \delta m,
$$
and observing that $\delta < (1+\delta)\epsilon$ if $\epsilon > 1$ (\ie $\epsilon > 1/2$ before the rescaling). 

\end{proof}

\section{PBP error decay analysis in a few special cases}
\label{sec:specialcasePBP}

We here demonstrate the results announced at the end of the Introduction, \ie the rate of the PBP error decay for an increasing number of measurements in the estimation of low-complexity signals sensed by the quantized random mapping defined in~\eqref{eq:Uniform-dithered-quantization}. This is achieved by combining the general guarantees of Sec.~\ref{sec:PBP-gen-error} with conditions ensuring, \whp, the RIP of a random sensing matrix $\tinv{\sqrt m}\bs \Phi$ generated by a RIP matrix distribution, the (L)LPD property of $\Amap$, and the Lipschitz continuity of the random matrix $\tinv{\sqrt m}\bs \Phi$, as imposed by Prop.~\ref{prop:main-result-embed} and Cor.~\ref{cor:llpd-noisy-linear}.

\subsection{Conditions for Lipschitz continuity} 
\label{sec:cond-lipsch-cont}

As explained in Sec.~\ref{sec:LPD-noisy-linear-map}, when $\cl K$ is a cone and if $\tinv{\sqrt m}\bs\Phi$ respects the RIP$(\cl K-\cl K,\epsilon)$ for some $0<\epsilon\leq 1$, this mapping $\tinv{\sqrt m}\bs\Phi$ is trivially $(\eta, \sqrt 2)$-Lipschitz continuous for any $\eta > 0$ from an easy use of~\eqref{eq:RIP-cone}. The next proposition shows that this remains true if $\cl K$ is bounded and star-shaped (\eg if $\cl K$ is a bounded, convex and symmetric set), and if $\bs \Phi \in \bb R^{m \times n}$ is generated by a RIP matrix distribution.

\begin{proposition}
  \label{prop:Lipschitz-bounded-cvxset} Let $\cl K$ be a bounded, star-shaped set with radius $\|\cl K\|=1$. Given $\eta >0$ and a probability of failure $0<\zeta<1$, if the random matrix $\tinv{\sqrt m}\bs \Phi \in \bb R^{m \times n}$ is drawn from a RIP matrix distribution $\DRIP$ over the \emph{local set} $\cl K_{\loc}^{(\eta)} := (\cl K-\cl K) \cap \eta \bb B^n$, then, provided that
\begin{equation}
  \label{eq:Lipschitz-bounded-cvxset}
  \ts m\ \gtrsim\ \min(1,\eta)^{-2}\, w(\cl K)^2\, \cl P_{\log}(m, n, 1, 1/\zeta),  
\end{equation}
$\tinv{\sqrt m}\bs\Phi$ is $(\eta, \sqrt 2)$-Lipschitz continuous over $\cl K$ with probability exceeding $1-\zeta$, \ie $\|\bs \Phi \cl K\| \leq \sqrt{2m}$ and 
\begin{equation}
  \label{eq:Lipschitz-for-subGauss}
\ts \|\bs \Phi \bs u\| \leq \eta \sqrt{2m}, \quad\forall \bs u \in \cl K_{\loc}^{(\eta)}.  
\end{equation}
\end{proposition}

\begin{proof}
The proof is easily established from the context of Def.~\ref{def:RIPGen}. If $\tinv{\sqrt m} \bs \Phi$ is drawn from a RIP matrix distribution $\DRIP(\cl K)$, then, taking $\epsilon=1$ and replacing $\cl K$ by its local set $\cl K^{(\eta)}_{\loc}$ in~\eqref{eq:gen-RIP-lsc-cond}, 
we have 
$$
\ts \big |\frac{1}{m}\|\bs\Phi \bs u\|^2 - \|\bs u\|^2 \big |\ \leq\ \|\cl K^{(\eta)}_{\loc}\|^2,\quad\forall \bs u \in \cl K^{(\eta)}_{\loc},
$$
with probability exceeding $1-\zeta$, provided $m \gtrsim \|\cl K^{(\eta)}_{\loc}\|^{-2}\,w(\cl K_{\loc}^{(\eta)})^2 \,\cl P_{\log}(m,n,1,1/\zeta)$.

The expressions \eqref{eq:Lipschitz-bounded-cvxset} and \eqref{eq:Lipschitz-for-subGauss} are then obtained from the following simplifications. Since $\cl K_{\loc}^{(\eta)} \subset \cl K - \cl K$, we observe that $w(\cl K_{\loc}^{(\eta)}) \leq w(\cl K - \cl K) \leq 2 w(\cl K)$, from the monotonicity and positive homogeneity of the Gaussian mean width. Moreover, $\|\bs u\| \leq \|\cl K_{\loc}^{(\eta)}\| \leq \eta$ for all $\bs u \in \cl K_{\loc}^{(\eta)}$, and $\|\cl K - \cl K \| \geq \|\cl K\| = 1$ if $\bs 0 \in \cl K$. Therefore, $\|\cl K_{\loc}^{(\eta)}\| = \min(\eta, \|\cl K - \cl K\|) \geq \min(\eta, 1)$ by analyzing the value of $\|\cl K_{\loc}^{(\eta)}\|$ for the following three possible cases: $(\cl K - \cl K) \subset \eta \bb B^n$, $(\cl K - \cl K) \supset \eta \bb B^n$ and $(\cl K - \cl K) \cap (\eta \bb B^n)^\compl \neq \emptyset$, where, in this last possibility, $\|(\cl K - \cl K) \cap \eta \bb B^n\| = \eta$ since $\cl K$ is star-shaped.

Finally, $\cl K$ being star-shaped, $\bs 0 \in \cl K$ and ${(\cl K - \cl K)\cap \bb B^n} \supset \cl K$, and setting above $\eta = 1$ provides $\|\bs \Phi \cl K\| \leq \sqrt{2m}$ with the same probability. The final result follows by union bound over this last event and the case~$\eta \neq 1$ (from a simple rescaling of $\zeta$). 
\end{proof}

\subsection{Conditions for (L)LPD} 
\label{sec:conditions-llpd}

We now characterize the conditions ensuring the (L)LPD of the quantized random mapping~$\Amap$ for the two categories of low-complexity sets considered in this work, \ie for structured sets and bounded star-shaped sets (see Sec.~\ref{sec:low-complex-space}).

\paragraph{A.\hspace{3mm} Structured sets}

\begin{proposition}[(L)LPD of $\Amap$ over structured sets and for RIP matrix distributions]
\label{prop:llpd-lsre}
Given a structured subset $\cl K \subset \bb R^n$ for which~\eqref{eq:structured-subset-KE-bound} holds, a distortion $0<\epsilon \leq 1$, a quantization resolution $\delta > 0$, a random matrix $\tinv{\sqrt m} \bs \Phi$ drawn from a RIP matrix distribution $\DRIP(\cl K)$, and a random dither $\bs \xi \sim \cl U^m([0,\delta])$, 
the random mapping $\Amap(\cdot) = \cl Q(\bs \Phi \cdot + \bs \xi)$ respects the LPD$(\cl
K, \bs\Phi, \epsilon(1+\delta))$, or the L-LPD$(\cl K, \bs\Phi, \bs x, \delta\epsilon)$ for a given $\bs x \in \bb R^n$, with probability exceeding $1-\zeta$ provided that
\begin{align}
\label{eq:condm-lsre-llpd}
\ts m&\ts\ \gtrsim\ \epsilon^{-2}\, w( (\cl K-\cl K) \cap \bb B^n)^2\,\log(1 + \frac{1}{\kappa(\epsilon)})\,\cl P_{\log}(m,n,1/\epsilon,1/\zeta), 
\end{align}
where $\kappa(\epsilon)$ equals $\epsilon^3$ for the LPD and $\epsilon$ for the L-LPD property, and $\cl P_{\log}$ is a polylogarithmic function specified by the RIP matrix distribution.  
\end{proposition}
\begin{proof}
For the function $\kappa = \kappa(\epsilon) \leq \epsilon \leq 1$ specified above, to establish the (L)LPD property of $\Amap$, we simply have to show that 
the requirements of Prop.~\ref{prop:main-result-embed} and Cor.~\ref{cor:llpd-noisy-linear} are verified if~\eqref{eq:condm-lsre-llpd} holds in the context of Prop.~\ref{prop:llpd-lsre}. 

First, both Prop.~\ref{prop:main-result-embed} and Cor.~\ref{cor:llpd-noisy-linear} impose that $\tinv{\sqrt m} \bs \Phi$ is a $(\kappa, L)$-Lipschitz continuous embedding for some $L>0$. 
From Def.~\ref{def:RIPGen} and \eqref{eq:gen-RIP-lsc-cond} applied to $(\cl K-\cl K) \cap \bb B^n$ with $\|(\cl K-\cl K) \cap \bb B^n\| = 1$, if~\eqref{eq:condm-lsre-llpd} holds, then $\tinv{\sqrt m} \bs \Phi$ satisfies the RIP$(\cl K - \cl K, \epsilon)$ with probability exceeding $1-\zeta$ since $\tinv{\sqrt m} \bs \Phi \sim \DRIP(\cl K)$. From~\eqref{eq:RIP-cone} and the considerations at the beginning of Sec.~\ref{sec:cond-lipsch-cont}, this mapping is thus $(\kappa, \sqrt 2)$-Lipschitz continuous since $0< \epsilon \leq 1$.  

Second, we note that both~\eqref{eq:cond-main-result-embed-unif} and~\eqref{eq:llpd-noisy-linear-cond-m} require $m \gtrsim \epsilon^{-2}\, \cl H(\cl K\cap \bb B^n,\, c \kappa)$. 
However, for a structured set $\cl K$, we have 
$$
\ts \cl H(\cl K\cap \bb B^n,\, c \kappa)\ \lesssim\ w(\cl K\cap \bb B^n)^2\,\log(1+\frac{1}{\kappa}).
$$
Since $\cl K$ is a cone, $\bs 0 \in \cl K$, $\cl K \subset \cl K - \cl K$, $w(\cl K\cap \bb B^n) \leq w((\cl K - \cl K)\cap \bb B^n)$, and $\cl P_{\log} \geq 1$ (by assumption, see Def.~\ref{def:RIPGen}), we observe that~\eqref{eq:cond-main-result-embed-unif} and~\eqref{eq:llpd-noisy-linear-cond-m} are satisfied if~\eqref{eq:condm-lsre-llpd} holds. A painless rescaling of $\epsilon$ in the LPD case concludes the proof.
\end{proof}

\paragraph{B.\hspace{3mm}Bounded, star-shaped sets:} For an arbitrary bounded, star-shaped set $\cl K \subset \bb B^n$, any random matrix drawn from a RIP matrix distribution $\DRIP(\cl K)$ verifies \whp the RIP$(\cl K, \epsilon)$ with controlled distortion $\epsilon>0$ provided $\epsilon^2 m$ is larger than $w(\cl K)^2$, up to a polylogarithmic factor (see Sec.~\ref{sec:cond-ensur-rip}). Moreover, thanks to Prop.~\ref{prop:Lipschitz-bounded-cvxset}, the same random constructions are shown to be $(\eta, \sqrt 2)$-Lipschitz continuous with high probability as soon as $\eta^2 m$ is also larger than $w(\cl K)^2$, up to a polylogarithmic factor. Therefore, following Prop.~\ref{prop:main-result-embed} or Cor.~\ref{cor:llpd-noisy-linear}, it remains to bound the Kolmogorov entropy of $\cl K$ in order to determine when the quantized random mapping $\Amap$ satisfies the (L)LPD, \ie using Sudakov's inequality \eqref{eq:sudakov}. 

\begin{proposition}[(L)LPD of $\Amap$ over bounded, star-shaped set and for RIP matrix distributions]
\label{prop:llpd-lsre-convex}
Given a bounded, star-shaped set $\cl K \subset
\bb B^n$ with $\|\cl K\|=1$, a distortion $0<\epsilon\leq 1$, a quantization resolution $\delta > 0$, a matrix $\tinv{\sqrt m} \bs \Phi \in \bb R^{m \times n}$ drawn from a RIP matrix distribution $\DRIP(\cl K)$, and a random dither $\bs \xi \sim \cl U^m([0,\delta])$,
the random mapping $\Amap(\cdot) = \cl Q(\bs \Phi \cdot + \bs \xi)$ respects the LPD$(\cl K, \bs\Phi, \epsilon(1+\delta))$, or the L-LPD$(\cl K, \bs \Phi, \bs x, \epsilon\delta)$ for a given $\bs x \in \bb R^n$, with probability exceeding $1-\zeta$ provided the following requirement is satisfied 
\begin{align}
\label{eq:condm-llpd-sgsors-convex}
\ts m&\ts \gtrsim\ \tinv{(\epsilon\kappa)^{2}}\, w(\cl K)^2\, \cl P_{\log}(m,n,1,1/\zeta),  
\end{align}
where $\kappa(\epsilon)$ is defined in Prop.~\ref{prop:llpd-lsre}.  
\end{proposition}
\begin{proof}
To establish the (L)LPD property of $\Amap$, we simply verify that the requirements of Prop.~\ref{prop:main-result-embed} and Cor.~\ref{cor:llpd-noisy-linear} hold if~\eqref{eq:condm-llpd-sgsors-convex} is respected in the context of Prop.~\ref{prop:llpd-lsre-convex}. 
First, both Prop.~\ref{prop:main-result-embed} and Cor.~\ref{cor:llpd-noisy-linear} impose that $\tinv{\sqrt m} \bs \Phi$ is $(\kappa, L)$-Lipschitz continuous over $\cl K$ for some $L >0$. 
Since $\tinv{\sqrt m} \bs \Phi \sim \DRIP(\cl K)$, from Prop.~\ref{prop:Lipschitz-bounded-cvxset}, this map is $(\kappa, \sqrt 2)$-Lipschitz continuous with probability exceeding $1-\zeta$ if 
$$
\ts m\ \gtrsim\ \kappa^{-2}\, w(\cl K)^2\, \cl P_{\log}(m, n, 1, 1/\zeta).  
$$
Second, since $\cl K \subset \bb B^n$, we note that both~\eqref{eq:cond-main-result-embed-unif} and~\eqref{eq:llpd-noisy-linear-cond-m} require $m \gtrsim \epsilon^{-2}\, \cl H(\cl K,\, c \kappa)$. 
We conclude the proof by observing that \eqref{eq:condm-llpd-sgsors-convex} involves both~\eqref{eq:cond-main-result-embed-unif} and~\eqref{eq:llpd-noisy-linear-cond-m} since $\ts \cl H(\cl K,\, c \kappa)\ \lesssim\ \kappa^{-2} w(\cl K)$ by Sudakov's inequality, $\epsilon \leq 1$, and $\cl P_{\log} \geq 1$ by assumption.
\end{proof}

\subsection{Decay rate of the PBP reconstruction error}
\label{sec:PBP-gen-error-decay}

Thanks to the three previous sections, we are now ready to analyze the rate of the PBP error decay in the context of the QCS model \eqref{eq:Uniform-dithered-quantization} when the number of measurements increases. Following the general analysis of Sec.~\ref{sec:PBP-gen-error}, we consider three types of low-complexity sets: the union of low-dimensional subspaces (ULS), low-rank matrices, and sets that are bounded, convex and symmetric; the first two cases being structured.

\paragraph{A.\hspace{3mm} ULS and low-rank models:} In the case where $\bs x \in \cl K \cap \bb B^n$, with a structured low-complexity set $\cl K \subset \bb R^n$ being either a ULS or the (vectorized) set $\lrset{r}$ of rank-$r$ matrices, Theorems~\ref{thm:PBP-Union-Subspace} and~\ref{thm:PBP-Low-Rank} establish that the reconstruction error of the estimate $\hat{\bs x}$ provided by PBP from the QCS observations of $\bs x$ is bounded as $\|\bs x - \hat{\bs x}\|\ \leq\ 2(\epsilon + \nu)$. Therefore, if $\tinv{\sqrt m}\bs \Phi$ respects the RIP$(\cl K - \cl K, \bs\Phi, \epsilon)$ for some $0<\epsilon\leq 1$, and if $\Amap(\cdot) = \cl Q(\bs \Phi \cdot + \bs \xi)$ satisfies the (L)LPD$(\cl K - \cl K, \bs\Phi, \nu)$, with $\nu = \epsilon\,(1+\delta)$ for the LPD and $\nu = \delta\epsilon$ for the L-LPD, respectively, we find 
\begin{equation}
  \label{eq:adapted-decay-struct-set}
  \|\bs x - \hat{\bs x}\|\ \leq\ C \epsilon\,(1+ \delta).  
\end{equation}

From Def.~\ref{def:rip-def}, any random matrix $\tinv{\sqrt m}\bs\Phi \sim \DRIP(\cl K)$ (\eg sub-Gaussian random matrix, partial random orthonormal matrix, BOS or SORS) satisfies the RIP$(\cl K - \cl K, \epsilon)$ with probability
exceeding $1-\zeta$ if
\begin{equation}
  \label{eq:cond-tmp-1}
  \ts m\ \gtrsim \epsilon^{-2}\,w( (\cl K - \cl K) \cap \bb B^n)^2\,\cl P_{\log}(m, n, 1/\epsilon, 1/\zeta),
\end{equation}
for some polylogarithmic function $\cl P_{\log}$ depending on $\DRIP$.

Concerning the (L)LPD of $\Amap$, from Prop.~\ref{prop:llpd-lsre}, if 
\begin{equation}
  \label{eq:cond-tmp-2}
\ts m \gtrsim \epsilon^{-2}\, {w}(\cl K^{\splus 4} \cap \bb B^n)^2 \log(1
+ \frac{1}{\kappa}) \cl P_{\log}(m,n, 1/\epsilon, 1/\zeta),   
\end{equation}
with $\kappa$ defined in Prop.~\ref{prop:llpd-lsre}, then the mapping $\Amap$, whose randomness depends on both $\tinv{\sqrt m}\bs \Phi \sim \DRIP(\cl K)$ and the random dither $\bs \xi \sim \cl U^m([0, \delta])$, satisfies the LPD$(\cl K-\cl K, \bs\Phi, \nu)$ or the L-LPD$(\cl K-\cl K, \bs\Phi, \bs x, \nu)$ with probability exceeding $1-\zeta$.

Consequently, since $\cl K - \cl K \subset \cl K^{\splus 4}$, \eqref{eq:cond-tmp-2} involves \eqref{eq:cond-tmp-1}, and the event where both the (L)LPD and the RIP hold occurs with probability exceeding $1-2\zeta$ by union bound. 
Taking the minimal $m$ that satisfies \eqref{eq:cond-tmp-2} and inverting the relation between $m$ and $\epsilon$ provide
$$
\epsilon = O\big(m^{-\frac{1}{2}} {w}( \cl K^{\splus 4} \cap \bb B^n)\big),
$$
up to some log factors\footnote{By reasonably assuming that all factors on the right hand side of \eqref{eq:cond-tmp-2} are bigger than one, we get $\epsilon^2 = \Omega(1/m)$ and all log factors depending on $1/\epsilon$ can be replaced by $\log(m)$, up to some constants.} in $m$, $n$, $1/\delta$ and $1/\zeta$. 
\medskip

Therefore, from \eqref{eq:adapted-decay-struct-set}, the value of $\nu$, Theorems~\ref{thm:PBP-Union-Subspace} and~\ref{thm:PBP-Low-Rank}, we can conclude the following fact:
\begin{quote}
\em Let $\cl K$ be either a ULS or a low-rank model. If $\tinv{\sqrt m} \bs \Phi \in \bb R^{m\times n}$ is generated by a RIP matrix distribution and if $\bs \xi \sim \cl U^m([0,\delta])$, then, with high probability and up to some missing log factors, PBP produces for all vectors $\bs x \in \cl K \cap \bb B^n$ an estimate $\hat{\bs x}$ from $\bs y = \Amap(\bs x) = \cl Q(\bs \Phi \bs x + \bs \xi)$ whose error has the following decay rate when $m$ increases:
\begin{equation}
  \label{eq:error-decay-lsre-uls-low-rank}
\ts \|\bs x - \hat{\bs x}\| = O\big(\frac{1 + \delta}{\sqrt m}\, w( \cl K^{\splus 4} \cap \bb B^n)\big).  
\end{equation}
Moreover, up to minor changes in the log factors, this rate is preserved in the non-uniform case, \ie if $\bs x$ is fixed.
\end{quote}
In other words, up to some log factors and up to a multiplicative factor depending on the structure of $\cl K$, the reconstruction error decreases like $O( (1+ \delta)\,m^{-1/2})$ when $m$ increases, non-uniformly or uniformly for all elements of $\cl K \cap \bb B^n$. 

\paragraph{B.\hspace{3mm}Bounded, convex and symmetric sets:} \label{par:BCS-set-decay-rate} If $\bs x$ belongs to a bounded, convex and symmetric set $\cl K \subset \bb B^n$ with $\|\cl K\|=1$, PBP can still estimate this signal from $\bs y = \Amap(\bs x) = \cl Q(\bs\Phi \bs x + \bs \xi)$ with $\bs \xi \sim \cl U^m([0, \delta])$ and $\tinv{\sqrt m}\bs\Phi \sim \DRIP(\cl K)$.  In this case, the reconstruction still decays when $m$ increases. However, as made clear below, the decay rate is slower than that achieved for signals belonging to a structured set. Moreover, there is a significant difference between the rate of the uniform reconstruction guarantees, \ie valid for all low-complexity signals given one couple $(\bs \Phi, \bs \xi)$, and the rate of the non-uniform error decay determined on a fixed signal $\bs x$.  

Actually, if $\tinv{\sqrt m}\bs \Phi$ respects the RIP$(\cl K, \epsilon)$ for some $0<\epsilon\leq 1$, and if the quantized mapping $\Amap$ verifies the LPD$(\cl K, \bs \Phi, \nu)$ or the L-LPD$(\cl K, \bs \Phi, \bs x, \nu)$, with $\nu = \epsilon\,(1+\delta)$ for the LPD or $\nu = \delta\epsilon$ for the L-LPD, Theorem~\ref{thm:PBP-convex-set} shows that
\begin{equation}
  \label{eq:rec-nu-eps-tmp}
  \|\bs x - \hat{\bs x}\|\ \leq\ (4 \epsilon + 2\nu)^{1/2}\ \leq\ C \big(\epsilon\,(1+\delta)\big)^{1/2},
\end{equation}
with the PBP estimate $\hat{\bs x}$. 

Concerning the RIP, if $\tinv{\sqrt m}\bs \Phi$ is generated by a RIP matrix distribution $\DRIP$ (See Def.~\ref{def:RIPGen}) with 
\begin{equation}\label{eq:RIP-RIPGen}
\ts m\ \gtrsim \epsilon^{-2}\,w(\cl K)^2\,\cl P_{\log}(m, n, 1/\epsilon, 1/\zeta),
\end{equation}
for some polylogarithmic function $\cl P_{\log}$ depending on $\DRIP$, then this mapping is RIP$(\cl K, \epsilon)$ with probability exceeding $1-\zeta$.

Moreover, for $\tinv{\sqrt m}\bs \Phi \sim \DRIP(\cl K)$ and $\bs \xi \sim \cl U^m([0, \delta])$, Prop.~\ref{prop:llpd-lsre-convex} explains that $\Amap(\cdot) := \cl Q(\bs \Phi \cdot + \bs \xi)$ respects the LPD$(\cl K,\bs\Phi, \nu)$ or the L-LPD$(\cl K, \bs\Phi, \bs x, \nu)$ with probability exceeding $1-\zeta$ provided 
\begin{equation}
  \label{eq:tmp-lpd-rip-convex}
  \ts m \ts \gtrsim\ \inv{(\epsilon\kappa)^{2}}\, w(\cl K)^2\, \cl P_{\log}(m, n, 1, 1/\zeta).  
\end{equation}
with $\kappa$ defined in Prop.~\ref{prop:llpd-lsre}. Therefore, provided $m \ts \gtrsim\ \inv{(\epsilon\kappa)^{2}}\, w(\cl K)^2\, \cl P_{\log}(m, n, 1/\epsilon, 1/\zeta)$, both 
\eqref{eq:tmp-lpd-rip-convex} and \eqref{eq:RIP-RIPGen} hold, and we can thus achieve both the RIP and the (L)LPD with probability exceeding $1-2\zeta$ (by union bound). 

Finally, as for the structured set analysis, by saturating this last requirement on $m$ and inverting the relation between $m$ and $\epsilon$ (via $\epsilon\kappa$), we can conclude the following fact from \eqref{eq:rec-nu-eps-tmp}.
\begin{quote}
\em Let $\cl K$ be a bounded, convex and symmetric set with $\|\cl K\| = 1$. If $\tinv{\sqrt m} \bs \Phi \in \bb R^{m\times n} \sim \DRIP(\cl K)$ and $\bs \xi \sim \cl U^m([0,\delta])$, then, with high probability and up to some log factors, PBP gives for all vectors $\bs x \in \cl K$ an estimate $\hat{\bs x}$ from $\bs y = \Amap(\bs x) = \cl Q(\bs \Phi \bs x + \bs \xi)$ with the following reconstruction error decay when $m$ increases:
\begin{equation}
  \label{eq:error-decay-convex-unif}
  \ts \|\bs x - \hat{\bs x}\| = O\big( (1 + \delta)^{\inv{2}} (\frac{w(\cl K)^2}{m})^{\inv{16}} ).
\end{equation}
In the case of a non-uniform reconstruction guarantee (\ie $\bs x$ is fixed), this rate becomes
\begin{equation}
  \label{eq:error-decay-convex-fixed}
  \ts \|\bs x - \hat{\bs x}\| = O\big( (1 + \delta)^{\inv{2}} (\frac{w(\cl K)^2}{m})^{\inv{8}}).  
\end{equation}
with minor changes in the hidden log factors.
\end{quote}
In other words, up to some log factors and up to a multiplicative factor depending on the structure of $\cl K$, the reconstruction error decreases like $O(m^{-\inv{p}})$ when $m$ increases, with $p=16$ if this must hold for all elements of $\cl K$, and $p=8$ if $\bs x$ is fixed (non-uniform case).

\paragraph{C.\hspace{3mm} Sub-optimality of the decay rate of the PBP reconstruction error:}
  Clearly, the loss of information induced by quantization prevents the exact reconstruction of a signal observed by a QCS system \cite{BJKS2015}. However, the decay rates of the PBP reconstruction error in \eqref{eq:error-decay-lsre-uls-low-rank}, \eqref{eq:error-decay-convex-unif} and \eqref{eq:error-decay-convex-fixed} do not vanish when the quantization distortion decreases, \ie when~${\delta \to 0}$.

This phenomenon is actually consistent with the reconstruction error reached by PBP in the absence of quantization. Indeed, an easy adaptation of the proofs of Thm~\ref{thm:PBP-Union-Subspace} (for signals in a union of subspaces), Thm~\ref{thm:PBP-Low-Rank} (for low-rank matrices), and Thm~\ref{thm:PBP-convex-set} (for bounded convex sets) shows that if the general distorted sensing model \eqref{eq:distorted-CS-model} reduces to $\bs y = \bs \Phi \bs x$ (\ie $\Dmap = \bs \Phi$), the reconstruction errors provided by these theorems are simply achieved by setting $\nu=0$; the LPD vanishing in this case (see also Rem.~\ref{rem:delta-0} for $\Dmap = \Amap$ when $\delta \to 0$).  This analysis is not new and corresponds to classical results of the CS literature, \eg as observed from the bound obtained on the first iteration of the iterative hard thresholding algorithm for the recovery of sparse signals, \ie the PBP of $\bs \Phi \bs x$ with $\bs x \in \cl K = \spset{k}$ (see, \eg \cite{BM2009}, \cite[Thm 3.5]{foucart2011hard}).

Note, however, that other recovery schemes do not have this drawback. For instance, for any matrix $\bs \Phi$ satisfying the RIP over sparse signals, by using the basis pursuit denoise program \cite{BJKS2015} (or its variant imposing the consistency of the estimate with the observed quantized measurements \cite{dai2009distortion,moshtaghpour2016consistent,DJR2017}), one can recover sparse signals observed by the QCS model \eqref{QCSproblem} (with or without dithering) with an error bound that is linear in $\delta$. Furthermore, for signals in the unit $\ell_2$-ball, and for randomly subsampled Gaussian circulant matrices combined with a dither formed by adding a Gaussian and a uniform random vector \cite{DJR2017}, the reconstruction error of the same signal estimate (additionally constrained to lie in the unit $\ell_2$-ball) reaches the decay rate~$O(\delta\,(s/m)^{1/6})$, which tends to zero both if $\delta$ vanishes or if $m$ is large.

We thus conclude this section with this open problem: \emph{Given an algorithm $\cl A$ (\eg basis pursuit) known to exactly reconstruct a low-complexity signal from the compressive observations achieved with an arbitrary RIP matrix $\bs \Phi$, can we adapt $\cl A$ to handle signal observations produced by the dithered QCS model \eqref{eq:Uniform-dithered-quantization} (\eg by introducing a consistency constraint) so that its reconstruction error must decay like $O(\delta m^{-1/p})$ for some $p > 0$ when $m$ increases?}

\section{Numerical verifications}
\label{sec:numericalPBP}

In this section, we study numerically the behavior of the PBP reconstruction error for signals belonging to the three low-complexity sets discussed in the Sec.~\ref{sec:specialcasePBP}, namely, the set of sparse vectors, the set of compressible signals and the set of low-rank matrices, when $m$ or~$\delta$ increases.  We carried out this analysis for several sensing matrices respecting the RIP property, \eg for sub-Gaussian random matrices and for one partial random orthonormal matrix construction, \ie a random partial DCT sensing matrix. Moreover, we empirically demonstrate the importance of the dither by observing how the reconstruction error is impacted when the dither is removed. 

\paragraph{A.\hspace{3mm}Comparison between three low-complexity sets:} 

\begin{figure}[!t]
  \centering
  \subfloat[\label{fig:gauss-sensing-sparse} $4$-sparse signals in $\spset{4}$]{\includegraphics[width=0.33\textwidth]{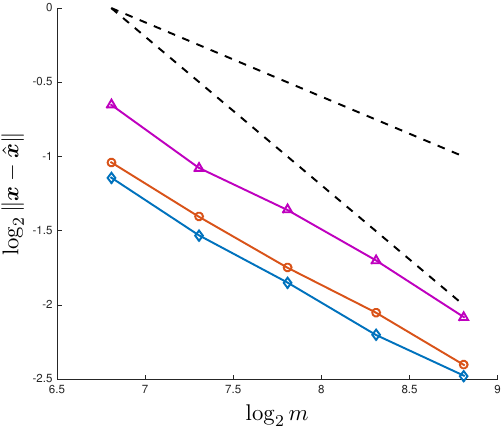}}
  \subfloat[\label{fig:gauss-sensing-comp} Compressible signals in $\cpset{4}$]{\includegraphics[width=0.33\textwidth]{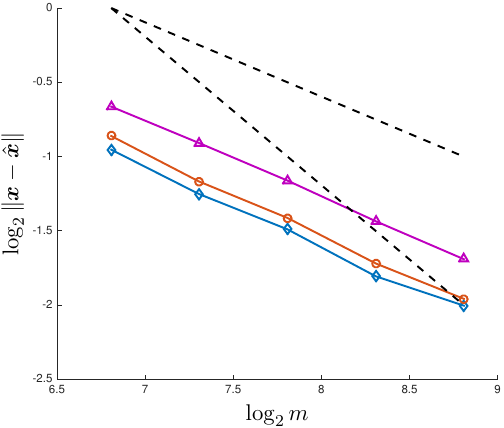}}
  \subfloat[\label{fig:gauss-sensing-lrank} rank-2 matrices]{\includegraphics[width=0.33\textwidth]{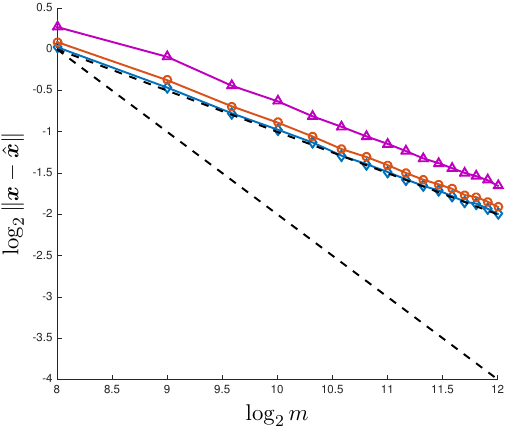}}
  \caption{PBP reconstruction error in function of $m$ (in a log-log plot) for the dithered QCS observations of
  low-complexity signals. For each plot, we present the reconstruction
  error for $\delta = 0.5$ (blue diamonds), $\delta=1$
  (orange circles) and $\delta=2$ (pink triangles). The two guiding dashed
  lines represent the two decaying rates as $m^{-1/2}$ and $m^{-1}$.}
  \label{fig:gauss-sensing}
\end{figure}

We first evaluate the PBP reconstruction error when $m$ increases for different values of the quantization resolution $\delta$, in the case where $\tinv{\sqrt m}\bs \Phi \in \bb R^{m \times n}$ is a Gaussian random matrix and for $\bs \xi \sim \cl U^m([0, \delta])$. For the QCS of $k$-sparse signals in $\spset{k}$ and compressible signals in $\cpset{k}$, we have set $n=512$, $k=4$ with $m \in [4k\log_2(n/k),\,n]$, where $4k\log_2(n/k)$ is an estimate of the minimal number of measurements yielding perfect reconstruction of $k$-sparse signals in absence of quantization (\eg using BPDN~\cite{candes2006stable}). Each $k$-sparse signal was generated by first picking its support uniformly at random among the $n \choose k$ possible $k$-length supports of~$[n]$, and then drawing each of its non-zero elements \iid from a standard normal distribution.
Each compressible signal $\bs x$ was generated by applying a random permutation to a vector whose components have a random sign (with equal probability) with an absolute value decaying like $i^{-\alpha_k}$ ($i \in [n]$), with $1<\alpha_k <2$ such that $\|\bs x\|_1 / \|\bs x\| \leq \sqrt k$. The resulting signal is further normalized to have unit $\ell_2$-norm.

Fig.~\ref{fig:gauss-sensing-sparse} and
Fig.~\ref{fig:gauss-sensing-comp} display the reconstruction errors of
PBP from the dithered quantized observations of $4$-sparse and compressible signals in $\cpset{4}$, respectively, as a function of $m$ (in a log-log plot). For each figure, three curves are given for $\delta \in \{0.5,1,2\}$. Two guiding dashed lines are also provided to represent the decay rates as $m^{-1/2}$ and $m^{-1}$. For every $\delta$ and $m$, the PBP reconstruction was tested over 100 trials of the random generation of $\bs \Phi$, $\bs \xi$ and $\bs x$. When $m$ increases, we clearly see that the decay rate of the reconstruction of sparse signals is slightly faster than $O(m^{-1/2})$ (\eg the curve at $\delta=1$ is well fitted by $O(m^{-0.67})$), as predicted by~\eqref{eq:error-decay-lsre-uls-low-rank}, while the decay rate of the error reconstruction closely matches $O(m^{-1/2})$ for compressible signals. This decay is faster than the fastest rate in $O(m^{-1/8})$ predicted in Sec.\ref{par:BCS-set-decay-rate}(B) for non-uniform estimation. We note, however, that Sec.\ref{par:BCS-set-decay-rate}(B) provided bounds valid for any convex sets $\cl K \subset \bb B^n$, not only the set of compressible signals. 

\medskip
For low-rank matrices, following Sec.~\ref{sec:PBP-low-rank}, we have analyzed the decay rate of the PBP reconstruction error for rank-$r$ $n_1\!\times\!n_2$-matrices with $n_1 = n_2 = 64$, \ie $n=n_1n_2=4096$, and $r=2$.  Each rank-2 matrix $\bs X$ was generated from the model $\bs X = \bs B \bs C^\top\!\!/\|\bs B \bs C^\top\|_F$ where $\bs B, \bs C \in \bb R^{\sqrt n \times 2}$ are two random matrices with entries \iid as a standard normal distribution. In this context, the random Gaussian sensing matrix $\tinv{\sqrt m}\bs\Phi$ thus operated over the vectorized form $\bs x = \ve(\bs X)$ of each low-rank matrix $\bs X$, before the dithered quantization defining the QCS sensing in~\eqref{eq:Uniform-dithered-quantization}.  Fig.~\ref{fig:gauss-sensing-lrank} displays the reconstruction error of the PBP estimate $\hat{\bs x} = \ve(\hat{\bs X})$ --- with an average over 50 trials with different generation of $\bs \Phi$, $\bs \xi$ and $\bs X$ --- when $m \in [r\,(n_1+n_2), n]$ and for $\delta \in \{0.5, 1, 2\}$. As predicted in~\eqref{eq:error-decay-lsre-uls-low-rank} for structured sets, the observed decay rate closely matches $O(m^{-1/2})$.

\paragraph{B.\hspace{3mm}Evolution of the error with $\delta$:}

\begin{figure}[!t]
  \centering
  \null\hfill\subfloat[\label{fig:evol-gauss-sensing-delta-sparse} $4$-sparse signals]{\includegraphics{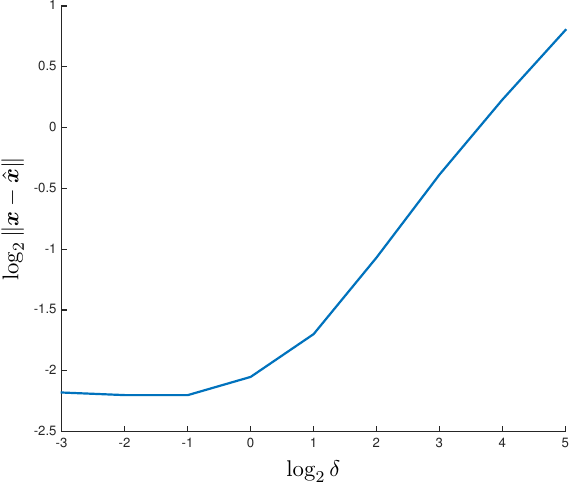}}\hfill
  \subfloat[\label{fig:evol-gauss-sensing-delta-comp} Compressible signals in $\cpset{4}$]{\includegraphics{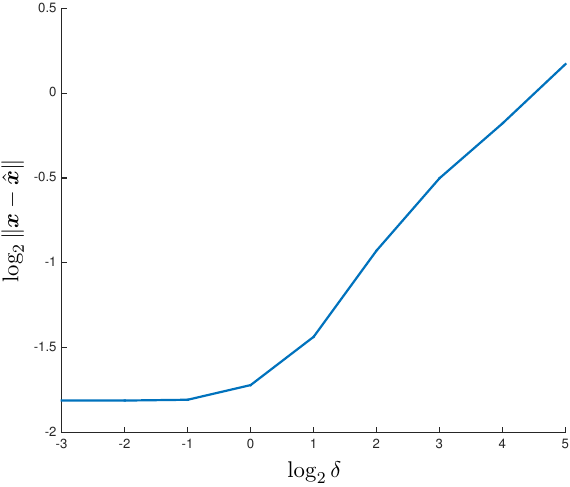}}\hfill\null
  \caption{PBP reconstruction error in function of $\delta$ (in a log-log plot) for the dithered QCS observations of sparse and compressible vectors.}
  \label{fig:evol-gauss-sensing-delta}
\end{figure}

In this experiment, sparse and compressible signals are generated according to the same setting as above and we test the evolution of the PBP reconstruction error as a function of $\delta$, with $\log_2 \delta \in \{-3,-2, \cdots, 5\}$ and $m=n/2$. The results are plotted in Fig.~\ref{fig:evol-gauss-sensing-delta-sparse} and Fig.~\ref{fig:evol-gauss-sensing-delta-comp} in the case of 4-sparse signals and compressible signals in $\cpset{4}$, respectively. Interestingly, we observe that both curves are compatible with the bound $\|\bs x -\hat{\bs x}\| = O(1+\delta)$ (as all other parameters are fixed), with some floor at small values of $\delta$ corresponding to the error achieved by PBP in the unquantized CS regime ($\delta = 0$).   

\paragraph{C.\hspace{3mm}Analysis of different random sensing matrices:}

\begin{figure}[!t]
  \centering
  \null\hfill\subfloat[\label{fig:Bern-sensing-sparse} Bernoulli random sensing matrix]{\includegraphics{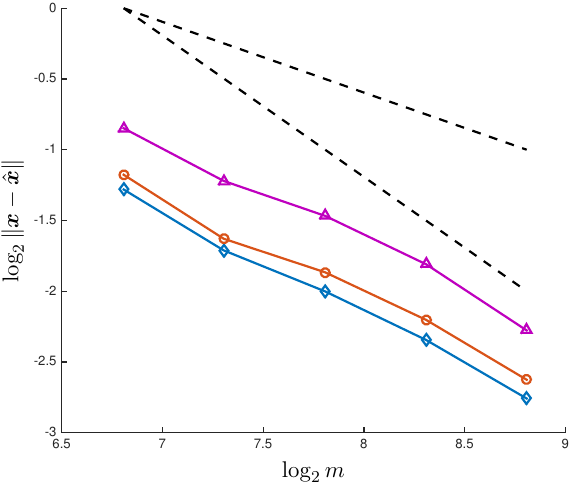}}\hfill
  \subfloat[\label{fig:DCT-sensing-sparse} Random partial DCT sensing matrix]{\includegraphics{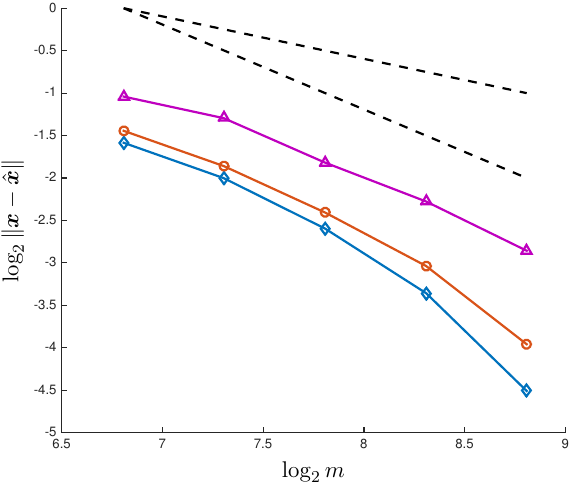}}\hfill\null
  \caption{Reconstruction error of PBP in function of $m$ (in a $\log$-$\log$ plot) for dithered QCS observations obtained with a Bernoulli random matrix and a random partial DCT sensing matrix. For each plot, we present the reconstruction error for $\delta = 0.5$ (blue diamonds), $\delta=1$ (orange circles) and $\delta=2$ (pink triangles). The two guiding dashed lines represent the two decaying rates as $m^{-1/2}$ and~$m^{-1}$.}
  \label{fig:various-mtx-sensing}
\end{figure}

We now test the reconstruction error of PBP for the estimation of $4$-sparse vectors where $\Amap$ is induced by the combination of a dither $\bs \xi \sim \cl U^m([0,\delta])$  with two different non-Gaussian sensing matrices, \ie a Bernoulli random matrix, with \iid entries equal to $\pm 1$ with equal probability, and a random partial DCT matrix obtained by picking (without replacement) $m$ rows uniformly at random among the $n$ rows of an $n\times n$ orthonormal DCT matrix~\cite{CT06,Rau10}.  The values of $n$, $m$, $\delta$, the number of averaged trials as well as the setting of 4-sparse signals in this experiment are the ones used in Sec.~\ref{sec:numericalPBP}.A. Fig.~\ref{fig:Bern-sensing-sparse} and Fig.~\ref{fig:DCT-sensing-sparse} show the reconstruction error of PBP achieved for the Bernoulli and the random partial DCT matrices, respectively. As predicted by theory, these two sensing matrices enjoy similar performances compared to the reconstruction error of PBP with a Gaussian random matrix. Moreover, numerically, a decay rate of $O(m^{-1/2})$ still bounds the reconstruction error of PBP from the dithered QCS observations when $m$ increases.

\paragraph{D.\hspace{3mm}On the importance of dithering:} 

\begin{figure}[!t]
  \centering
  \null\hfill\subfloat[\label{fig:Bern-sensing-sparse-nodith} Bernoulli random sensing matrix]{\includegraphics{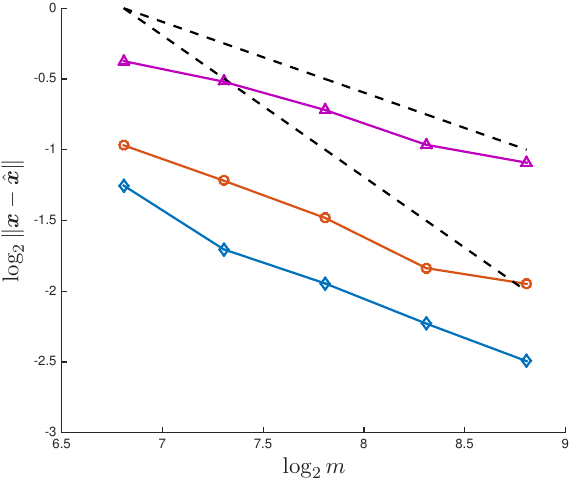}}\hfill
  \subfloat[\label{fig:DCT-sensing-comp-nodith} Random partial DCT sensing matrix]{\includegraphics{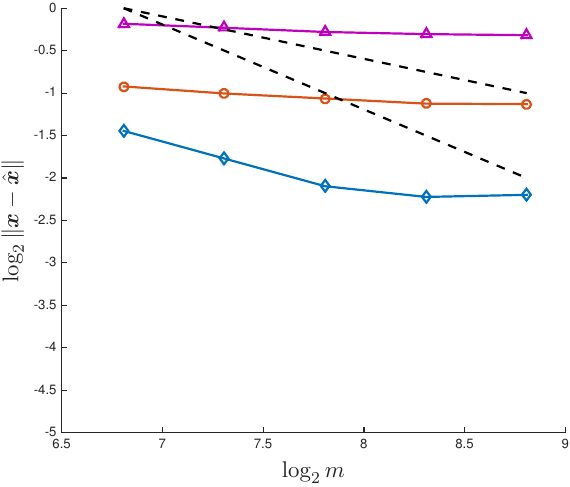}}\hfill\null
  \caption{PBP reconstruction error in function of $m$ (in a log-log plot) for the QCS observations obtained with Bernoulli and random partial DCT sensing matrices \emph{in the absence of dithering}. For each plot, we present the reconstruction error evolution for $\delta = 0.5$ (blue diamonds), $\delta=1$ (orange circles) and $\delta=2$ (pink triangles). The two guiding dashed
  lines represent the two decaying rates as $m^{-1/2}$ and~$m^{-1}$.}
  \label{fig:various-mtx-sensing-no-dithering}
\end{figure}

In this last experiment presented in Fig.~\ref{fig:Bern-sensing-sparse-nodith} and Fig.~\ref{fig:DCT-sensing-comp-nodith}, we remove the dither from the quantized mapping $\Amap$ and test how the reconstruction error of PBP for sparse signals is impacted in the case of Bernoulli and random partial DCT matrices. The generations of both the signals and the sensing matrices are the same as in the previous experiment. In the case of random Bernoulli sensing matrices (Fig.~\ref{fig:Bern-sensing-sparse-nodith}), the decay rate of the reconstruction error of PBP as $m$ increases is clearly slower than that obtained when the dithering is present in $\Amap$ (see Fig.~\ref{fig:Bern-sensing-sparse}). Fig.~\ref{fig:DCT-sensing-comp-nodith} also shows that this error reaches a constant floor for random partial DCT sensing matrices, conversely to the decay rates achieved with dithering (see Fig.~\ref{fig:DCT-sensing-sparse}). In fact, for this sensing (as well as for the Bernoulli matrices, see Sec.~\ref{sec:motivationSOA}), one can build counterexamples of two distinct sparse vectors that are sent to the same quantized observations for any value of~$m$~\cite{feuillen2018quantity}; this ambiguity thus prevents any reconstruction algorithm to reach an arbitrary small error.

\section{Conclusion and perspectives}
\label{sec:conclution}

This work has studied the combination of the compressive sensing of low-complexity signals, as supported by the numerous random matrix constructions now available in the CS literature, with the non-linear distortion induced by a uniform quantizer applied on the compressive signal observations, \eg for transmission or storage purposes. This association is enabled by the addition of a uniform random dither to the linear observations before their quantization. We have proved that this dither enables to estimate, with small error, all or one signal selected in a low-complexity set from its dithered, quantized observations. This estimation is ensured by at least one reconstruction method, the \emph{projected back projection} (PBP), whose reconstruction error is provably decaying when the number of measurements increases. For instance, we have characterized this phenomenon for several ``structured'' low-complexity sets, \eg the set of sparse signals, any union of low-dimensional subspaces, the set of low-rank matrices, or more advanced group-sparse models. In this case, given a quantization resolution $\delta > 0$, the decay rate of the reconstruction error is \whp $O((1+\delta)/\sqrt m)$ (up to log factors) when the number of measurements $m$ increases. For bounded, convex and symmetric sets, \eg for the set of compressible signals, the error is still decaying with $m$ but with the less favorable decay rate of $O((1+\delta)/m^{p})$, with $p=16$ or 8 for uniform or non-uniform reconstruction error, respectively.     

On the side, we have also established in Sec.~\ref{sec:PBP-gen-error} that for more general distorted sensing models, \ie beyond the quantized mapping described above, the reconstruction error of the PBP method can be bounded provided that the associated distorted mapping respects a certain limited projected distortion (LPD) property. This one bounds its discrepancy with a linear mapping assumed to respect the RIP. For instance, as shown in Sec.~\ref{sec:LPD-noisy-linear-map}, linear sensing models corrupted by an additive, sub-Gaussian random noise are quickly shown to satisfy the LPD property.  

Lastly, numerical tests have validated the theoretical reconstruction error bounds in several sensing scenarios for several low-complexity sets, and with a varying number of measurements and different quantization resolutions. In particular, we have empirically demonstrated the positive impact of the dithering in the quantization, especially for non-Gaussian sensing matrices, \eg Bernoulli and random partial DCT matrices. In fact, we have observed numerically that the impact of quantization on the PBP reconstruction error is indeed lessened by the presence of a random dither, with apparent limitations in the absence of dithering for specific sensing matrices (\eg for random partial DCT matrix).     

\medskip
As mentioned in the introduction, PBP can be seen as an initial guide for any advanced reconstruction algorithms in dithered QCS. For instance, PBP can undoubtedly be improved in the sense that its estimate is not \emph{consistent} with the quantized observations, \ie the sensing of this estimate with the same QCS model than the one that observed the true signal is not guaranteed to match the initial quantized observations. 
Therefore, an interesting line of work would be to characterize consistent iterative reconstruction methods whose initialization is equal (or related) to the PBP estimate. Examples of such methods are the quantized iterative hard thresholding (QIHT)~\cite{JDV2015} and the QCOSAMP~\cite{shi2016methods} algorithms, even if no reconstruction error guarantees have been proved for them so far. Speaking of QIHT, given a signal $\bs x$ in a low-complexity set $\cl K \subset \bb R^n$, this algorithm relies on the iterations 
$$
\ts \bs x^{(t+1)} = \cl P_{\cl K}\big(\bs x^{(t)} + \frac{\mu}{m} \bs \Phi^\top(\Amap(\bs x) - \Amap(\bs x^{(t)}))\big),\quad \text{with}\ \bs x^{(1)} = \cl P_{\cl K}\big(\frac{\mu}{m} \bs \Phi^\top \Amap(\bs x) \big), 
$$
where $\bs x^{(1)}$ matches the PBP estimate. A potential proof strategy could be to split the problem into two steps: first, proving that the QIHT algorithm is sure to converge \whp to a consistent solution $\bs x^*$ when it is initialized from the PBP estimate $\bs x^{(1)}$, and second, showing that any pair of signals in $\cl K$ that are consistent with respect to the quantized random mapping $\Amap$, as for $\bs x$ and $\bs x^*$, have a distance bounded by, \eg $O(w(\cl K)/m)$. This last bound, coined the \emph{consistency width} in~\cite{J2016}, is known to decay like $O(1/m)$ when $m$ increases if $\bs \Phi$ is a Gaussian random matrix and if $\cl K = \Sigma_k$ (see~\cite[Theorem 2]{J2016}). Unfortunately, available bounds on the consistency width decay more slowly when $m$ increases for more general RIP matrices\footnote{This is easily observed by enforcing consistency in Prop. 1 and Prop. 2 in~\cite{JC2016}.}~\cite{JC2016}. Knowing if the rate $O(1/m)$ holds for these, as observed numerically for sub-Gaussian sensing matrices~\cite{moshtaghpour2016consistent}, is an open problem.   

Another study could be carried out on the question of allowing other distributions for the generation of the dither, \eg the Gaussian distribution as in~\cite{DJR2017}. As made clear in our work, the uniform distribution cancels out the uniform quantization by expectation (Lemma~\ref{lem1}). However, it should be possible to admit other distribution $\cl D$ such that, if $\xi \sim D$, there exist two constants $0<\mu_0<\mu_1$ for which 
\begin{equation}
  \label{eq:novel-distr-dithering}
\ts \mu_0\,\lambda\ \leq\ \bb E_\xi \cl Q(\lambda + \xi)\ \leq\ \mu_1\, \lambda,\quad\forall \lambda \in \bb R,   
\end{equation}
with $\mu_0=\mu_1=1$ if $\cl D \sim \cl U([0,\delta])$. For a distribution $\cl D$ compatible with~\eqref{eq:novel-distr-dithering}, it should be possible to show that extra distortions impact the reconstruction performance of PBP, with a reduced influence if $\mu_0\approx \mu_1$. 

\section*{Acknowledgements}
\label{sec:acknowledgements}

We wish to thank Simon Foucart for his simplification of the proof of Thm.~\ref{thm:PBP-Low-Rank}, and Augusto Zebadua for interesting discussions on Bussgang's theorem and distorted correlators (see Sec.~\ref{sec:motivationSOA}).

\appendix

\section{Uniform dithering cancels quantization in expectation}
\label{app:vanishing-dithered-quantiz}

\begin{lemma}\label{lem1}
For $\cl Q(\cdot) := \delta \lfloor \cdot / \delta \rfloor$,  any $a\in \bb R$ and $\xi\sim\cl U([0,\delta])$, we have 
\begin{equation}\label{eq1.1}
\ts \bb E_{\xi} \cl Q(a + \xi) = \bb E_{\xi}\delta\lfloor\frac{a+\xi}{\delta} \rfloor =a.
\end{equation}
\end{lemma}

\begin{proof}
Without loss of generality, we set $\delta=1$ and denote $a'=\lfloor a\rfloor$ and $a''=a-a'\in [0,1)$. We can always write $\lfloor a+\xi\rfloor=a'+\lfloor a''+\xi\rfloor=a'+X$, with $X=\lfloor a''+\xi\rfloor$. Since $\xi\sim\cl U([0,1])$, $0\leq a''+\xi<2$ and $X \in \{0,1\}$. Moreover, 
$\bb P(X=0)=\bb P(a''+\xi<1)=\bb P(\xi<1-a'')=1-a''$, and $\bb P(X=1)=\bb P(a''+\xi>1)=\bb P(\xi>1-a'')=a''$.
Therefore, 
$$
\bb E(a'+X) = a'~\bb P(X=0)+(a'+1)~\bb P(X=1) =a'(1-a'')+(a'+1)a'' =a,
$$
and $\bb E_{\xi}\delta\lfloor\frac{a+\xi}{\delta} \rfloor =a$ holds by a simple rescaling argument for $\delta>0$. 
\end{proof}

\section{Moments of dithered QCS seen as a non-linear sensing model}
\label{app:values-from-Klasso}

Writing $\bs \Phi = (\bs \varphi_1,\,\cdots,\bs\varphi_m)^\top \in \bb R^{m \times n}$, the (scalar) QCS model~\eqref{eq:Uniform-dithered-quantization} can be seen as a special case of the more general, non-linear sensing model
\begin{equation}
  \label{eq:gen-non-lin-sensing}
  \ts y_i = f_i(\scp{\bs \varphi_i}{\bs x}),\quad i\in [m],
\end{equation}
with the random functions $f_i:\bb R \to \bb R$ being \iid as a random function $f: \bb R \to \bb R$ for $i\in [m]$. This is observed by setting $f_i(\lambda) = \cl Q(\lambda + \xi_i)$. 

In~\cite{PVY2017}, the authors proved that, for a Gaussian random matrix and a bounded star-shaped set $\cl K$, provided $f$ leads to finite $\mu := \bb E f(g)g$, $\eta^2 = \bb E f(g)^2$ and $\psi$ is the sub-Gaussian norm of $f(g)$ with $g\sim \cl N(0,1)$, \ie $\psi := \|f(g)\|_{\psi_2}$~\cite{V2012}, one can estimate the direction of $\bs x \in \cl K$ from the solution $\hat{\bs x}$ of the PBP of $\bs y$ defined in~\eqref{eq:gen-non-lin-sensing}. 

In particular,~\cite[Theorem 9.1]{PVY2017} shows that, given $\bs x \in \cl K$, $t > 0$, $0<s<\sqrt m$, the estimate of PBP satisfies, with probability exceeding $1 - 2\exp(-c s^2 \eta^4 / \psi^4)$,
$$
\ts \|\hat{\bs x} - \mu \|\bs x\|^{-1} \bs x\| \leq t + \frac{4\eta}{\sqrt m}(s + \frac{w_t(K)}{t}).   
$$  

This result can adapted to the context of this work. In particular, given some distortion $0<\epsilon < \frac{9\eta^2}{\psi^2}$, setting above $s=\frac{\epsilon \psi^2}{9\eta^2} \sqrt m$, $t=\frac{4 \psi^2}{9\eta}\epsilon $, and using the fact that $w_t(\cl K) \leq w(\cl K)$, provided 
$$
\ts m \geq \frac{9^4\eta^6}{\epsilon^{4}\psi^8} w^2(\cl K),
$$ 
it is easy to see that the same estimate satisfies, with probability exceeding $1 - 2\exp(-c \epsilon^2 m)$,
$$
\ts \|\hat{\bs x} - \mu \|\bs x\|^{-1} \bs x\| \leq \frac{\psi^2}{\eta} \epsilon.   
$$  

For the scalar QCS model~\eqref{eq:Uniform-dithered-quantization} where $f(\lambda) := \delta \lfloor (\lambda + \xi) /\delta \rfloor$ for $\xi \sim \cl U([0,\delta])$, using the law of total expectation and $\bb E f(\lambda) = \bb E \delta \lfloor (\lambda + \xi)/\delta \rfloor = \lambda$ (see Lemma~\ref{lem1},), we compute that, for $g\sim \cl N(0,1)$, 
\begin{align*}
 \mu&\ts := \bb E f(g)g = \delta \bb E_g \bb E_\xi \lfloor (g + \xi) /\delta \rfloor g = \bb E_g g^2 = 1,\\
\ts \psi&\ts := \|\delta \lfloor (g + \xi) /\delta\rfloor\|_{\psi_2} \leq \|g\|_{\psi_2} + \|\delta \lfloor (g + \xi) /\delta\rfloor - (g + \xi)\|_{\psi_2} + \|\xi\|_{\psi_2} \lesssim 1 + \delta,\\
 \eta&\ts := (\bb E f(g)^2)^{1/2} = (\bb E f(g)^2)^{1/2}(\bb E g^2)^{1/2} \geq \bb E f(g)g = 1,\\
\eta&\ts \leq \sqrt 2\sup_{p\geq 1} p^{-1/2}\,(\bb E f(g)^p)^{1/p} \leq \sqrt 2 \psi \lesssim 1 + \delta,
\end{align*}
where the second line used the triangular inequality of the sub-Gaussian norm, and the third one is based on Holder's inequality. Therefore, $\psi^2/\eta \leq \psi^2 \lesssim (1+\delta)^2$ and $\eta^6/\psi^8 \lesssim 1/\psi^2 \lesssim 1$. This proves that, provided $m \gtrsim \epsilon^{-4} w^2(\cl K)$ with $0<\epsilon<\frac{c}{(1+\delta)^2}$, 
$$
\|\hat{\bs x} - \mu \|\bs x\|^{-1} \bs x\| \leq (1+\delta)^2 \epsilon,
$$ 
with probability exceeding $1 - 2\exp(-c'\epsilon^2 m)$.

Roughly speaking, by saturating the condition on $m$, one can thus estimate \whp the direction of any low-complexity vector with error $O((1+\delta)^2 \sqrt{w(\cl K)}\, m^{-1/4})$.

\end{document}